\documentclass[aip,jmp,preprint]{revtex4-1}

\usepackage{amsmath,amsfonts,amssymb,graphicx}
\usepackage{verbatim,xcolor}
\usepackage{slashed}

\newtheorem{thm}{Theorem}
\newtheorem{prop}[thm]{Proposition}

\newcommand{\Img}{\text{Im}}
\newcommand{\Rel}{\text{Re}}

\newcommand{\im}{{\rm i}}

\newcommand{\vr}{{\boldsymbol r}}
\newcommand{\vrp}{\vr_{\parallel}}
\newcommand{\vR}{{\boldsymbol R}}
\newcommand{\vs}{{\boldsymbol s}}
\newcommand{\vbq}{{\boldsymbol q}}
\newcommand{\vbk}{{\boldsymbol k}}
\newcommand{\vp}{{\boldsymbol p}}
\newcommand{\vbv}{{\bf v}}

\newcommand{\tq}{\tilde q}
\newcommand{\tomega}{\tilde\omega}

\newcommand{\dF}{\digamma}
\newcommand{\sG}{\slashed{\mathcal{G}}}

\newcommand{\Arg}{\mathrm{Arg}}

\newcommand{\dv}{{\rm d}}

\newcommand{\lo}{l_{0}}
\newcommand{\no}{n_0}

\newcommand{\Cs}{\mathsf{C}}
\newcommand{\mA}{\mathfrak A}
\newcommand{\mH}{\mathfrak H}

\DeclareMathOperator*{\Res}{\mathrm{Res}}

\begin{document}

\title{Quantum mechanical model for charge excitation: Surface binding and dispersion}

\author{Dionisios Margetis}
\email{diom@umd.edu}

\affiliation{Department of Mathematics, and Institute for Physical Science
and Technology, University of Maryland, College Park, MD 20742, USA}



\begin{abstract}
 By an idealized quantum mechanical model, we formally describe the dispersion of nonretarded electromagnetic waves that express charge density oscillations near a fixed plane in three spatial dimensions (3D) at zero temperature. Our goal is to capture the interplay of microscopic scales that include a confinement length in the emergence of the surface plasmon, a collective low-energy charge excitation in the vicinity of the plane. We start with a time-dependent Hartree-type equation in 3D. This model accounts for particle binding to the plane and the repulsive Coulomb interaction associated with the induced charge density relative to the ground state. By linearizing the equation of motion, we formulate a homogeneous integral equation for the scattering amplitude of the particle wave function in the ($z$-) coordinate vertical to the plane. For a binding potential proportional to a negative delta function and symmetric-in-$z$ wave function, we apply the Laplace transform with respect to positive $z$ and convert the integral equation into a functional equation that involves several values of the transformed solution. The scattering amplitude and dispersion relation are derived exactly in terms of rapidly convergent series via the Mittag-Leffler theorem.  In the semiclassical regime, our result furnishes an asymptotic expansion for the energy excitation spectrum. The leading-order term is found in agreement with the prediction of a classical hydrodynamic model based on a projected-Euler-Poisson system.   \color{black}
 \end{abstract}

\maketitle 

\section{Introduction}
\label{sec:Intro}

Electron systems in atomically thin materials, surfaces and interfaces are technologically appealing.\cite{Andoetal1982,Castronetoetal2009,Torres2014,Maier2007} Related studies have focused on, e.g., semiconductor heterojunctions, semiconductor-insulator interfaces as well as graphene and van der Waals heterostructures.\cite{Andoetal1982,Castronetoetal2009,Geimetal2013,Novoselovetal2012,LiBasov2008,Daietal2015} The motion of electrons may allow for the generation of $p$-polarized electromagnetic wave modes, the surface plasmons (SPs), which rapidly decay away from the material.\cite{Andoetal1982,GrimesAdams1976,Maier2007,Lowetal2017,Yaoetal2018,Bludov2013} Macroscopically, these waves express two-dimensional (2D) low-energy collective charge excitations, and inspire novel imaging and sensing devices.\cite{Jablan2013,HuAvourisetal2019} The SP dispersion law depends on the material, yet might exhibit a universal form.\cite{Andoetal1982,Bludov2013,Jablan2013} For graphene, in the mid-infrared spectrum the SP length scale can be much smaller than the free-space wavelength.\cite{Lowetal2017,Jablan2013,Catarina-Peres-rev2019}

From a semiclassical point of view, an isotropic 2D material embedded in three spatial dimensions (3D) may host a time-harmonic transverse-magnetic-polarized  SP if the tangential electric field sufficiently lags behind the induced surface current.\cite{Lowetal2017,Hanson2008,Hanson-erratumarticle,Bludov2013,DM-Luskin06,Maier2017} For a nearly lossless material, the nonretarded-SP dispersion relation has the form\cite{Andoetal1982,Bludov2013,Jablan2013}
\begin{equation}\label{eq:disp-SPP}
	\frac{\omega^2}{q}\simeq  {\rm const.}\,,
	\end{equation}
where $\omega$ is the frequency, $q=\sqrt{q_x^2+q_y^2}$ is the wave number,  and $(q_x, q_y)=\vbq$ is the in-plane wave vector. The right-hand side of~\eqref{eq:disp-SPP} is an ($\omega$- and $q$-independent) constant that depends on the equilibrium density of the free charge carriers.\cite{Andoetal1982,Ritchie1957,Ferrell1958,Fetter1973,Jonson1976,Jablan2013,LowAvouris2014,Hwang-DasSarma07,Wunsch2006} Typically,~\eqref{eq:disp-SPP} is derived from Maxwell's equations by use of a boundary condition in which the surface conductivity enters as a parameter inverse proportional to $\omega$,  by analogy with the 
Drude model.\cite{Bludov2013,Jablan2013,Drude1,Drude2,Watson2023}  The ambient medium is assumed to be isotropic and homogeneous. Relation~\eqref{eq:disp-SPP} requires that $\omega\ll q c$, where $c$ is the speed of light in the surrounding medium, and $\hbar |\omega|$ be much larger than the free-particle kinetic energy at momentum $\hbar \vbq$. The Landau damping is neglected. 

Mathematically, \eqref{eq:disp-SPP} is construed as an implication of a Neumann-to-Dirichlet-type map via classical hydrodynamics.\cite{Ritchie1957,Fetter1973} To illustrate this idea, let us suppose that the electrons form an inviscid fluid in the 
$xy$-plane, modeled by a projected-Euler-Poisson system (Sec.~\ref{sec:classical_disp}).  The linearization of this model around a steady state yields an integro-partial differential equation (IPDE) for the induced surface charge density, $\rho_1$. This IPDE incorporates the mapping of the tangential electric field to the respective normal component on the boundary via the solution of the Poisson equation. Formula~\eqref{eq:disp-SPP} ensues from setting $\rho_1= e^{\im q_x x+\im q_y y-\im \omega t}+\mathrm{c.c.}$ in this IPDE, which balances out the second-order time derivative and the 2D half Laplacian. 

In this paper, we formally derive the exact dispersion relation of a SP-type wave excitation near a plane by linearized quantum dynamics and use of an ad hoc binding potential. In addition, we show how~\eqref{eq:disp-SPP} and higher-order terms arise from this result in the semiclassical regime.  To this end, we consider nonrelativistic, spinless charged particles bound to a fixed plane at zero temperature. Let the equilibrium surface number density be $\eta_0$ where $\pi (\ell_d/2)^2= \eta_0^{-1}$. Suppose that the particles interact repulsively via the Coulomb potential; and the excitation wave due to their charge fluctuation has momentum $\hbar\vbq$ and energy $\hbar\omega(\vbq)$. We define a typical electrostatic energy as $\epsilon_C=\tfrac{e^2}{4\pi\varepsilon_0 \ell_d}$, where $e$ is the particle charge ($e>0$) and $\varepsilon_0$ is the free-space dielectric permittivity.  Let $\ell_b$ be the binding length; $\ell_{\mathrm{dB}}=\sqrt{\tfrac{\hbar}{2m_* \omega}}$ be the de Broglie wavelength where $m_*$ is the particle effective mass; $\ell_p=q^{-1}$ be the length scale of the excited wave mode; and $\ell_C$ be the length defined from $\tfrac{\epsilon_C}{\ell_d}=\tfrac{|\epsilon_b(\ell_C)|}{\ell_C}$ with $\epsilon_b(\ell)=-\tfrac{\hbar^2}{2m_*}\ell^{-2}$, which balances out electrostatic and binding forces. We expect that~\eqref{eq:disp-SPP} requires the conditions
\begin{equation}\label{eq:scales}
	\ell_b \ll \ell_{\mathrm{dB}} \ll \ell_{p}\quad \mbox{and}\quad \ell_b\ll \ell_C~.
\end{equation}

Naturally, $\ell_b$ is viewed as the smallest length scale. The condition $\ell_b\ll \ell_C$ means that the Coulomb interaction is relatively weak, with $\tfrac{\ell_b}{\ell_C}$ resembling a Wigner-Seitz-type radius; see Sec.~\ref{sec:results}.\cite{Mahan-book} However, the  parameter $\tfrac{e^2\eta_0}{\varepsilon_0}\propto \ell_C^{-3}$ has an $\mathcal O(1)$ effect on the dispersion law. The condition $\ell_{\mathrm{dB}}\ll \ell_p$ alludes to a long-wavelength regime. Formula~\eqref{eq:disp-SPP} reads as 
$\tfrac{\ell_{\mathrm{dB}}^4}{\ell_C^3\ell_p}=\mathrm{const.}$; thus,~\eqref{eq:scales} can be refined to  $\ell_b\ll \ell_C\ll \ell_{\mathrm{dB}}\ll \ell_p$. For fixed $(\omega, q)$, binding potential 
 and parameter $\tfrac{e^2\eta_0}{\varepsilon_0}$, this length scale ordering may also be suggested by the limit $\hbar\downarrow 0$, since $\ell_b=\mathcal O(\hbar^2)$, $\ell_C=\mathcal O(\hbar^{2/3})$, $\ell_{\mathrm{dB}}=\mathcal O(\hbar^{1/2})$ and $\ell_p=\mathcal O(1)$. By using these length scales, we will see that, for~\eqref{eq:disp-SPP} to hold, it suffices 
 that~\eqref{eq:scales} is supplemented with $\ell_b\ll \tfrac{\ell_{\mathrm{dB}}^2}{\ell_p}$.

Our main goal is to describe the interplay of $\ell_b$, $\ell_{\mathrm{dB}}$, $\ell_p$ and $\ell_C$ in the emergence of a SP-type dispersion relation when $\ell_b < \ell_{\mathrm{dB}} < \ell_{p}$ and $\tfrac{\ell_b}{\ell_C}$ is not too large.  The model is a Hartree-type equation for one-particle motion in 3D, which is viewed as a mean-field limit of a quantum many-body system (Secs.~\ref{subsec:model},~\ref{subsec:mf}). By linearizing this equation around the ground state, we apply scattering theory; and derive the dispersion law exactly if the binding potential is a negative delta function and the scattering amplitude is symmetric in the vertical coordinate. Moreover, we obtain~\eqref{eq:disp-SPP} by a controlled approximation of our exact result, and compute a few higher-order terms. Our approach has limitations, leaving out the crystal microstructure, Fermi-Dirac statistics, Landau damping, and ohmic losses (Sec.~\ref{subsec:limitations}).\cite{Landau1946,Orr1907,Nguyen24} Numerical simulations lie beyond our present scope.

\subsection{Model and methodology}
\label{subsec:model}
The core ingredient of our theory is the equation  
\begin{equation}\label{eq:Hartree}
	\im \hbar \, \partial_t\psi(t,\vr)=\left\{\mathcal H^0 + \upsilon\star \left(|\psi|^2-|\psi_0|^2 \right)\right\}\psi(t,\vr)~;\quad (t, \vr)\in \mathbb{R}\times \mathbb{R}^3~,\ \vr=(\vrp, z)~,
\end{equation}
with $\vrp=(x,y)$. Here, $\mathcal H^0$ is the unperturbed Hamiltonian for binding to a plane, viz.,
\begin{equation}\label{eq:H0-def}
	\mathcal H^0=-\frac{\hbar^2}{2m_*}\Delta_{\vr} + V~,
\end{equation}
and $\upsilon\star (|\psi|^2-|\psi_0|^2)$ expresses, in a mean-field classical sense, the effect of Coulomb pairwise repulsive interactions due to the induced charge density relative to $|\psi_0|^2$. We thus have
\begin{equation}\label{eq:v-Coulomb}
\upsilon(\vr)=\frac{Q_e^2}{4\pi \varepsilon_0}\,\frac{1}{|\vr|}~,	
\end{equation}
where $Q_e$ is proportional to the particle charge, $e$ (Sec.~\ref{subsec:mf}). 
Moreover, $\psi(t, \vr)=\langle \vr|\psi(t)\rangle$ is the particle wave function and $\psi_0$ is the ground state of $\mathcal H^0$.
The star ($\star$) operation between two functions stands for their convolution in the {\em whole} $\mathbb{R}^3$, and $\Delta_{\vr}$ is the Laplacian on $\mathbb{R}^3$. Equation~\eqref{eq:Hartree} can be modified to include an externally imposed charge density or electric field. A plausible connection of~\eqref{eq:Hartree} to quantum many-body dynamics is sketched in Sec.~\ref{subsec:mf}.

  The binding potential $V$ in $\mathcal H^0$ depends only on $z$, is compactly supported, and satisfies 
\begin{equation}\label{eq:V-propert}
	\int_{-\infty}^{\infty}\dv z\,V(z)=-V_0 a<0~, 
\end{equation}
where $V_0$ has units of energy and $a$ is a microscopic length ($V_0>0$). 
These assumptions on $V$ suffice for the existence of $\psi_0$ as a bound state in $z$.\cite{Simon76,Klaus77}
To fix the binding plane at $z=0$, let the support of $V(z)$ be a small neighborhood of $z=0$. Ideally, $V(z)$ can also be assumed to be non-positive, have its minimum at $z=0$, be rapidly decreasing with $|z|$, and even. To simplify the analysis, we will set $V(z)=-V_0 a\, \delta(z)$ where $\delta(z)$ is the Dirac delta function. This $V$ implies
that $\psi_0(t,\vr)$ decays exponentially with $|z|$; and the ground state energy, the lowest eigenvalue of $\mathcal H^0$, is $E_b=\hbar\omega_b=-\tfrac{\hbar^2}{2m_*}\beta^2$ where $\beta:=\ell_b^{-1}=(\tfrac{\hbar^2}{m_* V_0 a})^{-1}$.

By imposing periodicity of the wave function $\psi$ in $(x,y)$, we linearize~\eqref{eq:Hartree} around $\psi_0$; and seek solutions in terms of plane waves with energy $\hbar\omega$ and wave vector $\vbq=(q_x, q_y)$. The convolution integral of the linearized Coulomb interaction is divergent in the conventional sense and is interpreted in the sense of distributions. The area, $A$, of the periodic cell, $\mathfrak C$, in the $xy$-plane will be let to approach infinity, for fixed surface number density, $\eta_0$; thus, the allowed values of $\vbq$ will tend to form a continuum in $\mathbb{R}^2$. 

In the linearized model, the $z$-dependent (scattering) amplitude of $\psi(t,\vr)-\psi_0(t,\vr)$ obeys a homogeneous integral equation due to particle confinement. The existence of nontrivial solutions implies a relation between $\omega$ and $\vbq$. This formulation is conceptually akin to, but technically distinct from, finding the zeros of the Lindhard dielectric function for a plasma in the translation invariant setting.\cite{Lindhard1954,Hedin1970,Ford1984,PinesNozieres-book} We will focus on this integral equation for $\psi-\psi_0$. 

For $V(z)=-V_0 a\delta(z)$ and an even-in-$z$ amplitude of $\psi(t,\vr)-\psi_0(t,\vr)$, we apply the Laplace transform with respect to $z$ ($z>0$) and convert the integral equation into a functional equation. The latter involves five values of the transformed solution. By an analytic-continuation procedure, we derive a dispersion relation  of the form $\Lambda(\omega, q)=0$ where $\Lambda$ involves several convergent series that depend on $\sqrt{1+\tfrac{\ell_b^2}{\ell_p^2}\pm \tfrac{\ell_b^2}{\ell_{\mathrm{dB}}^{2}}}$ and $\tfrac{\ell_b}{\ell_C}$. In this vein, we show how~\eqref{eq:disp-SPP} and higher-order terms emerge. Our results are outlined in Sec.~\ref{sec:results}. 

Our analysis formally yields the dispersion law of a low-lying collective excitation by the study of a linearized operator in effective one-particle quantum dynamics.\cite{Mahan-book,MartinRothen-book} We note in passing that this approach bears similarities to the heuristic derivation of the phonon spectrum via a sinusoidal variation, the sound wave, of the mass density of an atomic gas.\cite{Gross61,Pitaevskii61,Wu1961,Lieb-book}  In Sec.~\ref{subsec:past}, our approach is compared to the random phase approximation (RPA). 

Our work indicates that the confined Hartree dynamics is connected to a projected-Euler-Poisson system  (Sec.~\ref{sec:classical_disp}).
It may be claimed that the former model converges, in some sense, to a description of the latter as $\hbar\downarrow 0$.\cite{GolsePaul2022,Chenetal2024} We leave this issue unresolved in this paper. 

Equations~\eqref{eq:Hartree}--\eqref{eq:v-Coulomb} form a Schr\"odinger-Poisson system. The one-particle induced charge density  $\varrho_1=e(|\psi|^2-|\psi_0|^2)$ is subject to the Schr\"odinger dynamics for $\psi$ with Hamiltonian $\mathcal H^0+Ne\Phi[\varrho_1]$. Here,  $Ne\Phi=\upsilon\star (\varrho_1/e)$ is the overall electrostatic potential, $e$ is the particle charge, and $N$ is the total number of particles. The charge density $\varrho_1$ is the forcing in the Poisson equation for $\Phi$, and spatially averages out to zero. Thus, we have $Q_e=e\sqrt{N}$  in~\eqref{eq:v-Coulomb}.

\subsection{Mean-field limit interpretation: A heuristic view}
\label{subsec:mf}

 Next, we discuss how our model can plausibly be linked to a many-particle system. In principle, the Hartree equation may be viewed as a mean-field description of a large number, $N$, of charged bosons. Electrons are thus not strictly treated by this model; see Sec.~\ref{subsec:limitations}.

First, note that~\eqref{eq:Hartree} with $V=0$ and $\psi_0= 0$, under suitable initial data for $\psi$, has rigorously been derived from the quantum dynamics of interacting charged bosons.\cite{ErdosYau2001,Bardos2002, KnowlesPickl2010,LewinSabin2015,Spohn1980} In this setting, the $N$-body Schr\"odinger state vector, $\Psi_N(t)$, is symmetric with respect to particle interchange; and obeys the Schr\"odinger equation with the $N$-body Hamiltonian
\begin{equation}
	\mathcal H_N=-\frac{\hbar^2}{2m_*}\sum_{j=1}^N \Delta_{\vr_j}+\frac{1}{N}\sum_{j,l=1\atop j< l}^N\upsilon(\vr_j-\vr_l)\qquad (\vr_j\in\mathbb{R}^3)~,
\end{equation}
where $\upsilon$ is given by~\eqref{eq:v-Coulomb} and $\vr_j$ is the position of particle $j$.
The mean-field limit is roughly interpreted as follows: if the many-body state vector $\Psi_N(t)$ is a tensor product of one-particle states initially, i.e., $\Psi_N(0)=\otimes^N\varphi$ for a normalized state $\varphi$, for large $N$ the system is described by $\Psi_N(t)=e^{-\im \tfrac{H_N}{\hbar} t}\Psi_N(0)\simeq \otimes^N \psi(t)$ for $t>0$ where $\psi(t,\vr)=\langle \vr |\psi(t)\rangle$ satisfies~\eqref{eq:Hartree} with $\psi(0)=\varphi$, $V=0$ and $\psi_0= 0$. The one-particle state $\psi(t)$ represents the macroscopic state. Thus, in~\eqref{eq:v-Coulomb} we can set $Q_e^2=N e^2$. This interpretation is physically appealing but obscures the sense in which the many-body dynamics converge to the mean-field limit.\cite{ErdosYau2001} 

We heuristically extend this mean-field notion to a setting with confinement near the plane $\{z=0\}$. The first step is to add the binding potential $\sum_{j=1}^NV(z_j)$ to $\mathcal H_N$, which modifies the Hartree equation by addition of $V$ to the unperturbed Hamiltonian $\mathcal H^0$; cf.~\eqref{eq:H0-def}. 

The ensuing Hartree dynamics imply~\eqref{eq:Hartree} with interaction potential $\upsilon\star |\psi|^2$. By suppressing the time ($t$-) dependence, we write the underlying Hartree energy functional as
\begin{equation*}
	E_H[\chi]=\int_{\mathcal{C}\times\mathbb{R}} \dv\vr\,\dv\vr'\left\{ \chi^*(\vr)\epsilon_V(\vr,\vr')\chi (\vr')+\tfrac{1}{2}|\chi(\vr)|^2 \upsilon(\vr-\vr')|\chi(\vr')|^2\right\}~;
\end{equation*}
$\epsilon_V(\vr,\vr')=\{-\tfrac{\hbar^2}{2m_*}\Delta_{\vr}+V(z)\}\delta(\vr-\vr')$ and $\chi$ replaces $\psi$. 
We have been unable to express the (global) minimizer of this $E_H[\chi]$ in simple closed form, if $V(z)$ is a negative delta function.

We formally incorporate a simplifying mechanism into the mean-field model. Let the Coulomb potential, $\upsilon$, be felt by the particle only through the induced charge density $e (|\psi(t,\vr)|^2-|\psi_0(t,\vr)|^2)$.  The Hartree energy functional is thus modified to become 
\begin{equation*}
	E_H^m[\chi]=\int \dv\vr\,\dv\vr'\left\{ \chi^*(\vr)\epsilon_V(\vr,\vr')\chi(\vr')+\tfrac{1}{2}\left(|\chi(\vr)|^2-|\psi_0(\vr)|^2\right)\upsilon(\vr-\vr')\left(|\chi(\vr')|^2-|\psi_0(\vr')|^2\right)\right\}.
\end{equation*}
If $\psi_0$ is the ground state of $\mathcal H^0$, $\chi(\vr)=\psi_0(\vr)$ is the minimizer of $E_H^m[\chi]$, since the Fourier transform of $\upsilon$ (in 3D) is positive and, thus, the integral of the interaction term in $E_H^m$ is nonnegative. This modified Hartree energy, under a constraint on the
$L^2$-norm of the solution $\chi$, implies the stationary version of~\eqref{eq:Hartree} for $\chi=\psi$ through a variational principle. 

\subsection{Limitations and open problems}
\label{subsec:limitations}
The use of Hartree-type equation~\eqref{eq:Hartree} with unperturbed Hamiltonian~\eqref{eq:H0-def} and interaction potential~\eqref{eq:v-Coulomb} raises a few concerns. In a nutshell, in the limit of zero binding ($V\to 0$), this model describes a free atomic gas. The Pauli exclusion principle, which is physically compelling for electrons, is not taken into account here. Periodic or quasi-periodic potentials for crystal microstructures are not included. Ohmic losses are left out.

In more detail, if $V(z)=-2\beta\delta(z)$, the limit $\beta\downarrow 0$ ($\ell_b\to +\infty$) with fixed $\ell_C$ and $\ell_p$ implies that the linearized Coulomb interaction potential tends to vanish because of the wave function normalization. The excitation energy  $\hbar\omega(q)$ becomes $\tfrac{\hbar^2 q^2}{2m_*}$ plus a term that tends to scale linearly with $\beta$, eventually furnishing the dispersion relation for free-particle dynamics. We deem the zero-binding limit as not physically meaningful in this model, thus considering our approach useful for moderate or small enough $\tfrac{q}{\beta}=\tfrac{\ell_b}{\ell_p}$.  Using length scales, we assume that $\sqrt{\ell_b\ell_p} < \ell_{\mathrm{dB}} < \ell_p$ while $\tfrac{\ell_b}{\ell_C}$ is not too large ($\ell_b=\beta^{-1}$); see Sec.~\ref{sec:results}.            
 
For electrons, the many-body state vector  is antisymmetric under particle interchange. By the Hartree-Fock theory,\cite{Mahan-book} this state is nearly a Slater determinant of $N$ orthogonal one-electron states. The particle dynamics may be described by $N$ coupled PDEs for the electronic wave functions that include the classical Coulomb interaction and exchange terms.\cite{Bardos2003,Elgart2004,Frohlich2011,LewinSabin2015,Benedikter2014} By neglect of such terms, $N$ coupled Hartree equations formally ensue.\cite{Elgart2004,LewinSabin2015} We anticipate that this description yields scaling law~\eqref{eq:disp-SPP} in the semiclasssical regime. Treatments of the plasmon excitation via the density operator~\cite{Nguyen24} or the Wigner function~\cite{Mendonca2023} with Fermi-Dirac statistics account for the Landau damping,\cite{Nguyen24} but lie beyond our present scope.

Our model does not express screened Coulomb interactions, since the effect of ions is left out. This situation is different from the premise of the jellium model.\cite{Mahan-book,MartinRothen-book,Mendonca2023} We ignore the atomistic structure of 2D materials, e.g., the honeycomb lattice of graphene\cite{Kotovetal2012} in which the slow electron motion obeys Dirac dynamics.\cite{FeffermanWeinstein2012,Hwang-DasSarma07,Wunsch2006} In the semiclassical regime, the Dirac dynamics can yield a different scaling of $\tfrac{\omega^2}{q}$ with the charge density of electrons, in comparison to our approach. Thus, our theory cannot predict all key features of the graphene SP.\cite{Wunsch2006,Hwang-DasSarma07,Falkovsky2007a}  

A way to account for ohmic losses in an atomically thin material would be to use the density operator with a Poisson random process for collision events.\cite{Watson2023,Cances2017,Schulz-Baldes1998} This formulation may result in a finite relaxation time for the optical conductivity.~\cite{Schulz-Baldes1998} The study of this effect under spatial confinement, when collisions are spatially nonhomogeneous, deserves attention.

\subsection{Related past works}
\label{subsec:past}

Our goal is to derive a dispersion relation for charge density oscillations under a particular 1D binding potential. Let us now place our work in the appropriate context of past literature. 

A method of a similar purpose is the RPA.\cite{Mahan-book,MartinRothen-book,PinesNozieres-book,Hedin1970,EhrenreichCohen1959} Roughly, the idea is to approximately reduce the Heisenberg equation of motion for the number density operator on Fock space to an equation for the {\em average} number density, on some convenient (e.g., momentum) basis of a one-particle Hilbert space. This average is defined with respect to the ground state of a reference Hamiltonian.\cite{Mahan-book,MartinRothen-book} The scheme employs the approximate factorization of averages of operators that are quadratic in the number density.\cite{MartinRothen-book} The RPA typically invokes the polarizability and dielectric function of the electron gas, which are computed self consistently.\cite{MartinRothen-book,PinesNozieres-book,Hedin1970,Lindhard1954} In translation invariant settings, the plasmon energy excitation spectrum comes from zeros of the dielectric function in the Fourier space.\cite{Mahan-book,Hwang-DasSarma07,Wunsch2006,Stern1967} The RPA  is a linearized version of the time-dependent Hartree theory,\cite{Hedin1970,EhrenreichCohen1959} This notion is inherent to our model, but we leave out the Pauli exclusion principle and directly use  the wave function, $\psi$. 

The RPA has been extended to settings with spatial confinement such as metallic nano-films through external potentials or boundary conditions.\cite{Adler1962,Dahl1977,Andersen2012,Andersen2013,VegaAbajo2017,Echarri2019,Thorn2012} These extensions usually invoke representations of the polarizability on basis sets that diagonalize a suitably chosen (unperturbed) one-particle Hamiltonian. A similar approach is adopted for the study of edge states.\cite{Christensen2014} Here, we avoid the explicit use of the polarizability and dielectric function per se; instead, we solve an integral  equation for $\psi-\psi_0$ within linearized  Hartree-type dynamics. We end up deriving the even $z$-dependent scattering amplitude as a series expansion in non-orthogonal decaying exponentials of $|z|$. This expansion comes from applying the Laplace transform in $z$ to an integral equation, and may not have been speculated from the unperturbed Hamiltonian. The SP-type dispersion relation follows directly in this vein.

Another past approach is to obtain the optical conductivity by electronic-structure means and use it as a parameter in boundary conditions for Maxwell's equations.\cite{Falkovsky2007a,Falkovsky2007b,KaxirasMosallaei2014,Bludov2013,Catarina-Peres-rev2019} Here, we couple the Poisson equation, the quasi-electrostatic limit of Maxwell's equations, with effective one-particle quantum dynamics in 3D. Thus, we account for confined electron transport in the vertical ($z$-) coordinate without explicitly imposing any boundary condition at the binding ($xy$-) plane. A price paid for our exact solution is its limited applicability due to the Hartree dynamics and the idealized form of our binding potential, $V$.

\subsection{Notation and terminology}
\label{subsec:notation}
Boldface symbols denote vectors or matrices, e.g., $\boldsymbol e_s$ ($s= x, y, z$) is a unit Cartesian vector. $\mathbb{C}$, $\mathbb{R}$ ($\mathbb{R}_+)$ and $\mathbb{Z}$ are the sets of complex, real (positive real) numbers and integers, respectively, while $\mathbb{Z}_A=\frac{2\pi}{\sqrt{A}}\mathbb{Z}$. We write $f=\mathcal O(h)$ ($f=o(h)$) if $|f/h|$ is bounded (tends to zero) in a prescribed limit; and $f\sim h$ if $f-h=o(h)$.  The hat on top of a symbol indicates the Fourier transform with respect to time $t$ and coordinates $(x,y)$. The star ($*$) as a superscript indicates the complex conjugate or Hermitian adjoint. We often express linear integral operators by their kernels. The term ``macroscopic limit'' means taking $A\to \infty$ with fixed $\tfrac{N}{A}$. We distinguish the SP-type wave afforded by our model from the graphene SP.\cite{Hwang-DasSarma07,Wunsch2006}  We define the semiclassical regime by $\sqrt{\ell_b\ell_p}\ll \ell_{\mathrm{dB}}\ll \ell_p$ with $\ell_{b}\ll \ell_C$. (Strong binding means $\ell_b\ll \ell_{\mathrm{dB}}\lesssim \ell_p$ with $\ell_b\lesssim \ell_C$). We set $\hbar=1=2m_*$ unless we state otherwise.

\subsection{Paper outline}
\label{subsec:outline}

In Sec.~\ref{sec:results}, we summarize our results. Section~\ref{sec:classical_disp} presents the derivation of a nonretarded-SP dispersion law by a hydrodynamic model. In Sec.~\ref{sec:scat_form}, we formulate the quantum mechanical problem, and derive a nonlinear integral equation for the wave function. Section~\ref{sec:disp_qm} focuses on the linearized integral equation and definition of the SP-type excitation wave. In Sec.~\ref{sec:derivation-disp}, we derive the dispersion relation for even-in-$z$ amplitudes. In Sec.~\ref{sec:asymptotics}, we derive~\eqref{eq:disp-SPP} and higher-order terms in the semiclassical regime.  In Sec.~\ref{sec:conclusion}, we conclude the paper. 

\section{Main results}
\label{sec:results}
In this section, we outline our results for the quantum mechanical setting. A canonical model of classical hydrodynamics for the nonretarded SP is discussed in Sec.~\ref{sec:classical_disp}.

Consider the complex-valued periodic-in-$\vrp$ function $h_A(t,\vrp)$ with the cell $\mathfrak C=[0, \sqrt{A}]^2\subset \mathbb{R}^2$.  The spacetime Fourier transform of $h_A(t, \vrp)$ is 
\begin{align*}
	\hat{h}_A(w, \vbk_\parallel)=\int_{\mathbb{R}}\int_{\mathfrak C} h_A(t,\vrp)\,e^{\im w t-\im \vbk_\parallel\cdot \vrp}\,\dv \vrp\,\dv t~,\quad w\in\mathbb{R}~,\ \vbk_\parallel\in \mathbb{Z}_A^2~.
\end{align*}
Suppose that $\hat h_A(w,\vbk)$ is analytic in $w$ for $\Img\, w>0$. The Fourier inversion formally yields
\begin{align}\label{eq:FT-inv-discr}
	h_A(t, \vrp)=\frac{1}{A}\sum_{\vbk_\parallel\in \mathbb{Z}_A^2}\int_{\Gamma}e^{-\im w t+\im \vbk_\parallel\cdot \vrp}\,\hat{h}_A(w,\vbk_\parallel)\,\frac{\dv w}{2\pi}~,
\end{align}
where the path $\Gamma$ lies in the upper $w$-plane (if $\Img\,w>0$) and is parallel to the real axis. In the macroscopic limit, as $h_A$ ($\hat{h}_A$) approaches $h$ ($\hat{h}$) in an appropriate sense, \eqref{eq:FT-inv-discr} becomes
\begin{align}\label{eq:FT-inv}
	h(t, \vrp)=\int_{\mathbb{R}^2}\int_{\Gamma}e^{-\im w t+\im \vbk_\parallel\cdot \vrp}\,\hat{h}(w,\vbk_\parallel)\,\frac{\dv w}{2\pi}\,\frac{\dv \vbk_\parallel}{(2\pi)^2}~.
\end{align}

We focus on PDE~\eqref{eq:Hartree} with unperturbed Hamiltonian~\eqref{eq:H0-def} under $\mathfrak C$-periodicity in $(x ,y)$. The forward spacetime propagator, $G_A(t,\vrp,z; z')$, of this problem obeys 
\begin{equation}\label{eq:propagator-PDE}
\left(\im\partial_t-\mathcal H^0\right)G_A(t,\vrp, z; z')=-\delta(\vrp) \delta(z-z') \delta(t)~,\qquad \vrp\in \mathfrak C~,\ (z,z')\in\mathbb{R}^2~,
\end{equation}
and $G_A(t,\vrp,z; z')\equiv 0$ if $t<0$ while $G_A(t,\vrp,z; z')$ is bounded in $(\vrp, z-z')$ for fixed $t>0$.   

We employ $V(z)=-V_0a\delta(z)$. The ground state wave function is $\psi_0(t, \vr)=\sqrt{\tfrac{\beta}{A}} e^{-\im \omega_b t} f_0(z)$ with $f_0(z)=e^{-\beta |z|}$, $\omega_b=-\beta^2$ and $\beta=\frac{1}{2}V_0a$ (Sec.~\ref{subsec:unpert-prob}). Let $\vbk=\vbk_\parallel$ for ease of notation.

The following propositions (Propositions~\ref{prop:FT-G}--\ref{prop:strong-b}) and comments summarize our main results.  

\begin{prop}\label{prop:FT-G}
	If $V(z)=-V_oa\delta(z)$, the Fourier transform of $G_A(t,\vrp,z; z')$  is
\begin{subequations}
\begin{equation}\label{eq:FT-G}
\widehat{G}_A(w,\vbk,z;z')=\frac{1}{2\alpha} \left\{
 e^{-\alpha |z-z'|}+\beta (\alpha-\beta)^{-1} e^{-\alpha(|z|+|z'|)}\right\}~,
\end{equation}
for all $(z, z')\in\mathbb{R}^2$ and $(w, \vbk)\in \mathfrak R_A:= \{ (w, \vbk)\in \mathbb{C}\times \mathbb{Z}_A^2\,\big|\,\alpha \neq \beta~,\,\Rel\,\alpha> 0\}$ where
\begin{equation}
\alpha=\alpha(w, \vbk)=\sqrt{k^2-w}~; \quad \beta=\ell_b^{-1}=\frac{1}{2}V_0 a~.	
\end{equation}
\end{subequations}
\end{prop}

Proposition~\ref{prop:FT-G} is proved in Sec.~\ref{subsec:causal-prop}. For fixed $k=|\vbk|$, the top Riemann sheet for $\alpha(w,\vbk)$ in the $w$-plane is defined by  $\Rel\,\alpha> 0$, and the branch cut is the half line $[k^2, +\infty)$ in the real axis. In the macroscopic limit, a closed-form expression for the respective propagator, $G$, is given in Sec.~\ref{subsec:causal-prop}. $\widehat G_A(w,\cdot)$ is ill defined for $w$ in the point spectrum of $\mathcal H^0$. These points $w$ are: (i) the simple pole $w=w_p$ of $\widehat{G}_A$, by $\alpha(w, \vbk)=\beta$; and (ii) the branch point $w=k^2$ of $\widehat{G}_A$ and the half line $(k^2, +\infty)$ along the real axis of the $w$-plane. 
Let $\alpha=\sqrt{|k^2-w|}e^{\im \theta_w/2}$ with $\theta_w=\Arg(k^2-w)$ and $|\theta_w|<\pi$; thus, $\alpha(w, \vbk)^*=\alpha(w^*, \vbk)$. By use of scattering theory, we formulate a nonlinear integral equation for the wave function $\psi$ (Sec.~\ref{subsec:nonlin-int_eq}). 

We linearize this integral equation around $\psi_0$; see Sec.~\ref{subsec:plasmon-def}. Let $u_s(t,\vr)$ be a solution of the linearized equation, which corresponds to $\psi(t,\vr)-\psi_0(t,\vr)$. In the absence of any forcing, the (homogeneous) linearized integral equation reads (Sec.~\ref{subsec:plasmon-def})
\begin{align}\label{eq:scat_FE}
u_s(t,\vr)+\int_{\mathfrak{C}\times\mathbb{R}}\int_{\mathbb{R}}G_A(t-t',\vrp-\vrp',z;z') \{\upsilon\star  (\psi_0 u_s^*+\psi_0^* u_s)\psi_0(t',\vr')\}\,\dv t'\,\dv\vr'=0~.	
\end{align}
For suitable values $(\omega, \vbq)\in \mathbb{C}\times\mathbb{Z}_A^2$, we can seek  nontrivial solutions to~\eqref{eq:scat_FE} of the form 
\begin{equation}\label{eq:plane-ansatz}
u_s(t,\vr)=\frac{1}{\sqrt{2A}}\left\{f_+(z)\,e^{\im \vbq\cdot \vrp-\im (\omega_b+\omega) t}+f_-(z)^*\,e^{-\im \vbq\cdot \vrp-\im (\omega_b-\omega) t}\right\}~,\quad f_{\pm}\in L^2(\mathbb{R})~.	
\end{equation}
The physics-motivated SP-type excitation wave is defined for $\omega\in\mathbb{R}$ via~\eqref{eq:plane-ansatz} and the combination $(e^{\im\omega_b t}u_s(t,\vr)+\mathrm{c.c.})/\sqrt{2}$. The scattering amplitude, $F(z)$, is introduced by
\begin{equation}\label{eq:plane-ansatz-SP}
u^{SP}(t,\vr)=\frac{1}{\sqrt{2A}}\left\{F(z)e^{\im \vbq\cdot \vrp-\im \omega t}+\mathrm{c.c.}\right\}~;\quad F(z):=\{f_+(z)+f_-(z)\}/\sqrt{2}~.	
\end{equation}
See Definition~1 (Sec.~\ref{subsec:plasmon-def}).    

\begin{prop}\label{prop:lin-int_eq}
In the macroscopic limit, \eqref{eq:scat_FE}--\eqref{eq:plane-ansatz-SP} entail the linear integral equation	
\begin{subequations}\label{eqs:F-int_eq}
\begin{equation}\label{eq:int_eq-F}
	F(z)=-\frac{\beta}{2q} \frac{e^2 \eta_0}{\varepsilon_0}\int_{-\infty}^{\infty}\dv z'\, f_0(z')\,\widehat{\slashed{G}}(\omega, \vbq, z; z')\,\mathcal D_q(f_0 F)(z')\,,\  z\in\mathbb{R}~,
\end{equation}
where $F\in L^2(\mathbb{R})$, $\eta_0=N/A$ is the surface number density of charged particles, $f_0(z)=e^{-\beta|z|}$, $\beta=\tfrac{1}{2}V_0a$, $\vbq\in\mathbb{R}^2\setminus\{0\}$, and $\omega$ must be determined for $\omega\in\mathbb{C}$ and $|\Arg(q^2+|\omega_b|\pm \omega)|<\pi$. The linear operator $\mathcal D_q$ ($\mathcal D_q: L^2(\mathbb{R}) \rightarrow L^2(\mathbb{R})$) is defined by
\begin{equation}\label{eq:D-op}
\mathcal D_q(u)(z):=\int_{-\infty}^{\infty} \dv z'\, e^{-q|z-z'|}\,u(z')\qquad (q=|\vbq|>0)~.	
\end{equation}
The function $\widehat{\slashed{G}}(\omega, \vbq, z; z')$ is 
\begin{equation}\label{eq:kernel-def}
	\widehat{\slashed{G}}(\omega, \vbq, z; z'):=\widehat{G}\big(\omega_b+\omega, \vbq, z; z'\big)+\widehat{G}\big(\omega_b-\omega, \vbq, z; z'\big)~,\quad (z,z')\in\mathbb{R}^2~,
\end{equation}
\end{subequations}
where $\widehat{G}$ is given by~\eqref{eq:FT-G} with domain $\mathfrak R= \{ (w, \vbk)\in \mathbb{C}\times \mathbb{R}^2\,\big|\,\alpha \neq \beta~,\,\Rel\,\alpha> 0\}$.
\end{prop}

Proposition~\ref{prop:lin-int_eq} is proved in Sec.~\ref{subsec:int_eq-z}. Equation~\eqref{eq:int_eq-F} eventually reads $\mathcal L(\omega, \vbq)(f_0 F)=0$, where $\mathcal L$ is a linearized integral operator with a kernel described  in Sec.~\ref{subsec:int_eq-z}.  
In~\eqref{eq:kernel-def}, $\widehat{\slashed{G}}(\omega, \vbq, \cdot)$ comes from the Fourier transform in $(t, \vrp)$  of $e^{\im \omega_b t}G(t, \vrp, \cdot)+(e^{\im \omega_b t}G(t, \vrp, \cdot))^*$ with $\widehat{G}\big(\omega^*, \cdot)^*=\widehat{G}\big(\omega, \cdot)$; and depends on  $\vbq$ only through $|\vbq|$ by in-plane rotational invariance. 

For an even scattering amplitude $F(z)$, we reduce~\eqref{eq:int_eq-F} to a nonconvolution Volterra-type equation for $f_0F$, which is converted into a five-point functional equation for the Laplace transform of $f_0F$ with respect to positive $z$ (Sec.~\ref{subsec:funct-eq}). We apply an analytic-continuation procedure in the Laplace domain to obtain a convergent series expansion for $F$ (Sec~\ref{subsec:disp-soln}). Odd functions $F(z)$ can be treated similarly but are not studied in this paper. 
\begin{prop}\label{prop:disp_reln}
Consider even amplitudes $F(z)$, and $\omega\in \mathbb{C}$ with $0<q<|\omega|/q< \beta$ and  $|\Arg(q^2+\beta^2\pm \omega)|<\pi$. Equation~\eqref{eq:int_eq-F} leads to the dispersion relation 
\begin{align}\label{eq:Lambda-disp-reln}	
\Lambda(\omega,q)=\mathrm{det}(\boldsymbol{\mathfrak{A}}(\omega, q)-\boldsymbol{I})=0~,
\end{align}
where $\boldsymbol{\mathfrak{A}}(\omega, q)=[\mathfrak A^\mu_\nu(\omega, q)]$ is a matrix valued function, the indices $\mu$ and $\nu$ take values in the set $\{+, -, ++, -+, +-, -- \}$, and $\boldsymbol{I}$ is the $6\times 6$ identity matrix. The matrix elements $\mA^\mu_\nu$ are given by the following uniformly convergent series (for $\varsigma,\, \sigma,\, \varsigma',\, \sigma'=\pm$): 
\begin{subequations}\label{eqs:prop3-matrix}    
\begin{align}
\mA^{\varsigma}_{\varsigma'}&=\frac{4\Cs_0(\beta)}{4\tq^2-\tomega^2}\sum_{n=0}^\infty \left\{\delta_{\varsigma',+}\sum_{\sigma'=\pm} \Cs_0^{\sigma'}\frac{\Lambda_n^{\sigma'}(\beta)}{\varsigma \tilde q+\tilde\alpha_{\sigma'}+2n+1}+ \frac{\slashed{\Lambda}_n(\beta)}{(\varsigma+1) \tilde q+2(n+1)}  \right\},\label{eq:A-varsigma-varsigmap}\\
\mA^{\varsigma\sigma}_{\varsigma'}&=\frac{4\Cs_0(\beta)}{4\tq^2-\tomega^2}
\sum_{n=0}^\infty\left\{\delta_{\varsigma', +} \sum_{\sigma'=\pm} \Cs_0^{\sigma'} \frac{\Lambda_n^{\sigma'}(\beta)}{\varsigma\tilde \alpha_\sigma+\tilde\alpha_{\sigma'}+2(n+1)}+ \frac{\slashed{\Lambda}_n(\beta)}{\varsigma\tilde \alpha_\sigma+\tilde q+2n+3}\right\}, \label{eq:A-varsigma-sigma-varsigmap}\\
\mA^{\varsigma}_{\varsigma'\sigma'}&=\frac{4\Cs_0(\beta)}{4\tq^2-\tomega^2} \sum_{n=0}^\infty \frac{\Cs_{\varsigma'}^{\sigma'}\Lambda_n^{\sigma'}(\beta)}{\varsigma \tilde q+\tilde\alpha_{\sigma'}+2n+1},\label{eq:A-varsigma-sigma-varsigmasigma}\\     
\mA^{\varsigma\sigma}_{\varsigma'\sigma' }&=\frac{4\Cs_0(\beta)}{4\tq^2-\tomega^2}\sum_{n=0}^\infty \frac{\Cs_{\varsigma'}^{\sigma'}\Lambda_n^{\sigma'}(\beta)}{\varsigma\tilde \alpha_\sigma+\tilde\alpha_{\sigma'}+2(n+1)}.\label{eq:A-varsigma-sigma-varsigmap-sigmap}
\end{align}
\end{subequations}
We define $\tilde {\vbq}:=\tfrac{\vbq}{\beta}$ ($\tilde q=\tfrac{q}{\beta}$), $\tilde\alpha_\sigma:=\tfrac{\alpha(\omega_b+\sigma\omega, \vbq)}{\beta}=\alpha(-1+\sigma\tilde\omega, \tilde{\vbq})$, $\tilde\omega:=\frac{\omega}{\beta^2}$, and $\delta_{\varsigma, +}=1$ if $\varsigma=+$ while $\delta_{\varsigma, +}=0$ otherwise (with $\Rel\,\tilde\alpha_\sigma>0$); and introduce the dimensionless parameters
\begin{subequations}\label{eqs:prop3-pmts}
\begin{align}
 \Cs_0(\beta)&=\frac{e^2\eta_0}{2\varepsilon_0 \beta^3}=\biggl(\frac{2\ell_b}{\ell_C}\biggr)^3~,\quad \Cs_0^\sigma=-\tilde q \,\frac{\tilde\alpha_\sigma+1}{\tilde q^2-\sigma\tilde\omega}~,\\
 \Cs^\sigma_+&=-\frac{(\tilde\alpha_\sigma+1)^2}{\tilde q^2-\sigma\tilde\omega}\frac{(\tilde\alpha_\sigma-1)^2-\tilde q^2}{4\tilde\alpha_\sigma}~,\ \Cs^\sigma_-=- \frac{(\tilde\alpha_\sigma+1)^2-\tilde q^2}{4\tilde\alpha_\sigma}~.
\end{align}
\end{subequations}
The sequences $\{\Lambda_n^\sigma(\beta)\}_{n\in \mathbb{N}}$ and $\{\slashed{\Lambda}_n(\beta)\}_{n\in\mathbb{N}}$ are defined so that $\Lambda_0^\sigma=1=\slashed{\Lambda}_0$ (if $n=0$) and
\begin{subequations}\label{eqs:residues-n-scaled}
\begin{align}
	\Lambda_n^\sigma(\beta)
	&= \frac{(-\Cs_0(\beta))^n}{n!}\prod_{j=1}^{n}\left\{\frac{1}{j+\tilde\alpha_\sigma}\left[1+\frac{\sigma \tilde \omega}{4j(j+\tilde\alpha_\sigma)-2\sigma\tilde\omega} \right]\right.\notag\\
	& \left.\hphantom{\frac{1}{n!}\biggl(-\frac{e^2 \eta_0}{2\varepsilon_0} \biggr)^n}\times \frac{1}{1+(2j-1)^2+2(2j-1)\tilde\alpha_\sigma-\sigma\tilde \omega}\right\}~,\quad n\ge 1~,\\
	\slashed{\Lambda}_n(\beta)&=\frac{(-\Cs_0(\beta))^n}{n!} \prod_{j=1}^n
	\frac{1}{j+\tilde q}\,\frac{j(j+1)+\left(j+\tfrac{1}{2}\right)\tilde q}{4[j(j+1)+(j+\tfrac{1}{2})\tilde q]^2-\tilde \omega^2/4}~,\quad n\ge 1~.
\end{align}
\end{subequations}
\end{prop}

Proposition~\ref{prop:disp_reln} is proved in Sec.~\ref{sssec:proof-Prop3}. In our proof, we invoke an expansion of the Laplace transform of $f_0F$ via the Mittag-Leffler theorem; see Lemma~1 (Sec.~\ref{sssec:exp-ML}). Notice that 
 $\Lambda(\omega, q)=\Lambda(-\omega, q)$, because the replacement of $\omega$ by $-\omega$ amounts to the interchange of the index ($\mu$ or $\nu$) values $++$ and $+-$ with $-+$ and $--$, respectively, in the matrix $[\mathfrak A^\mu_\nu(\omega, q)]$. The ensuing even amplitude $F(z)$ is described in Sec.~\ref{sssec:explct-F}. 
 
The imposed condition $q<\tfrac{|\omega|}{q}<\beta$ simplifies the steps towards the proof of Proposition~\ref{prop:disp_reln} and is compatible with the semiclassical regime. This condition is equivalent to $\sqrt{\ell_b\ell_p}<\ell_{\mathrm{dB}}<\ell_p$ which implies  $\ell_b< \ell_{\mathrm{dB}}< \ell_p$. More generally, we expect that, for fixed $q$ with $0< q<\beta$, the determinant $\Lambda(\omega, q)$ of~\eqref{eq:Lambda-disp-reln} is analytic in complex $\omega$ with 
 $\Rel\,\alpha(\omega_b\pm \omega, \vbq)>0$ and $q<\tfrac{|\omega|}{q}<  2(\beta+\beta_0)$ for some $\beta_0>0$. Despite the factor $(4\tq^2-\tomega^2)^{-1}$ that appears in~\eqref{eqs:prop3-matrix}, the points $\omega=\pm 2q\beta$ are regular for $\Lambda(\omega, q)$ if $0< q<\beta$, as discussed in Sec.~\ref{subsec:soln-rmks}. In contrast, $\omega=\pm q^2$ are simple poles of $\Lambda(\omega, q)$, which are inherited from $\widehat{G}(\omega_b\pm \omega, \vbq, z; z')$. The  free-particle energy $q^2$ should be a lower bound of the macroscopic-like energy $|\omega(q)|$ near the semiclassical regime; and the restriction $q/\beta <1$ ($q\ell_b< 1$) is deemed consistent with the neglect of the Landau damping. The phase velocity $\tfrac{|\omega|}{q}$ should plausibly be bounded above by a constant close to $2\beta=V_0 a$, the minimal velocity for particle escape from confinement.   

For complex $\omega$, the power series entering the matrix elements $\mathfrak A^\mu_\nu$ are found to converge uniformly with $\tilde\omega$, $\tilde q$ and $\Cs_0$, if $\tilde q^2<|\tilde \omega|<\tilde q$ and $\Cs_0\le \mathcal O(1)$, and become rapidly convergent if $\Cs_0\ll 1$. It is tempting, although not physically compelling, to interpret $\Cs_0$ as an effective Wigner-Seitz-type radius,  by some loose analogy with the jellium model:\cite{Mahan-book,MartinRothen-book} Set $\Cs_0=\tfrac{e^2}{4\pi\varepsilon_0 r_0}/\tfrac{\hbar^2}{m_* r_0^2}=\tfrac{r_0}{a_0}$ where $r_0=\eta_0\,4\pi \ell_b^3=2\tfrac{(2\ell_b)^3}{\ell_d^2}$ and $a_0=\tfrac{\hbar^2 4\pi\varepsilon_0}{e^2 m_*}$ is the Bohr radius. 

Equation~\eqref{eq:Lambda-disp-reln} is simplified considerably in the semiclassical regime; cf.~\eqref{eq:scales}.  
\begin{prop}\label{prop:strong-b}
If $0<\tilde q^2\ll |\tilde\omega|\ll \tilde q$ and $\Cs_0\ll 1$, dispersion relation~\eqref{eq:Lambda-disp-reln} is reduced to
\begin{align}\label{eq:leading-order}
	\frac{\omega^2}{q}=\frac{e^2 \eta_0}{\varepsilon_0} \left\{1+\mathcal O\biggl(\frac{\tilde q^4}{\tilde \omega^2}\biggr)+\mathcal O(\tilde q)\right\}~.
\end{align}
\end{prop}

The proof of Proposition~\ref{prop:strong-b} invokes the truncation of all series expansions for $\boldsymbol{\mathfrak{A}}$; see Sec.~\ref{sec:asymptotics}. In the semiclassical limit, \eqref{eq:leading-order} reduces to~\eqref{eq:disp-SPP} which in turn reads  $(2\ell_{\mathrm{dB}})^4=\ell_C^3\ell_p$; cf.~\eqref{eq:classical-SP-disp2} (Sec.~\ref{sec:classical_disp}). Note that, by our approach, the constant on the right-hand side of~\eqref{eq:disp-SPP} is linear with the surface charge density $\eta_0$, in contrast to the case of the graphene SP.\cite{Jablan2013,Wunsch2006,Hwang-DasSarma07}

\section{Classical SP by Projected-Euler-Poisson system}
\label{sec:classical_disp} 
This section revisits a simplified hydrodynamic model that yields a classical nonretarded-SP dispersion law.\cite{Ritchie1957,Fetter1973} We show how~\eqref{eq:disp-SPP} is predicted if electrons are assumed to form a 2D inviscid fluid in the macroscopic limit. By linearization of the governing equations, we derive a dispersive IPDE for the induced surface number density. This model does not account for a variety of hydrodynamic effects in 2D materials, but captures a key mapping that implies the scaling law of~\eqref{eq:disp-SPP}. 
For reviews of electronic hydrodynamics, see Refs.~\onlinecite{Gurzhi-1968,Ritchie1957,Fetter1973,Giuliani2005quantum-book,Sun2018universal,LucasFong2018}.

Consider the geometry of Fig.~\ref{fig:geometry}. The fluid has surface number density $\eta(t, \vrp)$ and velocity field $\vbv(t, \vrp)$ in the $xy$-plane; $\boldsymbol e_z\cdot \vbv=0$. We assume that the vector valued field $(\eta(t,\vrp), \vbv(t,\vrp))$ is smooth in its domain, $\mathbb{R}\times \mathbb{R}^2$. The mass conservation statement reads
\begin{subequations}\label{eqs:Euler-Poisson}
\begin{equation}\label{eq:local_charge}
\partial_t \eta(t,\vrp) + \nabla_\parallel\cdot \big(\eta(t,\vrp)\, \vbv(t ,\vrp)\big)=0~,	\qquad \vrp\in \mathbb{R}^2~,
\end{equation}
where $\nabla_\parallel=(\partial_x, \partial_y)$. The deviation of the surface charge density, $e\eta(t, \vrp)$, from a given constant density $e \eta_0$ is the source in the Poisson equation for the scalar potential $\Phi$, viz.,
\begin{equation}\label{eq:Poisson}
	\varepsilon_0\{(-\Delta_{\vr})\Phi(t,\vr)\}=e\{\eta(t, \vrp)-\eta_0\}\delta(z)~,\qquad \vr=(\vrp, z)\in \mathbb{R}^3~.
\end{equation}
The tangential electric field, $-e \nabla_\parallel \Phi$, provides the forcing in the Euler momentum equation,
\begin{equation}\label{eq:Euler}
	\partial_t \vbv+(\vbv\cdot \nabla_\parallel)\vbv=-\frac{e}{m_*}(\nabla_\parallel\Phi)\big|_{z=0}~.
\end{equation}
\end{subequations}
We refer to~\eqref{eqs:Euler-Poisson} as the projected-Euler-Poisson system.\cite{Tadmor-priv} 

Let us linearize~\eqref{eqs:Euler-Poisson} around the steady state $(\eta, \vbv)=(\eta_0, 0)$. The Euler equations entail
\begin{subequations}\label{eqs:Euler-Poisson-lin}
\begin{equation}\label{eq:local_charge-lin}
\partial_t \eta_1(t, \vrp) + \eta_0\nabla_\parallel\cdot\vbv_1(t, \vrp)=0~,	
\end{equation}
\begin{equation}\label{eq:Euler-lin}
	\partial_t \vbv_1(t, \vrp)=-\frac{e}{m_*}(\nabla_\parallel\Phi(t,\vrp))\big|_{z=0}~,\quad \vrp\in \mathbb{R}^2~,
\end{equation}
where $(\eta_1, \vbv_1)$ corresponds to the solution $(\eta-\eta_0, \vbv)$ of the nonlinear system. We assume that $(\eta_1, \vbv_1)\to 0$ sufficiently fast as $|\vrp|\to \infty$ (for fixed $t$). Solving~\eqref{eq:Poisson}, we assert that 
\begin{equation}\label{eq:Poisson-soln-fracL}
\big(\nabla_\parallel \Phi(t,\vrp)\big)\big|_{z=0}= \frac{e}{2\varepsilon_0}\nabla_\parallel \{(-\Delta_\parallel)^{-1/2}\eta_1(t, \vrp)\}~,
\end{equation}
\end{subequations}
where the fractional Laplacian $(-\Delta_\parallel)^{-1/2}$ on $\mathbb{R}^2$ is defined by~\cite{Samko-book}
\begin{equation*}
	(-\Delta_\parallel)^{-1/2}h(\vrp)=\frac{1}{2\pi}\int_{\mathbb{R}^2}\dv \vrp'\ \frac{1}{|\vrp-\vrp'|}\,h(\vrp')~.
\end{equation*}

\begin{figure}
  \centering
\includegraphics[scale=.45,trim=0.5in 1.5in 0in 1.4in]{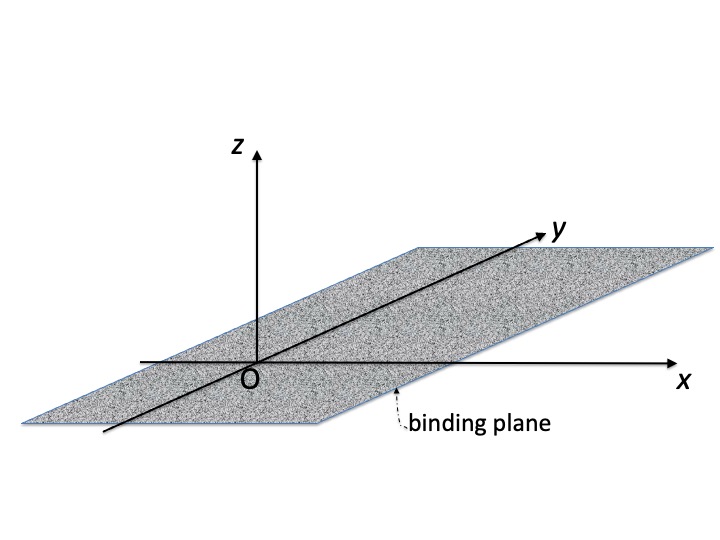}
  \caption{%
    Geometry of the problem. Charged particles move near the $xy$-plane in a homogeneous and isotropic unbounded medium. In the classical case (Sec.~\ref{sec:classical_disp}), the particle system is treated as an inviscid fluid in the $xy$-plane. In the quantum mechanical setting (Secs.~\ref{sec:scat_form}--\ref{sec:derivation-disp}), the particle motion is wavelike under the external potential $V(z)=-V_0a \,\delta(z)$ by mean-field Hartree-type dynamics.}
  \label{fig:geometry}
\end{figure}

We now derive a linear dispersive IPDE for the time evolution of $\eta_1$. By differentiating both sides of~\eqref{eq:local_charge-lin} with respect to time, interchanging the order of $\partial_t$ and $\nabla_\parallel$, we obtain 
\begin{equation}\label{eq:lin-disp_PDE}
	\partial_t^2\eta_1(t,\vrp)+\frac{e^2 \eta_0}{2m_*\varepsilon_0}\{(-\Delta_\parallel)^{1/2}\eta_1(t, \vrp)\}=0~,\qquad (t,\vrp)\in \mathbb{R}\times\mathbb{R}^2~,
\end{equation}
where $(-\Delta_\parallel)^{1/2}=(-\Delta_\parallel)(-\Delta_\parallel)^{-1/2}$. By linearity, the same IPDE holds for the induced surface charge density, $\varrho_1=e\eta_1$. Equation~\eqref{eq:lin-disp_PDE} incorporates a relation between the normal and tangential components of the electric field in the 2D fluid via the solution of the Poisson equation. This relation expresses a Neumann-to-Dirichlet-type map in the $xy$-plane.

By formally taking the Fourier transform of~\eqref{eq:lin-disp_PDE} with respect to $(t, \vrp)$, we find
\begin{subequations}\label{eqs:classical-SP-disp}
\begin{equation}\label{eq:classical-SP-disp1}
\left(-\omega^2 +\frac{e^2 \eta_0}{2m_*\varepsilon_0} |\vbq|\right)\hat\eta_1(\omega, \vbq)=0\qquad ((\omega, \vbq)\in \mathbb{R}^3)~. 
\end{equation}
This equation implies~\eqref{eq:disp-SPP} if the Fourier transform of $\eta_1$ is nonzero, $\hat\eta_1(\omega, \vbq)\neq 0$, viz.,\cite{Ferrell1958,Fetter1973,Andoetal1982}  
\begin{equation}\label{eq:classical-SP-disp2}
	\frac{\omega^2}{q}=\frac{e^2 \eta_0}{2m_*\varepsilon_0}~.
\end{equation}
\end{subequations}
Notably, the constant on the right-hand side of this equation is linear with the surface number density $\eta_0$. This scaling with $\eta_0$ is modified for the graphene SP (see Sec.~\ref{sec:conclusion}).\cite{Jablan2013,Hwang-DasSarma07,Wunsch2006}

The geometry of this classical setting could be altered so that the inviscid fluid is confined within a layer of nonzero width in the vicinity of the $xy$-plane. Thus, the spatial domain of the Euler equations would be a suitable subset of $\mathbb{R}^3$. The confinement could be implemented by use of a forcing in the momentum equation, or, alternatively, by imposition of boundary conditions for the mass flux and density. This extension is not further discussed here. 


\section{Quantum mechanical problem: Formulation}
\label{sec:scat_form} 

In this section, we focus on the one-particle mean-field quantum dynamics. We convert evolution law~\eqref{eq:Hartree} into a nonlinear integral equation, when the binding potential $V(z)$ is a negative delta function. For this purpose, we describe the ground state, $\psi_0$, of the unperturbed Hamiltonian, $\mathcal H^0$, and compute the Fourier transform of the forward spacetime propagator, $G_A(t, \vrp, z; z')$.  

\subsection{Ground state of unperturbed Hamiltonian $\mathcal H^0$}
\label{subsec:unpert-prob}

Consider the geometry of Fig.~\ref{fig:geometry} and the unperturbed Hamiltonian $\mathcal H^0$ given by~\eqref{eq:H0-def}. The Schr\"odinger equation reads
\begin{equation*}
	\mathcal H^0\psi(t,\vr)=\im \,\partial_t \psi(t,\vr)~,\qquad (t, \vr)=(t, \vrp, z)\in \mathbb{R}\times \mathfrak C\times \mathbb{R}~.
\end{equation*}
For stationary particle motion, we seek solutions of the form
\begin{equation*}
\psi(t,\vr)=	A^{-1/2} e^{-\im  w t}e^{\im \vbk\cdot \vrp}f(z)~,\qquad \vbk \in \mathbb{Z}_A^2~,
\end{equation*}
where $w$ is the particle energy and $\Vert f\Vert_{L^2(\mathbb{R})}=1$.

For the ground state $\psi=\psi_0$ of $\mathcal H^0$, we set $\vbk=0$. Thus, the  pair $(f, w)=:(f_b, \omega_b)$ solves
\begin{equation*}
	\{-\partial_z^2+V(z)-\omega_b\}f_b(z)=0\quad (\Vert f_b\Vert_{L^2(\mathbb{R})}=1)~,
\end{equation*}
where $\omega_b$ is the lowest eigenvalue of $\mathcal H^0_z:=-\partial_z^2+V$ with $\omega_b<0$. For a compactly supported $V(z)$ subject to~\eqref{eq:V-propert}, there is at least one bound state of $\mathcal H^0_z$.\cite{Simon76,Klaus77} To allow for an analytically tractable model, we use a negative binding potential $V(z)$ that is localized at one point, 
\begin{equation*}
	V(z)=-V_0a \delta(z)\qquad (V_0 a>0)~.
\end{equation*}
Hence, only one ($z$-dependent) bound state of $\mathcal H^0_z$ exists. By the continuity of $f_b(z)$ at $z=0$ and the jump condition $[\partial_z f_b(z)]_{z=0}:=(\partial_z f_b)|_{z=0^+}-(\partial_z f_b)|_{z=0^-}=-V_0a f_b(0)$, we compute 
\begin{equation*}
	f_b(z)= \sqrt{\beta} e^{-\beta |z|}\quad (z\in\mathbb{R})~;\quad  \omega_b=-\beta^2~,\quad \beta=\frac{V_0 a}{2}>0~.
\end{equation*}
Thus, we have
\begin{equation}\label{eq:bs}
	\psi_0(t,\vr)=\sqrt{\frac{\beta}{A}}\,e^{-\im \omega_b t} 
	e^{-\beta |z|}~,\quad (t, \vr)=(t, \vrp, z)\in\mathbb{R}\times\mathbb{\mathfrak C}\times \mathbb{R}\quad (\Vert \psi_0\Vert_{L^2(\mathfrak{C}\times \mathbb{R})}=1)~.
\end{equation}

 It is of some interest to consider other binding potentials, $V(z)$, that are localized at one point, in the spirit of Ref.~\onlinecite{Wu2002}. We note in passing that negative square-well potentials may allow for more than one bound states, and lie beyond the scope of our paper.

\subsection{Forward propagator $G_A$ and its Fourier transform}
\label{subsec:causal-prop}

Consider PDE~\eqref{eq:propagator-PDE}. The propagator, $G_A(t, \vrp, z; z')$, is $\mathfrak C$-periodic  in $\vrp$, vanishes if $t<0$ and is bounded in $(\vrp, z-z')$ for fixed $t>0$. By~\eqref{eq:FT-inv-discr}, we apply the Fourier representation 
\begin{equation}\label{eq:G-inv-periodic}
G_A(t,\vrp, z; z')=\frac{1}{A}\sum_{\vbk\in \mathbb{Z}_A^2}\int_\Gamma\frac{\dv w}{2\pi}\ e^{-\im w t+\im \vbk\cdot \vrp}\,\widehat G_A(w, \vbk, z; z')\qquad (t\in\mathbb{R})~. 
\end{equation}
The integration path is the line $\Gamma:=\{w \in\mathbb{C}\,\big|\, -\infty < \Rel\,w < +\infty,\, \Img\,w=\gamma_1>0\}$, where $\gamma_1$ is a small positive number which may eventually approach zero;\cite{PaleyWiener} see Fig.~\ref{fig:integration}. The function $\widehat{G}_A$  is holomorphic in $w$ if $\Img\,w>0$ for fixed $(\vbk, z, z')$, as suggested by the vanishing of $G_A$ for all $t<0$; and should be bounded as $|z-z'|\to \infty$ for $t>0$. By~\eqref{eq:propagator-PDE}, $\widehat G_A$ must satisfy
\begin{equation}\label{eq:FT-G-ODE}
	\{\partial_z^2-\alpha^2-V(z)\}\widehat G_A(w, \vbk, z; z')=-\delta(z-z'),\quad (z, z')\in \mathbb{R}^2~.
\end{equation}

A few comments on properties of $\widehat{G}_A$ are in order. Recall that $\alpha=\sqrt{k^2-w}$, which has a branch point at $w=k^2$. The top Riemann sheet for $\alpha$ as a function of $w$ is defined by $\Rel\,\alpha>0$ with the cut $\{w\in\mathbb{C} : \Rel\,w\ge k^2,\, \Img \,w=0\}$ for fixed $k$ (Fig.~\ref{fig:integration}). By inspection of~\eqref{eq:FT-G-ODE}, we see that $\widehat{G}_A$ satisfies $\widehat{G}_A(\cdot, z; z')=\widehat{G}_A(\cdot, z'; z)$ for all $(z, z')\in \mathbb{R}^2$; in addition, $\widehat{G}_A(\cdot, z; z')=\widehat{G}_A(\cdot, -z; -z')$. Thus, it suffices  to determine $\widehat{G}_A$ only for $0\le  |z|\le z'$. 

\begin{figure}
  \centering
\includegraphics[scale=.55,trim=1in 2.1in 0in 0in]{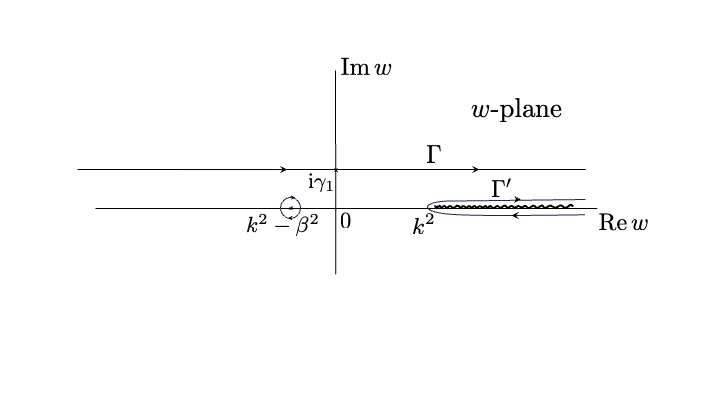}
  \caption{%
    Schematic of branch-cut configuration and contour integration in the $w$-plane in regard  to~\eqref{eq:G-inv-periodic}. The integration path $\Gamma$ is a line parallel to and at distance $\gamma_1$ from the real axis in the upper half $w$-plane. The wavy curve emanating from $w=k^2$ indicates the branch cut for  $\alpha(w)=\sqrt{k^2-w}$; and $w=w_p=k^2-\beta^2$ is the simple pole of $\widehat{G}_A$. The path $\Gamma$ is deformed to path $\Gamma'$ which is wrapped around the cut, picking up the residue at $w_p$ through a small circle around $w_p$. This change of path is used for the computation of $G_A$ in the macroscopic limit (Appendix~\ref{app:sec:propagator}).}
  \label{fig:integration}
\end{figure}

\subsubsection{Proof of Proposition~\ref{prop:FT-G}}
\label{sssec:proof-Prop1}

In view of differential equation~\eqref{eq:FT-G-ODE},  we keep $z'$ fixed, with $z'> 0$, and write
\begin{align*}
	\widehat G_A(w,\vbk, z;z')=C^+(z') \cosh(\alpha (z-z'))+C_1(z')\sinh(\alpha(z-z'))~,\quad 0< z\le z'~,
\end{align*}
where the integration constants $C^+$ and $C_1$ must be determined. Evidently, $\widehat G_A$ as a function of $z$ ($z> 0$) must be continuous at $z=z'$, and have a discontinuous first derivative across this point with a jump equal to $[\partial_z \widehat G_A]_{z=z'}=-1$. Let us set $\widehat G_A=C^+(z') e^{-\alpha (z-z')}$ for $0< z' \le  z$, with $\Rel\,\alpha>0$. By using $C^+(z')$, we find 
\begin{align*}
	C_1(z')=\alpha^{-1}-C^+(z')~,
\end{align*}
which in turn implies the formula
\begin{subequations}\label{eqs:FT-G-forms}
\begin{align}\label{eq:FT-G-form1}
	\widehat G_A(w,\vbk, z;z')=C^+(z') e^{-\alpha (z-z')}+\alpha^{-1}\sinh(\alpha(z-z'))~,\quad 0< z\le z'~.
\end{align}
On the other hand, we write
\begin{align}\label{eq:FT-G-form2}
	\widehat G_A(w, \vbk, z; z')=C^-(z') e^{\alpha z}~,\quad z\le  0<  z'~.
\end{align}
\end{subequations}

The task at hand is to find the integration constants $C^{\pm}(z')$. In view of the binding potential, $V(z)=-2\beta \delta(z)$, the function $\widehat G_A$ must be continuous with respect to $z$ at $z=0$ and subject to the jump discontinuity $[\partial_z\widehat G_A]_{z=0}=-2\beta \widehat G_A|_{z=0}$. These two conditions entail 
\begin{align*}
	C^+ e^{\alpha z'}-C^-&=\alpha^{-1}\sinh(\alpha z')~,\\
	\alpha C^+ e^{\alpha z'}-(2\beta-\alpha)C^-&= \cosh(\alpha z')~.
\end{align*}
The solution of this linear system for $(C^+, C^-)$ is
\begin{align*}
	(C^+(z'), C^-(z'))=\left(\frac{1}{2}\left(\alpha-\beta\right)^{-1}\left\{1-2\alpha^{-1}\beta e^{-\alpha z'}\sinh(\alpha z')\right\}, \frac{1}{2}(\alpha-\beta)^{-1}e^{-\alpha z'}\right)~,
\end{align*}
provided $\alpha\neq \beta$. By substitution of the above formulas for $C^\pm(z')$ into~\eqref{eqs:FT-G-forms}, we derive 
\begin{align*}
	\widehat G_A(w,\vbk, z; z')=\frac{1}{2} (\alpha-\beta)^{-1} e^{\alpha(z-z')}~, \quad z\le 0\le z'~,
\end{align*}
\begin{align*}
	\widehat G_A(w,\vbk, z; z')=\frac{1}{2\alpha}
	\left\{e^{\alpha(z-z')}+\beta (\alpha-\beta)^{-1} e^{-\alpha(z+z')} \right\}~, \quad 0\le z\le z'~.
\end{align*}
The two last formulas are combined into the expression
\begin{align*}
	\widehat G_A(w,\vbk, z; z')=\frac{1}{2\alpha}
	\left\{e^{-\alpha (z'-z)}+\frac{\beta}{\alpha-\beta} e^{-\alpha(z'+|z|)} \right\}\quad \mathrm{for\ all}\ z\le z'~,\ z'\ge 0~.
\end{align*}

By the symmetry of $\widehat G_A$ as a function of $(z,z')$, we can extend this formula to all $(z, z')\in \mathbb{R}^2$ and obtain~\eqref{eq:FT-G}. In view of the conditions $\alpha\neq \beta$ and $\Rel\,\alpha>0$, the second one of which ensures that $\widehat G_A$ is bounded as $|z-z'|\to \infty$, the Fourier domain of $\widehat G_A$ is  
\begin{align*}
	\mathfrak R_A=\{(w,\vbk)\in \mathbb{C}\times \mathbb{Z}_A^2\,
	\big|\,\alpha\neq \beta~,\ \Rel\,\alpha> 0\}~.\qquad \qquad \hskip1.5in\Box 
\end{align*}

In the macroscopic limit, the corresponding propagator, $G(t, \vrp, z; z')$, is found explicitly by the inversion of  formula~\eqref{eq:FT-G} according to~\eqref{eq:FT-inv}. The calculation is described in Appendix~\ref{app:sec:propagator}, with recourse to the modified path $\Gamma'$ of Fig.~\ref{fig:integration}. We mention in passing that the result is
\begin{subequations}\label{eqs:G-TD}
\begin{align}\label{eq:G-TD-total}
G(t,\vrp,z;z')&=G_{f}(t,\vrp,z-z')\left\{1 + e^{\im\pi/4}\beta \sqrt{\pi t}\, e^{-\im \tfrac{|z-z'|^2}{4t}}e^{\im \beta^2 t-\beta(|z|+|z'|)}\right. \notag\\
& \qquad \left. \times \mathrm{erfc}\biggl( \frac{e^{-\im\pi/4}}{2\sqrt{t}}(|z|+|z'|)-e^{\im\pi/4}\beta \sqrt{t}\biggr)\right\}~,\quad t>0~,
\end{align}
and $G(t,\vrp,z;z')\equiv 0$ if $t<0$, where $G_{f}(t, \vr; z')$ denotes the 3D free-particle propagator 
\begin{align}\label{eq:G-TD-fs}
	G_{f}(t,\vr)=e^{-\im\pi/4}(4\pi t)^{-3/2} e^{\im \tfrac{r^2}{4t}}~,\quad t> 0~.
\end{align}
\end{subequations}
In the above, $\mathrm{erfc}(\cdot)$ is the complementary error function, which is defined for all $\xi\in\mathbb{C}$ by
\begin{align*}
	\mathrm{erfc}(\xi)=\frac{2}{\sqrt{\pi}}\int_{\xi}^{+\infty} e^{-s^2}\,\dv s~.
\end{align*}

\subsection{Nonlinear integral equation}
\label{subsec:nonlin-int_eq}
Next, we apply basic scattering theory to formally derive a nonlinear integral equation in the place of the governing Hartree-type equation~\eqref{eq:Hartree}. This integral equation will express the forward-in-time propagation of initial data for $\psi-\psi_0$ by reference to the corresponding time-dependent state that evolves under the unperturbed Hamiltonian, $\mathcal H^0$.\cite{Nguyen24} 

Let $\psi(t, \vr)$  satisfy~\eqref{eq:Hartree} under some initial data that is $\mathfrak{C}$-periodic in $(x,y)$, setting $\psi(t_0, \vr)=\pi_0(\vr)$ for fixed $t_0\in\mathbb{R}$; $\pi_0,\,\psi(t, \cdot) \in L^2(\mathfrak{C}\times \mathbb{R})$. The scattering wave function is
\begin{align*}
	\psi_s(t, \vr):= \psi(t, \vr)-\psi_0(t, \vr)~,
\end{align*}
where $\psi_0$ is given by~\eqref{eq:bs}. This $\psi_s(t,\vr)$ satisfies the initial-value problem 
\begin{subequations}\label{eqs:scatt-wavfcn}
\begin{align}
&(\im\partial_t-\mathcal H^0)\psi_s(t, \vr)=\{\upsilon\star(|\psi_0+\psi_s|^2-|\psi_0|^2)(t,\vr)\}(\psi_0+\psi_s)(t,\vr)~,\ \, t>t_0~; \label{eq:scatt-wavfcn-PDE}\\
 &\psi_s=\pi_0-\psi_0\ \mbox{at}\ t=t_0~.\label{eq:scatt-wavfcn-IC}
\end{align}
\end{subequations}

Now consider the non-causal backward propagator $G^{-}_A(t, \cdot):=G_A(-t, \cdot)^*$, which obeys 
\begin{align}\label{eq:Gstar-t-rev}
(\im\partial_t-\mathcal H^0)G^-_A(t, \vrp, z; z')=-\delta(\vrp)\delta(z-z')\,\delta(t)~;\quad G^-_A(t, \cdot)\equiv 0,\quad  t>0~.
\end{align}
Let us carry out the following contractions: Equation~\eqref{eq:scatt-wavfcn-PDE} with $G^-_A(t-t', \vrp-\vrp', z; z')^*$, and the complex conjugate of~\eqref{eq:Gstar-t-rev} with $\psi_s(t,\vr)$. We subtract the resulting equations, and invoke the Hermiticity of $\mathcal H^0$ as well as the  $\mathfrak{C}$-periodicity of related functions in $\vrp$ and the symmetry property $G_A(\cdot, \vrp, \cdot)=G_A(\cdot, -\vrp, \cdot)$. Hence, we obtain 
\begin{align}\label{eq:forw-time-integral}
	&\psi_s(t,\vr)=-\im \int_{\mathfrak{C}\times\mathbb{R}}G_A(t-t_0, \vrp-\vrp', z; z')\psi_s(t_0, \vr')\,\dv\vr'\notag\\
	&-\int_{t_0}^\infty \int_{\mathfrak{C}\times\mathbb{R}}G_A(t-t', \vrp-\vrp', z; z')\{\upsilon\star(|\psi_0+\psi_s|^2-|\psi_0|^2)\}(\psi_0+\psi_s)(t',\vr')\,\dv\vr'\,\dv t'~,
\end{align}
where $(t, \vr)=(t, \vrp, z)\in (t_0, \infty)\times\mathfrak{C}\times \mathbb{R}$. The integral of the convolution ($\star$) operation inside the braces ($\{\,\}$) in the second line of~\eqref{eq:forw-time-integral} is evaluated at $(t', \vr')$. 

We compare the time evolution of $\psi_s$ expressed by~\eqref{eq:forw-time-integral} to the corresponding evolution brought about by the unperturbed Hamiltonian, $\mathcal H^0$, under the same initial data (at time $t_0$). Hence, we consider the unperturbed evolution problem
\begin{align*}
\im\partial_t \psi^0(t, \vr)=\mathcal H^0 \psi^0(t, \vr)~,\ t>t_0~;\quad \psi^0=\pi_0\ \mbox{at}\ t=t_0~.
\end{align*}
Without further ado, by defining $\psi_s^0(t,\vr):=\psi^0(t,\vr)-\psi_0(t,\vr)$ we have
\begin{align}\label{eq:psi-s-zero}
	\psi_s^0(t,\vr)= -\im \int_{\mathfrak{C}\times\mathbb{R}} G_A(t-t_0, \vrp-\vrp', z; z') \psi_s^0(t_0, \vr')\,\dv\vr'~,\quad (t, \vr)\in (t_0, \infty)\times\mathfrak{C}\times \mathbb{R}~.
\end{align}
Since $\psi_s=\psi_s^0$ at time $t=t_0$, the above equation for $\psi_s^0(t,\vr)$ is combined with~\eqref{eq:forw-time-integral} to furnish
\begin{align}\label{eq:nonl-int-eq-t0}
&\psi_s(t,\vr)+\int_{-\infty}^\infty \int_{\mathfrak{C}\times\mathbb{R}}G_A(t-t',\vrp-\vrp', z;z')\{\upsilon\star(|\psi_0+\psi_s|^2-|\psi_0|^2)\}(\psi_0+\psi_s)(t',\vr')\,\dv\vr'\,\dv t'\notag\\
&\qquad =\psi_s^0(t,\vr)~,	\quad (t, \vr)\in (t_0, \infty)\times \mathfrak{C}\times\mathbb{R}~,
\end{align}
which is the desired integral equation for $\psi_s(t,\vr)$. The right-hand side, $\psi^0_s(t,\vr)$, is viewed as a forcing in this equation. In~\eqref{eq:nonl-int-eq-t0}, we have applied $\psi(t, \vr)\equiv 0$ for $t< t_0$. 

In the remainder of this paper, we will study the linearized version of~\eqref{eq:nonl-int-eq-t0} in the macroscopic limit, replacing $|\psi_0+\psi_s|^2$ by $|\psi_0|^2+2\Rel(\psi_0\psi_s^*)$ (Sec.~\ref{subsec:plasmon-def}). We expect that the application of the Fourier transform  in $(t, \vrp)$ to the linearized integral equation  yields an equation of the form
\begin{equation}\label{eq:op-scattering}
	\hat{u}_{fs}(w, \vbq, z)+\mathcal G_1(w,\vbq)(\hat{u}_{fs}(w, \vbq,\cdot)+\hat{u}_{fs}(2\omega_b-w^*, -\vbq,\cdot)^*)(z)=g_f(w, \vbq, z)~.
\end{equation}
Here, $\hat{u}_{fs}(w, \vbq, z)=f_0(z) \hat{u}_s(w, \vbq, z)$, $f_0(z)=e^{-\beta|z|}$, $u_s(t,\vrp, z)$ takes the  place of the scattering solution $\psi_s(t,\vr)$ of the nonlinear problem, $g_f$ is the forcing that corresponds to $f_0\psi_s^0$, and $(w, \vbq)$ is the variable dual to $(t, \vrp)$   (Sec.~\ref{sec:disp_qm}); cf.~Ref.~\onlinecite{Nguyen24}. Moreover, $\mathcal G_1$ ($\mathcal G_1(w,\vbk): L^2(\mathbb{R})\rightarrow L^2(\mathbb{R})$) is a linear integral  operator that comes from the transformed forward propagator, $\widehat{G}$, acts on $z$-dependent scattering amplitudes and accounts for particle confinement; see Sec.~\ref{subsec:int_eq-z}.  
The replacements $w=\omega_b+ \omega$ and $w=\omega_b-\omega^*$ in~\eqref{eq:op-scattering} suggest the transformation $\hat{u}_{fs}\mapsto \hat{F}_s$ with $\hat{F}_s(\omega, \vbq, z)=\hat{u}_{fs}(\omega_b+\omega, \vbq; z)+\hat{u}_{fs}(\omega_b-\omega^*, -\vbq; z)^*$. The resolvent equation calls for defining the operator    
\begin{equation}\label{eq:Lindhard-op}
	 \mathcal L(\omega, \vbq)= 1+\mathcal G_1(\omega_b+\omega, \vbq)+\mathcal G_1(\omega_b-\omega, \vbq)~,
\end{equation}
by virtue of the identity $\mathcal G_1(\omega_b-\omega^*, -\vbq)^*=\mathcal G_1(\omega_b-\omega, -\vbq)=\mathcal G_1(\omega_b-\omega, \vbq)$ for the Hermitian adjoint, $\mathcal G_1(\cdot)^*$, of operator $\mathcal G_1(\cdot)$. In this setting, if  $g_f\equiv 0$ in~\eqref{eq:op-scattering}, vectors in the null space of $\mathcal L(\omega, \vbq)$ in principle represent  SP-type wave  excitations.  In Sec.~\ref{sec:derivation-disp}, we compute such nontrivial vectors by solving directly the linearized version of~\eqref{eq:nonl-int-eq-t0}. Their existence  implies a relation of the form $\Lambda(\omega, q)=0$, the SP-type dispersion relation (Sec.~\ref{sec:derivation-disp}).

\section{Linearized integral equation}
\label{sec:disp_qm}
In this section, we linearize integral equation~\eqref{eq:nonl-int-eq-t0} and define the SP-type wave (Sec.~\ref{subsec:plasmon-def}). We also formulate an integral equation for the $z$-dependent scattering amplitude (Sec.~\ref{subsec:int_eq-z}). 

\subsection{Definition of SP-type wave}
\label{subsec:plasmon-def}
Consider~\eqref{eq:nonl-int-eq-t0} as $t_0\to -\infty$, together with the identity $|\psi_0+\psi_s|^2=|\psi_0|^2+2\Rel(\psi_s\psi_0^*)+|\psi_s|^2$.
By the neglect of $|\psi_s|^2$ and replacement of $\psi_s$ by $u_s$, for $(t, \vr)\in \mathbb{R}\times\mathfrak{C}\times\mathbb{R}$ we obtain  
\begin{align}\label{eq:lin-int-eq}
u_s(t,\vr)+\int\limits_{\mathfrak{C}\times\mathbb{R}} \int\limits_{\mathbb{R}}G_A(t-t',\vr;\vr')\{\upsilon\star(u_s^*\psi_0+u_s\psi_0^*)\}\psi_0(t',\vr')\,\dv t'\,\dv \vr'=0~,	
\end{align}
where $\psi_s^0(t,\vr)$ is taken to vanish in view of~\eqref{eq:psi-s-zero}.
Since $e^{\im \omega_b t}\psi_0(t, \vr)$ is real, we write~\eqref{eq:lin-int-eq} as 
\begin{align}\label{eq:lin-int-eq-transf-hom}
\slashed{u}_s(t,\vr)+\frac{\beta}{A}\int\limits_{\mathfrak{C}\times\mathbb{R}} \int\limits_{\mathbb{R}}\slashed{G}_A(t-t',\vr;\vr')\,\{\upsilon\star(\slashed{u}_s f_0)(t',\vr')\}f_0(z')\,\dv t'\,\dv \vr'=0~, 
\end{align}
where $(t, \vr)=(t,\vrp,z)\in \mathbb{R}\times\mathfrak{C}\times\mathbb{R}$. Here, we define 
\begin{align}
&\slashed{U}(t,\cdot):=e^{\im\omega_b t}U(t,\cdot)+\mbox{c.c.}\quad  (U=u_s, G_A)~.\label{eq:slashed-def}
\end{align}

An option is to introduce the SP-type wave as a nontrivial solution $u_s$ of~\eqref{eq:lin-int-eq}; see also~\eqref{eq:scat_FE}. Physically, this choice suggests that this wave accounts for emission and absorption of radiation with respect to the ground state binding energy, $\omega_b$. Thus, we can write $u_s$ as a linear superposition of waves with energies $\omega_b\pm \omega$. Alternatively, we can define the SP-type excitation wave as a nontrivial solution $\slashed{u}_s$ of~\eqref{eq:lin-int-eq-transf-hom}. 
Note that the linearized operator $\mathcal L(\omega, \vbq)$ of~\eqref{eq:Lindhard-op} comes directly from the left-hand side of~\eqref{eq:lin-int-eq-transf-hom}. This $\mathcal L(\omega, \vbq)$ is described  in Sec.~\ref{subsec:int_eq-z}.

 We choose to define the SP-type excitation via plane-wave solutions $\slashed{u}_s$. By virtue of~\eqref{eq:slashed-def}, the excitation energy $\omega$ is introduced relative to the unperturbed binding energy $\omega_b$. 
 
 \medskip 
 
\noindent{\bf Definition 1} {\em The SP-type excited wave is any nontrivial solution to~\eqref{eq:lin-int-eq-transf-hom} of the form~\eqref{eq:plane-ansatz-SP} for suitable $\omega\in\mathbb{R}$ and $F\in L^2(\mathbb{R})$, and every $\vbq \in \mathbb{Z}_A^2$.} 

Ansatz~\eqref{eq:plane-ansatz-SP} is a particular real-valued function $\slashed{u}_s(t ,\vr)$, according to the transformation $u_s\mapsto \slashed{u}_s$ of~\eqref{eq:slashed-def}. Since $\slashed{G}_A$ and $f_0$ are real valued in~\eqref{eq:lin-int-eq-transf-hom}, we can equivalently seek nontrivial plane-wave solutions of the form
\begin{align}\label{eq:plane-ansatz-SP-cmpx}
	\slashed{u}_s(t, \vr)=u^{SP}(t,\vr)= A^{-1/2}F(z) e^{-\im\omega t+\im \vbq \cdot \vrp}~,
\end{align}
where $(\omega_b\pm \omega, \vbq)\in \mathfrak R_A$ in view of Proposition~\ref{prop:FT-G} ($F\in L^2(\mathbb{R})$). Ansatz~\eqref{eq:plane-ansatz-SP-cmpx} in principle implies a relation between $\omega$ and $\vbq$. Recall that $\Rel\,\alpha(\omega_b\pm \omega, \vbq)>0$ and $\vbq\in \mathbb{Z}_A^2$.
  In the macroscopic limit, for nontrivial even scattering amplitude $F(z)$ we will explicitly derive the SP-type dispersion relation (Sec.~\ref{sec:derivation-disp}). Physically, for low enough $q$ ($q< \beta$), we expect (but have not proved) that the dispersion relation has only real solutions $\omega(q)$. 

\subsection{Integral equation for scattering amplitude in macroscopic limit}
\label{subsec:int_eq-z}
Next, starting from~\eqref{eq:lin-int-eq-transf-hom} with~\eqref{eq:plane-ansatz-SP-cmpx}, we derive~\eqref{eqs:F-int_eq} as $A\to +\infty$ with $\tfrac{N}{A}= \eta_0=\mathcal O(1)$, thus proving Proposition~\ref{prop:lin-int_eq}. Furthermore, we show that~\eqref{eq:int_eq-F} is reduced to a Volterra-type integral equation for the function $f_0(z) F(z)$. This description reveals the linearized operator $\mathcal L(\omega, \vbq)$ which we loosely introduced in~\eqref{eq:Lindhard-op}. We also briefly discuss the possibilities of even and odd solutions for the scattering amplitude $F(z)$.

\subsubsection{Proof of Proposition~\ref{prop:lin-int_eq}}
\label{sssec:proof-Prop2}

We replace $\slashed{u}_s$ by ansatz~\eqref{eq:plane-ansatz-SP-cmpx} in integral equation~\eqref{eq:lin-int-eq-transf-hom}. After a formal manipulation of the requisite convolution integral, we obtain
\begin{align*}
	\upsilon\star (f_0\slashed{u}_s)(t,\vr)=\frac{1}{\sqrt{A}}\frac{e^2N}{\varepsilon_0}e^{-\im\omega t+i\vbq\cdot\vrp}\int_{\mathbb{R}}\dv z'\,\check{G}_e(\vbq, z-z')\,\{f_0(z')F(z')\}
\end{align*}
where
\begin{align*}
	\check{G}_e(\vbq, z)=\frac{1}{4\pi}\int_{\mathbb{R}^2}\dv\vR\ \frac{e^{-\im\vbq\cdot \vR}}{\sqrt{z^2+R^2}}~,\quad \vbq\in \mathbb{Z}_A^2\setminus \{0\}~.
\end{align*}
The last integral expresses the Fourier transform (at $\vbq\neq 0$) with respect to $\vrp$ of the fundamental solution, $G_e(\vrp, z)$, to the Laplace equation in $\mathbb{R}^3$. This $G_e(\vr)$ satisfies $-\Delta_{\vr}G_e(\vr)= \delta(\vrp) \delta(z)$, with $\vr=(\vrp, z)\in\mathbb{R}^3$. The Fourier transform in $\vrp$ of this PDE furnishes the differential equation $(-\partial_z^2+q^2) \check{G}_e(\vbq, z)=\delta(z)$ ($z\in\mathbb{R}$). The transformed function, $\check{G}_e(\vbq, z)$, must be continuous in $z$ and have a derivative with the jump discontinuity 
$[\partial_z \check{G}_e]_{z=0}=-1$ at $z=0$, while $\check{G}_e(\vbq,z)\to 0$ as $|z|\to \infty$. Hence, we compute 
\begin{align*}
\check{G}_e(\vbq, z)=\frac{e^{-q|z|}}{2q}~,\quad q=|\vbq|>0~.
\end{align*}

By the change of variable $(t', \vrp')\mapsto (\tau,\vs)=(t-t',\vrp-\vrp')$, \eqref{eq:lin-int-eq-transf-hom} is reduced to
\begin{align*}
	F(z)+\frac{e^2(N/A)}{\varepsilon_0}\frac{\beta}{2q}\int_{\mathbb{R}}\int_{\mathfrak{C}}\int_{\mathbb{R}}\slashed{G}_A(\tau,\vs, z;z')e^{\im\omega\tau-\im\vbq\cdot \vs}f_0(z') \{\mathcal D_q(f_0 F)(z')\}\,\dv\tau\,\dv \vs\,\dv z'=0~,
\end{align*}
where the linear operator $\mathcal D_q: L^2(\mathbb{R}) \rightarrow L^2(\mathbb{R})$ is defined in~\eqref{eq:D-op}. By carrying out the integration with respect to $(\tau, \vs)$, in view of~\eqref{eq:slashed-def} for $U=G_A$, we obtain
\begin{align*}
	F(z)+\frac{\beta}{2q} \frac{e^2 (N/A)}{\varepsilon_0}\int_{\mathbb{R}}\widehat{\slashed{G}}_A(\omega, \vbq, z; z')\,f_0(z')\,\{\mathcal D_q(f_0 F)(z')\}\,\dv z'=0~,\ z\in\mathbb{R}~,
\end{align*}
with $\omega\in\mathbb{C}$, $\Rel\,\alpha(\omega_b\pm \omega, \vbq)>0$ and $\vbq\in\mathbb{Z}_A^2\setminus \{0\}$, where  
\begin{equation*}
	\widehat{\slashed{G}}_A(\omega, \vbq, z; z')=\widehat{G}_A\big(\omega_b+\omega, \vbq, z; z'\big)+\widehat{G}_A\big(\omega_b-\omega^*, -\vbq, z; z'\big)^*~.
\end{equation*}
Now apply the identity $\widehat{G}_A(\omega_b-\omega^*, -\vbq, z; z')^*=\widehat{G}_A(\omega_b-\omega, -\vbq, z; z')=\widehat{G}_A(\omega_b-\omega, \vbq, z; z')$, by virtue of the property $\alpha(w, \vbk)^*=\alpha(w^*, \vbk)$ in the top Riemann sheet of $\alpha(w, \vbk)$.
In the macroscopic limit, we thus obtain~\eqref{eqs:F-int_eq}. In this limit, the function $\widehat{G}_A$ is replaced by $\widehat{G}$; and the Fourier domain $\mathfrak{R}_A$ of $\widehat G_A(w, \vbk, z; z')$ becomes $\mathfrak{R}=\{(w,\vbk)\in \mathbb{C}\times \mathbb{R}^2\,
	\big|\,\alpha(w, \vbk)\neq \beta~,\ \Rel\,\alpha(w, \vbk)> 0\}$; cf.~Proposition~\ref{prop:FT-G}. The conditions $\Rel\,\alpha(\omega_b\pm\omega, \vbq)>0$, where the binding energy is $\omega_b=-\beta^2<0$, entail the restriction $|\Arg(q^2+|\omega_b|\pm \omega)|< \pi$. \hfill  $\Box$

\subsubsection{Integral equation for $f_0 F$}
\label{sssec:Fredholm-IE}
Equation~\eqref{eq:int_eq-F} is recast into the  form
\begin{align}\label{eq:Fredholm-eq}
	\dF(z)+\int_{-\infty}^\infty \sG(\omega, \vbq, z; z') \dF(z')\,\dv z'=0\quad  \forall\ z\in\mathbb{R}~,\quad \dF(z):=f_0(z) F(z)~,
\end{align}
where the kernel $\sG(\omega, \vbq, z; z')$ is described  below. We will see that the structure of this kernel as a function of $(z, z')$ suggests an integral equation of the Volterra type. We will seek even solutions to~\eqref{eq:Fredholm-eq} via the (one-sided) Laplace transform in $z$ ($z>0$). This approach leads to a somewhat uncommon functional equation for the Laplace-transformed solution.

By inspecting the integral of~\eqref{eq:int_eq-F}, we write 
\begin{subequations}\label{eq:int-eqs-G1}
\begin{align}
	\sG(\omega,\vbq, z; z')&=\frac{\beta}{2q}\frac{e^2 \eta_0}{\varepsilon_0}f_0(z) \left\{{\mathcal G}(\omega_b+\omega,\vbq, z; z')+{\mathcal G}(\omega_b-\omega,-\vbq, z; z')\right\}~,\label{eq:G-G1}
\end{align}
where
\begin{align}
	{\mathcal G}(w,\vbq, z; z')&:=\int_{-\infty}^{\infty}\widehat{G}(w,\vbq, z;z'')f_0(z'') e^{-q|z''-z'|}\ \dv z''={\mathcal G}(w,-\vbq, z; z')\label{eq:G1-def}
\end{align}
and $\widehat{G}(w,\vbq, z;z')$ is defined by~\eqref{eq:FT-G} in the macroscopic limit, with Fourier domain $\mathfrak{R}=\{(w,\vbq)\in \mathbb{C}\times \mathbb{R}^2\,
	\big|\,\alpha(w, \vbq)\neq \beta~,\ \Rel\,\alpha(w, \vbq)> 0\}$. By~\eqref{eq:G-G1}, $\sG(\omega,\vbq, z; z')$ involves the following values of $\alpha(w, \vbk)=\sqrt{k^2-w}$:
\begin{align}\label{eq:alpha-pm-def}
\alpha_{\pm}:=\alpha(\omega_b\pm \omega, \vbq)=\sqrt{\beta^2+q^2\mp \omega}~,\quad \Rel\,\alpha_\pm>0~.
\end{align}
\end{subequations}

By comparison of~\eqref{eq:Fredholm-eq} to~\eqref{eq:op-scattering}, we set 
\begin{subequations}\label{eqs:G1-L-ops}
\begin{align}\label{eq:G1-op-def}
\mathcal G_1(w,\vbq, z; z')&=\frac{\beta}{2q}\frac{e^2 \eta_0}{\varepsilon_0}f_0(z){\mathcal G}(w,\vbq, z; z')~.
\end{align}
Therefore, the linearized operator $\mathcal L(\omega, \vbq)$, loosely introduced in~\eqref{eq:Lindhard-op}, is now defined by  
\begin{align}
\mathcal L(\omega, \vbq, z; z')&=\delta(z, z')+\sG(\omega, \vbq, z; z')\notag\\
&=\delta(z, z')+	\frac{\beta}{2q}\frac{e^2 \eta_0}{\varepsilon_0}f_0(z) \left\{{\mathcal G}(\omega_b+\omega,\vbq, z; z')+{\mathcal G}(\omega_b-\omega,\vbq, z; z')\right\}~,
\end{align}
\end{subequations}
where $\delta(z, z')=\delta(z-z')$. Equation~\eqref{eq:Fredholm-eq} reads $\mathcal L(\omega, \vbq)\dF=0$.

Let us compute $\mathcal G(w,\vbq, z; z')$. By virtue of the symmetry property
\begin{align}\label{eq:symm-G1-ker}
	\mathcal G(w,\vbq, z; -z')=\mathcal G(w,\vbq, -z; z')~,\quad  (z,z')\in\mathbb{R}^2~,
\end{align}
it suffices to derive formulas for $\mathcal G(w,\vbq, z; z')$ if $(z,z')\in \mathbb{R}_+\times\mathbb{R}$.
By carrying out the integration in~\eqref{eq:G1-def} in view of~\eqref{eq:FT-G}, we obtain the following formulas.

\noindent {\em Region with $0\le  z'\le z$}. In this case, we find
\begin{subequations}\label{eqs:form-ker}
\begin{align}\label{eq:form-ker-I}
	\mathcal G(w,\vbq, z; z')&= \frac{2\beta q\,e^{-\alpha z-q z'}}{(\alpha-\beta)[(\alpha+\beta)^2-q^2](\alpha-\beta+q)}
	+\frac{e^{-q(z-z')-\beta z}}{(\alpha-\beta-q)(\alpha+\beta+q)}\notag\\
	&\quad -\frac{q\,e^{-\alpha(z-z')-\beta z'}}{\alpha [(\alpha-\beta)^2-q^2]}
	-\frac{\beta q\, e^{-\alpha(z+z')-\beta z'}}{\alpha(\alpha-\beta)[(\alpha+\beta)^2-q^2]}~.
\end{align}

\noindent{\em Region with $0\le z\le z'$}. Our calculation yields
\begin{align}\label{eq:form-ker-II}
	\mathcal G(w,\vbq, z; z')&= \frac{2\beta q\,e^{-\alpha z-q z'}}{(\alpha-\beta)[(\alpha+\beta)^2-q^2](\alpha-\beta+q)}
	+\frac{e^{-q(z'-z)-\beta z}}{(\alpha-\beta+q)(\alpha+\beta-q)}\notag\\
	&\quad -\frac{q\,e^{-\alpha(z'-z)-\beta z'}}{\alpha [(\alpha+\beta)^2-q^2]}
	-\frac{\beta q\, e^{-\alpha(z+z')-\beta z'}}{\alpha(\alpha-\beta)[(\alpha+\beta)^2-q^2]}~.
\end{align}

\noindent{\em Region with $z'\le 0\le z$}. We now compute
\begin{align}\label{eq:form-ker-III}
	\mathcal G(w,\vbq, z; z')&= -\frac{2\beta q\,e^{-\alpha z+q z'}}{(\alpha-\beta)[(\alpha+\beta)^2-q^2](\alpha-\beta-q)}+\frac{e^{-q(z-z')-\beta z}}{(\alpha-\beta-q)(\alpha+\beta+q)}\notag\\
	&\quad -\frac{q\,e^{-\alpha(z-z')+\beta z'}}{(\alpha-\beta) [(\alpha+\beta)^2-q^2]}~.
\end{align}
\end{subequations}
In~\eqref{eqs:form-ker}, we use $\alpha=\alpha(w, \vbq)=\sqrt{q^2-w}$.

Equation~\eqref{eq:symm-G1-ker} enables us to extend~\eqref{eqs:form-ker} to $z\le 0$, for all $z'\in\mathbb{R}$.  Thus, we readily verify that, for fixed $z\ge 0$ and $(w, \vbq)\in \mathfrak{R}=\{(w,\vbq)\in \mathbb{C}\times \mathbb{R}^2\,
	\big|\,\alpha(w, \vbq)\neq \beta, \Rel\,\alpha(w, \vbq)> 0\}$, this $\mathcal G(w,\vbq, z; z')$ is a continuously differentiable function of $z'$, in accord with the integral of~\eqref{eq:G1-def} with~\eqref{eq:FT-G}. 
Note that for real $\omega$ and $\alpha_\pm >0$, the kernels ${\mathcal G}(\omega_b\pm \omega,\vbq, z; z')$ are real symmetric; hence, the operator $\mathcal L(\omega, \vbq)$ is Hermitian. We expect that the energy spectrum $\omega(q)$ of the dispersion relation lies in this regime within our model.

\subsubsection{On even- and odd-in-$z$ scattering amplitudes}
\label{sssec:even-odd}
 
We will routinely argue that~\eqref{eq:Fredholm-eq}, written as $\mathcal L(\omega, \vbq)\dF=0$, can be recast into two decoupled equations for the even and odd parts of the function $\dF(z)$. Suppose that a nontrivial solution $\dF$ of~\eqref{eq:Fredholm-eq} with energy $\omega=\omega(q)$ exists ($\dF\in L^2(\mathbb{R})$). Let $\mathcal P$ be the parity operator on $L^2(\mathbb{R})$, which amounts to reflection through the origin, $z=0$; thus, $\mathcal P^2=1$. 
The property $\sG(\omega,\vbq, -z; z')=\sG(\omega,\vbq, z; -z')$ implies that $[\mathcal P, \mathcal L(\omega, \vbq)]=\mathcal P\mathcal L(\omega, \vbq)-\mathcal L(\omega, \vbq)\mathcal P=0$, i.e., the operators $\mathcal P$ and $\mathcal L$ commute. If $\dF^{\mathrm{e}}:=\tfrac{1}{2}(1+\mathcal P)\dF$ and $\dF^{\mathrm{o}}:=\tfrac{1}{2}(1-\mathcal P)\dF$, the even ($\dF^{\mathrm{e}}$) and odd ($\dF^{\mathrm{o}}$) part of $\dF$, the contraction of $\mathcal P$ with $\mathcal L(\omega, \vbq)(\dF^{\mathrm{e}}+\dF^{\mathrm{o}})=0$ yields 
\begin{equation*}
\mathcal L(\omega, \vbq)(\dF^{\mathrm{e}}-\mathcal \dF^{\mathrm{o}})=0~,
\end{equation*}
since $\mathcal P\dF^{\mathrm{e}}=\dF^{\mathrm{e}}$ and $\mathcal P\dF^{\mathrm{o}}=-\dF^{\mathrm{o}}$. Hence, we find 
\begin{align}\label{eq:F-even-odd}
	\mathcal L(\omega, \vbq)\dF^{\mathrm{e,o}}=0~.
\end{align}
If the null space of $\mathcal L(\omega(q), \vbq)$ has dimension $1$ for given energy $\omega(q)$,  the nontrivial amplitude $\dF(z)$ of~\eqref{eq:Fredholm-eq} is an even  or  odd function, obeying one of the two equations in~\eqref{eq:F-even-odd}.

Let us rewrite~\eqref{eq:Fredholm-eq} by distinguishing the cases of even and odd amplitudes $\dF_{\mathrm{e,o}}(z)$ with domain $\mathbb{R}_+$. In accord with~\eqref{eq:F-even-odd}, we obtain	
\begin{align}\label{eq:int-eq-pos_domain}
	\dF^{\mathrm{e,o}}(z)+\int_{0}^\infty \left\{\sG(\omega, \vbq, z; z')\pm \sG(\omega, \vbq, -z; z')\right\} \dF^{\mathrm{e,o}}(z')\,\dv z'=0\quad  \mbox{all}\ z\ge 0~,
\end{align}
 where $\dF^{\mathrm{e}, \mathrm{o}}\in L^2(\mathbb{R}_+)$. For definiteness, we will focus on even amplitudes, $\dF^{\mathrm{e}}$, intuitively expecting that these should be characterized by a lower energy, $\omega(q)$. Our methodology for even $\dF(z)$, based on the application of the Laplace transform in positive $z$, should be applicable to odd-in-$z$ amplitudes as well. Only the former case is studied in this paper. 
 
\section{Dispersion relation for even scattering amplitude}
\label{sec:derivation-disp}
In this section, we solve~\eqref{eq:int-eq-pos_domain} for even-in-$z$ amplitudes, $\dF=\dF^{\mathrm{e}}$. In particular, 
we apply the Laplace transform with respect to positive $z$,\cite{Widder-book} and derive a functional equation for the transformed solution, $\breve\dF(s)$. This equation relates five values of $\breve\dF(s)$. By an analytic-continuation procedure, we show that $\breve\dF(s)$ is a meromorphic function having a sequence of simple poles in the left half $s$-plane ($\Img\,s<0$). By using these singularities, we extract an expansion for $\breve\dF(s)$ in terms of partial fractions via the Mittag-Leffler theorem.\cite{Ahlfors-book} 

We assume that $\omega\in\mathbb{C}$ with  $0< q<|\omega|/q< \beta$. Recall that $\Rel\,\alpha_\sigma>0$ where 
\begin{equation*}
\alpha_\sigma=\sqrt{\beta^2+q^2-\sigma\omega}\qquad (\sigma=\pm)~.
\end{equation*}
In particular, if $0<q<\omega/q< \beta$ (for $\omega>0$) then $\alpha_->\alpha_+ >0$.

By~\eqref{eq:Fredholm-eq}, we have $\dF(z)=e^{-\beta z} F(z)$ where $F$ enters ansatz~\eqref{eq:plane-ansatz-SP}; see Definition~1. We posit that $F\in L^1(\mathbb{R}_+)\cap L^2(\mathbb{R}_+)$, anticipating that the amplitude $F(z)$ of~\eqref{eq:plane-ansatz-SP} is bounded for all $z\ge 0$ and decays exponentially with $z$ as $z\to +\infty$. 

The Laplace transform of $\dF(z)$ is  
\begin{subequations}\label{eqs:LT}
\begin{align}\label{eq:LT-dir}
	\breve\dF(s)=\int_0^\infty e^{-sz}\dF(z)\,\dv z~,\quad \Rel\,s> \Rel\,s_*~,
\end{align}
for some $s_*\in\mathbb{C}$ with $\Rel\,s_*\le -\beta$. This integral defines an analytic function if $\Rel\,s>\Rel\,s_*$.\cite{Widder-book} The point $s_*$ is the singularity of $\breve\dF(s)$ that is closest to and to the left of the line $\{\Rel\,s=-\beta\}$ in the complex $s$-plane;\cite{Widder-book} see~\eqref{eq:s-star} in Sec.~\ref{subsec:funct-eq}.
The inversion of~\eqref{eq:LT-dir} gives\cite{Widder-book}
\begin{align}\label{eq:LT-inv}
	\dF(z)=\frac{1}{2\pi \im}\int\limits_{\gamma_0-\im \infty}^{\gamma_0+\im \infty}\breve \dF(s) e^{sz}\,\dv s~,\quad \gamma_0> \Rel\,s_*\quad (\mbox{all}\ z\ge 0)~.
\end{align}
\end{subequations}
In fact, $\breve\dF(s)$ turns out to be meromorphic, as discussed in Sec.~\ref{subsec:funct-eq}. Thus, \eqref{eq:LT-inv} may lead to an explicit result for $F(z)$ via the residue theorem, if the integration path is appropriately shifted away from and to the left of the line $\{\Rel\,s=\Rel\,s_*\}$; see Sec.~\ref{sssec:explct-F}.  

\subsection{Functional equation and analyticity of $\breve\dF(s)$}
\label{subsec:funct-eq} 
By~\eqref{eq:int-eq-pos_domain} with an even amplitude $\dF(z)$, the governing integral equation is
\begin{align}\label{eq:int-eq-pos_domain-rep}
	\dF(z)=-\int_{0}^\infty \mathfrak K(z,z') \dF(z')\,\dv z'~,\  z\ge 0~;\quad \mathfrak K(z, z'):=\sG(\omega, \vbq, z; z')+ \sG(\omega, \vbq, z; -z')~.
\end{align}
Notice that the kernel $\mathfrak K(z,z')$ is real symmetric if $\omega$ is real with $0<q<|\omega|/q<\beta$. We can further manipulate the integral on the right-hand side of~\eqref{eq:int-eq-pos_domain-rep}.
By use of~\eqref{eqs:form-ker}, after some algebra we find 
\begin{align}\label{eq:integral-RHS}
	\int_{0}^\infty \mathfrak K(z,z') \dF(z')\,\dv z'&=\frac{\beta}{2q}\frac{e^2 \eta_0}{\varepsilon_0}
	\sum_{\sigma=\pm}\left\{\breve \dF(q)\left(-\frac{4q^2\beta e^{-(\alpha_\sigma+\beta) z}}{(\alpha_\sigma-\beta)[(\alpha_\sigma+\beta)^2-q^2][(\alpha_\sigma-\beta)^2-q^2]}\right. \right. \notag\\
	&\qquad \qquad \left. +\sum_{\varsigma=\pm 1}\frac{e^{-(2\beta+\varsigma q)z}}{(\alpha_\sigma-\beta-\varsigma q)(\alpha_\sigma+\beta+\varsigma q)}\right)\notag\\
	&\qquad - \frac{\breve \dF(\alpha_\sigma+\beta)}{(\alpha_\sigma+\beta)^2-q^2}\frac{q}{\alpha_\sigma}\left( e^{-(\alpha_\sigma+\beta) z}\frac{\alpha_\sigma+\beta}{\alpha_\sigma-\beta}+e^{(\alpha_\sigma-\beta) z} \right) \notag\\
	&+ e^{-2\beta z}\sum_{\varsigma=\pm 1}\varsigma\left[\frac{1}{(\alpha_\sigma-\beta-\varsigma q)(\alpha_\sigma+\beta+\varsigma q)}\int_0^z e^{-\varsigma q(z-z')} \dF(z')\,\dv z' \right.\notag\\
	&\left.\left.+\frac{q}{\alpha}\frac{1}{(\beta+\varsigma\alpha_\sigma)^2-q^2}\int_0^z e^{(\beta+\varsigma\alpha_\sigma)(z-z')}\dF(z')\,\dv z'\right]\right\}~,\quad z\ge 0~,
\end{align}
where the integrals containing $\dF(z)$ over the interval $[0, \infty)$ have been replaced by the values of $\breve\dF(s)$ at $s=q, \alpha_\pm+\beta$. 
Because the convolution integrals of the form $\int_0^z (\cdot)\,\dv z'$, shown in the last two lines of~\eqref{eq:integral-RHS}, are multiplied by the exponential $e^{-2\beta z}$, the overall integral equation~\eqref{eq:int-eq-pos_domain-rep} resembles a nonconvolution Volterra equation on $\mathbb{R}_+$.\cite{Gripenberg-book,Brunner-book}
The associated kernels are exponentials. Therefore, we anticipate that the Laplace transform with respect to $z$ of~\eqref{eq:int-eq-pos_domain-rep} would lead to an equation involving $\breve \dF(s)$ and $\breve \dF(s+2\beta)$ with the latter being multiplied by a sum of partial fractions of the form $C_{\#}\,(s-s_{\#})^{-1}$, where $C_{\#}$ is a coefficient that depends on $s_{\#}$ and $s_{\#}=-2\beta\pm q, -\beta+\alpha_\pm, -\beta-\alpha_\pm$, as shown below. The ensuing functional equation for $\breve\dF(s)$ would also involve the values $\breve \dF(q)$ and $\breve \dF(\beta+\alpha_\pm)$ which must be determined self consistently, as parts of the solution $\breve\dF(s)$. 

By application of the Laplace transform in $z$ to~\eqref{eq:int-eq-pos_domain-rep}, we obtain the functional equation
\begin{align}\label{eq:funct-eq}
\breve\dF(s)&=\frac{\beta}{2q}\frac{e^2 \eta_0}{\varepsilon_0}\sum_{\sigma=\pm}\left\{\breve \dF(q)\left(\frac{4q^2\beta}{(\alpha_\sigma-\beta)[(\alpha_\sigma+\beta)^2-q^2][(\alpha_\sigma-\beta)^2-q^2]}\,\frac{1}{s+\beta+\alpha_\sigma}\right. \right. \notag\\
	&\hphantom{\frac{\beta}{2q}\frac{e^2 \eta_0}{\varepsilon\varepsilon_0}}\qquad \qquad \quad \left. +\sum_{\varsigma=\pm 1}\frac{1}{(\beta+\varsigma q)^2-\alpha_\sigma^2}\,\frac{1}{s+2\beta+\varsigma q}\right)\notag\\
	&\hphantom{\frac{\beta}{2q}\frac{e^2 \eta_0}{\varepsilon\varepsilon_0}}\quad \qquad +\frac{\breve \dF(\alpha_\sigma+\beta)}{(\alpha_\sigma+\beta)^2-q^2}\frac{q}{\alpha_\sigma}\left(\frac{\alpha_\sigma+\beta}{\alpha_\sigma-\beta}\,\frac{1}{s+\beta+\alpha_\sigma}+\frac{1}{s+\beta-\alpha_\sigma} \right) \notag\\
	&\hphantom{\frac{\beta}{2q}\frac{e^2 \eta_0}{\varepsilon\varepsilon_0}}\qquad 
	+ \breve\dF(s+2\beta)\sum_{\varsigma=\pm 1}\varsigma\left[\frac{1}{(\beta+\varsigma q)^2-\alpha_\sigma^2} \,\frac{1}{s+2\beta+\varsigma q} \right.\notag\\
	&\hphantom{\frac{\beta}{2q}\frac{e^2 \eta_0}{\varepsilon\varepsilon_0}\sum_{\varsigma=\pm 1}}\left.\left.\qquad -\frac{q}{\alpha_\sigma}\frac{1}{(\beta+\varsigma\alpha_\sigma)^2-q^2}\,\frac{1}{s+\beta-\varsigma \alpha_\sigma}\right]\right\}~,\quad \Rel\,s> \Rel\,s_*~.	
\end{align}
The restriction $\Rel\,s> \Rel\,s_*$ is typically included in the formulation as a reminder for the convergence of Laplace integral~\eqref{eq:LT-dir} and the possible presence of a singularity at $s=s_*$. We will analytically continue $\breve\dF(s)$ to $\{\Rel\,s< \Rel\,s_*\}$ by exploiting the structure of~\eqref{eq:funct-eq}. 

A few implications of~\eqref{eq:funct-eq} should be pointed out; for more details, see Appendix~\ref{app:sec:prop-funct-eq}. The points $s=-\beta+\alpha_\pm$ and $s=-2\beta+q$ are removable singularities of $\breve\dF(s)$, although they appear as poles of coefficients in~\eqref{eq:funct-eq}; see Appendix~\ref{app:subsec:remov-sing}. Equation~\eqref{eq:funct-eq} entails that $\breve\dF(s)$ is analytic for $\Rel\,s> -\beta-\min\{\Rel\,\alpha_+, \Rel\,\alpha_-\}$   (Appendix~\ref{app:subsec:reg-analt}). By successively applying the shift $s\rightarrow s+2\beta$ $n$ times in~\eqref{eq:funct-eq} ($n=0,\,1,\,\ldots$), we can show that $\breve\dF(s)$ has simple poles at  
\begin{subequations}\label{eqs:singularities}
\begin{equation}\label{eq:poles}
s=s_n^\sigma=-\alpha_\sigma-(2n+1)\beta\quad \mbox{and}\quad s=\slashed{s}_n=-q-2(n+1)\beta~,\quad \mbox{all}\ n\in\mathbb{N}\quad (\sigma=\pm)~;
\end{equation}
see Appendix~\ref{app:subsec:poles-res}. $\breve\dF(s)$ is analytic everywhere in $\mathbb{C}$ with the exception of the points $\{s_n^\sigma \}_{n\in \mathbb{N}}$ and $\{\slashed{s}_n\}_{n\in\mathbb{N}}$. 
By~\eqref{eq:LT-dir}, $s_*$ is the pole of $\breve\dF(s)$ with the smallest value of $|\Rel\,s+\beta|$, e.g.,
\begin{align}\label{eq:s-star}
	s_*= -\beta-\alpha_+~,\qquad  \mbox{if}\quad 0<q^2<\omega< q\beta~. 
\end{align}
\end{subequations}
If $\omega\in\mathbb{C}$ and $0<q^2<|\omega|< q\beta$, we have $s_*=-\beta-\alpha_{\sigma_*}$ with $\sigma_*={\rm Argmin}_{\sigma\in\{\pm\}}(\Rel\,\alpha_\sigma)$. 

The residues of $\breve\dF(s)$ are expressed by 
\begin{align}\label{eq:res}
R_n^\sigma:=\Res_{\ s=s_n^\sigma}\breve\dF(s)=\Lambda_n^\sigma R_0^\sigma~,\quad 	\slashed{R}_n:=\Res_{\ s=\slashed{s}_n}\breve\dF(s)=\slashed{\Lambda}_n \slashed{R}_0~,
\end{align}
for all $n\in\mathbb{N}$ and $\sigma=\pm$; see Appendix~\ref{app:subsec:poles-res}. The zeroth-order residues are computed as 
\begin{subequations}\label{eqs:residues0}
\begin{align}
	R_0^\sigma&= \frac{\beta e^2 \eta_0}{2\varepsilon_0}\left\{\frac{4q\beta \breve\dF(q)}{(\alpha_\sigma-\beta)[(\alpha_\sigma+\beta)^2-q^2][(\alpha_\sigma-\beta)^2-q^2]}+\frac{1}{\alpha_\sigma}\left[\frac{\alpha_\sigma+\beta}{\alpha_\sigma-\beta} \frac{\breve\dF(\beta+\alpha_\sigma)}{(\alpha_\sigma+\beta)^2-q^2} \right.\right.\notag\\
	& \hphantom{\frac{\beta e^2 \eta_0}{2\varepsilon_0}} \left.\left. \qquad 
	+ \frac{\breve\dF(\beta-\alpha_\sigma)}{(\alpha_\sigma-\beta)^2-q^2} \right]\right\}~, \label{eq:resR0}\\
\slashed{R}_0
&= \frac{2e^2 \eta_0}{\varepsilon_0}\frac{\beta^2}{4q^2\beta^2-\omega^2}\{\breve\dF(q)+\breve\dF(-q) \}~.\label{eq:slashed-resR0}
\end{align}
\end{subequations}
Furthermore, the factors $\Lambda_n^\sigma$ and $\slashed{\Lambda}_n$ are given by  $\Lambda_0^\sigma=\slashed{\Lambda}_0=1$ (if $n=0$, $\sigma=\pm$), and
\begin{subequations}\label{eqs:residues-n}
\begin{align}
	\Lambda_n^\sigma&=\biggl(-\frac{\beta e^2 \eta_0}{\varepsilon_0} \biggr)^n\, \prod_{j=1}^n \left\{\frac{1}{[\alpha_\sigma+(2j-1)\beta]^2-q^2}\sum_{\sigma'=\pm}\frac{1}{(\alpha_\sigma+2j\beta)^2-\alpha_{\sigma'}^2}\right\}\notag\\
	&= \frac{1}{n!}\biggl(-\frac{e^2 \eta_0}{2\varepsilon_0} \biggr)^n\prod_{j=1}^{n}\left\{\frac{1}{j\beta+\alpha_\sigma}\left[1+\frac{\sigma \omega}{4j\beta(j\beta+\alpha_\sigma)-2\sigma\omega} \right]\right.\notag\\
	& \left.\hphantom{\frac{1}{n!}\biggl(-\frac{e^2 \eta_0}{2\varepsilon_0} \biggr)^n}\times \frac{1}{[1+(2j-1)^2]\beta^2+2(2j-1)\beta\alpha_\sigma-\sigma\omega}\right\}~,\quad  n\ge 1~,\label{eq:resRn-Lambda-n}
\end{align}
\begin{align}	
	\slashed{\Lambda}_n&=\frac{1}{n!}\biggl(-\frac{e^2\eta_0}{4\varepsilon_0} \biggr)^n\, \prod_{j=1}^n \left\{\frac{1}{j\beta+q}\sum_{\sigma'=\pm}\frac{1}{[(2j+1)\beta+q]^2-\alpha_{\sigma'}^2}\right\}\notag\\
	&=\frac{1}{n!}\biggl(-\frac{e^2\eta_0}{\varepsilon_0} \biggr)^n \prod_{j=1}^n
	\frac{\beta}{j\beta+q}\,\frac{2j(j+1)\beta+(2j+1)q}{4\beta^2[2j(j+1)\beta+(2j+1)q]^2-\omega^2}~,\quad n\ge 1~.
	\label{eq:slashed-resRn-Lambda-n}
\end{align}
\end{subequations}
Notably, $R_n^\sigma$ and $\slashed{R}_n$ tend to $0$ at least as fast as $(n!)^{-4}$ in the limit $n\to +\infty$ (Sec.~\ref{subsec:disp-soln}). 

Evidently, $\breve\dF(s)$ approaches $0$ as $s\to \infty$ in the complex $s$-plane away from the poles, $s_n^\pm$ and $\slashed{s}_n$, which asymptotically approach the negative real axis as $n\to +\infty$ for fixed $(\omega, q)\in \mathbb{C}\times \mathbb{R}_+$. This property of $\breve\dF(s)$ is shown in Appendix~\ref{app:subsec:lim-F}. We have not been able to find nontrivial solutions to~\eqref{eq:funct-eq} in simple closed form.\cite{Sahoo-Kannappan-book,Czerwik-book} We proceed to develop a nontrivial solution $\breve{\dF}(s)$ via a convergent series expansion in the complex $s$-plane.

\subsection{Solution of functional equation for $\breve\dF(s)$}
\label{subsec:disp-soln} 
Next, we state and prove a key lemma (Lemma~1); and use it to derive the SP-type dispersion relation, proving Proposition~\ref{prop:disp_reln}. In particular, we expand $\breve\dF(s)$ in a series of partial fractions. This series linearly depends on six values of $\breve\dF(s)$, for $0<q<|\omega|/q<\beta$. We will formulate a homogeneous linear system for these quantities, and extract a relation of the form $\Lambda(\omega, q)=0$ for nontrivial solutions. We will thus describe the amplitude $F(z)$.

\subsubsection{Lemma for $\breve\dF(s)$}
\label{sssec:exp-ML}
The following result will be invoked in the proof of Proposition~\ref{prop:disp_reln}.

\medskip 

\noindent{\bf Lemma 1} 
	{\em Any solution $\breve\dF(s)$ of~\eqref{eq:funct-eq} admits the expansion} 
\begin{align}\label{eq:ML-exp}
	\breve\dF(s)=\sum_{n=0}^\infty \frac{R_n^+}{s-s_n^+}+\sum_{n=0}^\infty \frac{R_n^-}{s-s_n^-}+\sum_{n=0}^\infty \frac{\slashed{R}_n}{s-\slashed{s}_n}~,
\end{align}
{\em which converges uniformly for all $s\in\mathbb{C}\setminus\{s_n^\pm~,\,\slashed{s}_n\}_{n\ge 0}$.}

\medskip

\noindent {\em Proof.} First, we show that each of the series on the right-hand side of~\eqref{eq:ML-exp} converges for all $s\in\mathbb{C}$ with $s\neq s_n^\pm~,\,\slashed{s}_n$ ($\forall\ n\in\mathbb{N}$), for fixed parameters $\beta$, $q$ and $\omega\in\mathbb{C}$ with $0<q<|\omega|/q<\beta$. To this end, in view of~\eqref{eq:res}--\eqref{eqs:residues-n} we invoke the inequalities (for $\sigma=\pm 1$)
\begin{subequations}\label{eqs:ineq-sums}
\begin{align}\label{eq:ineq-sum1}
	\left|\frac{1}{[\alpha_\sigma+(2j-1)\beta]^2-q^2}\sum_{\sigma'=\pm}\frac{1}{(\alpha_\sigma+2j\beta)^2-\alpha_{\sigma'}^2}\right|\le \frac{C_0}{j^4 \beta^4}~,
\end{align}
\begin{align}\label{eq:ineq-sum2}
	\left|\frac{1}{q+j\beta}\sum_{\sigma'=\pm}\frac{1}{[(2j+1)\beta+q]^2-\alpha_{\sigma'}^2}\right|\le \frac{C_0}{j^3 \beta^3}~,\quad  \forall\ j\ge 1~,
\end{align}
\end{subequations}
where $C_0>0$ is a (parameter- and $j$-independent) numerical constant. Inequality~\eqref{eq:ineq-sum2} is implied by the estimate
\begin{align*}
(q+j\beta)\left|[(2j+1)\beta+q]^2-\alpha_{\sigma'}^2\right|\ge j\beta\left|(4j^2+4j)\beta^2+2(2j+1)\beta q-|\omega|\right|\ge 4j^3\beta^3~.
\end{align*}
In regard to~\eqref{eq:ineq-sum1}, we distinguish the cases $\sigma'=\sigma$ and $\sigma'\neq \sigma$, recalling that $\Rel\,\alpha_\sigma>0$. In the case with $\sigma'=\sigma$, the estimates 
\begin{equation*}
\Rel[2(2j-1)\beta\alpha_\sigma-\sigma\omega]\ge 2(2j-1)\beta\Rel\,\alpha_\sigma-|\omega|\ge  -\beta q\quad (\sigma=\pm~,\ j\ge 1)
\end{equation*}
and $|(\alpha_\sigma+2j\beta)^2-\alpha_\sigma^2|=4|j^2\beta^2+j\beta\alpha_\sigma|\ge 4j^2\beta^2$ entail
\begin{align*}
	\left|[\alpha_\sigma+(2j-1)\beta]^2-q^2\right|\,\left|(\alpha_\sigma+2j\beta)^2-\alpha_\sigma^2\right|&\ge  [(4j^2-4j+2)\beta^2-\beta q](4j^2\beta^2)\notag\\
	&\ge 4j^2(2j-1)^2\beta^4\ge 4j^4\beta^4 \quad (\sigma=\pm~,\ j\ge 1)~.
\end{align*}
In the case with $\sigma'\neq \sigma$ for~\eqref{eq:ineq-sum1}, we apply the estimate
\begin{equation*}
	\Rel[4j\beta\alpha_\sigma+(\sigma'-\sigma)\omega]\ge -|\sigma'-\sigma|\,|\omega|\ge -2\beta q~.
\end{equation*}
Hence, we obtain (for $\sigma\sigma'=-1$, $\sigma=\pm 1$ and $j\ge 1$)
\begin{align*}
	\left|[\alpha_\sigma+(2j-1)\beta]^2-q^2\right|\,\left|(\alpha_{\sigma}+2j\beta)^2-\alpha_{\sigma'}^2\right|&\ge (2j-1)^2\beta^2\,\left|4j^2\beta^2+4j\beta\alpha_\sigma+(\sigma'-\sigma)\omega  \right| \notag\\
	&\ge j^2\beta^2(4j^2\beta^2-2\beta q)\ge j^2(4j^2-2)\beta^4\ge 2j^4\beta^4~.
\end{align*}
Therefore, we verify \eqref{eqs:ineq-sums} with $C_0=3/4$.

Thus, by~\eqref{eqs:residues-n} we establish the estimates
\begin{align*}
	|\Lambda_n^\sigma|\le  \frac{(C_{\beta})^n}{(n!)^4}~,\quad |\slashed{\Lambda}_n|\le \biggl(\frac{C_\beta}{4}\biggr)^n  \frac{1}{(n!)^4}~,\quad \mbox{all}\ n\ge 1~;\quad C_\beta:=\frac{C_0 e^2 \eta_0}{\beta^3\varepsilon_0}~.
\end{align*}
By~\eqref{eq:res}, we have $R_n^\sigma=\Lambda_n^\sigma R_0^\sigma$ and $\slashed{R}_n=\slashed{\Lambda}_n \slashed{R}_0$. Furthermore, by~\eqref{eq:poles} $|s_n^\sigma/n|$ and $|\slashed{s}_n/n|$ approach ($n$-independent) constants as $n\to +\infty$. Each series entering the right-hand side of~\eqref{eq:ML-exp} converges for all complex $s\neq s_n^\sigma$ and $\slashed{s}_n$ (all $n\in\mathbb{N}$), as evinced by comparison of the magnitude of each series to the convergent series $\sum_{n=0}^\infty C^n/[(n+1) (n!)^4]$ for some constant $C>0$. In fact, each of the series in~\eqref{eq:ML-exp} converges uniformly on any compact subset of $\mathbb{C}$ by omission of the terms that become infinite at points of that set.\cite{Ahlfors-book}

 Now recall that $\breve\dF(s)$ is a meromorphic function with simple poles at $s=s_n^\pm\,,\ \slashed{s}_n$ ($n\in\mathbb{N}$); see Sec.~\ref{subsec:funct-eq} and Appendix~\ref{app:subsec:poles-res}. By application of the Mittag-Leffler theorem,\cite{Ahlfors-book} we write
\begin{align}\label{eq:g-analt}
\breve\dF(s)=\sum_{n=0}^\infty \frac{R_n^+}{s-s_n^+}+\sum_{n=0}^\infty \frac{R_n^-}{s-s_n^-}+\sum_{n=0}^\infty \frac{\slashed{R}_n}{s-\slashed{s}_n}+g(s)~,	
\end{align}
where $g(s)$ is an entire function. We claim that $g(s)\equiv 0$.

Let us prove this claim. By functional equation~\eqref{eq:funct-eq}, we find that $g(s)$ must obey 
\begin{align}\label{eq:funct-eq-g}
g(s)&=\frac{\beta}{2q}\frac{e^2 \eta_0}{\varepsilon_0} 
	\sum_{\sigma=\pm}\sum_{\varsigma=\pm 1}\varsigma\left\{\frac{1}{(\beta+\varsigma q)^2-\alpha_\sigma^2} \,\frac{g(s+2\beta)-g(-\varsigma q)}{s+2\beta+\varsigma q} \right.\notag\\
	&\hphantom{\frac{\beta}{2q}\frac{e^2 \eta_0}{\varepsilon\varepsilon_0}\sum_{\varsigma=\pm 1}}\left.\qquad -\frac{q}{\alpha_\sigma}\frac{1}{(\beta+\varsigma\alpha_\sigma)^2-q^2}\,\frac{g(s+2\beta)-g(\beta+\varsigma\alpha_\sigma)}{s+\beta-\varsigma \alpha_\sigma}\right\}~,\quad \mbox{all}\ s\in\mathbb{C}~.	
\end{align}
We have argued that $\breve \dF(s)$ approaches $0$ uniformly in the $s$-plane as $s\to \infty$ away from the poles of $\breve \dF(s)$; see Appendix~\ref{app:subsec:lim-F}. Each of the series of partial fractions in~\eqref{eq:g-analt} has the same property, because the series convergence is uniform and each term approaches $0$ as $s\to\infty$ away from the poles. Therefore, from~\eqref{eq:g-analt}, $g(s)$ tends to $0$ uniformly as $s\to \infty$ away from the poles. To simplify the analysis, let us assume that all poles lie in the negative real axis, which happens if $\omega\in\mathbb{R}$ with $0<q<|\omega|/q<\beta$. To figure out what happens to $g(s)$ for real $s$ as $s\to -\infty$, we notice that~\eqref{eq:funct-eq-g} does not allow for an unbounded $|g(s)|$ in this limit. In fact, the only asymptotic behavior of $g(s)$ compatible with~\eqref{eq:funct-eq-g} as $s\to -\infty$ is $g(s)=\mathcal O(1/s)$; see Appendix~\ref{app:sec:asympt-g}. To verify this property heuristically, one may try $g(s)\sim c\, (-s)^\gamma$ as $s\to -\infty$ in~\eqref{eq:funct-eq-g} and apply dominant balance to obtain $\gamma=-1$ for some constant $c$. Hence, we conclude that $g(s)$ is bounded everywhere in the complex $s$-plane. By Liouville's theorem, $g(s)$ is a constant which should be identically zero because of the limit of $g(s)$ as $s\to \infty$. 
Thus, \eqref{eq:g-analt} reduces to~\eqref{eq:ML-exp}.  This argument is extended to complex $\omega$ with $0<q<|\omega|/q<\beta$, since the poles $s_n^\pm$ or $\slashed{s}_n$ that have nonzero imaginary parts asymptotically approach the negative real axis as $n\to +\infty$. \hfill $\Box$

\medskip 

The series expansions entering~\eqref{eq:ML-exp} in Lemma~1 affect the character of the dispersion relation, $\Lambda(\omega, q)=0$ (Proposition~\ref{prop:disp_reln}). In particular, $\Lambda(\omega, q)$ is described by series that converge uniformly with $\tilde q$, $\tilde \omega$ and $\Cs_0(\beta)$, where $0<\tilde q^2 < |\tilde \omega| < \tilde q$ and $\Cs_0$ is bounded.

\subsubsection{Proof of Proposition~\ref{prop:disp_reln}}
\label{sssec:proof-Prop3}
Let us first delineate the main idea of the proof. The residues $R_n^\pm$ and $\slashed{R}_n$ entering expansion~\eqref{eq:ML-exp} of Lemma~1 depend linearly  on the six values $\breve\dF(\varsigma q)$ and $\breve\dF(\beta+\varsigma\alpha_\sigma)$, for $\sigma, \varsigma=\pm$; cf.~\eqref{eq:res}--\eqref{eqs:residues-n}. By expressing each of these values of $\breve\dF(s)$ in terms of expansion~\eqref{eq:ML-exp}, we self-consistently formulate a $6\times 6$ linear homogeneous system of equations. The requirement of nontrivial solutions to this system leads to the desired dispersion relation.

By applying~\eqref{eq:ML-exp} at $s=\varsigma q$ and $s= \beta+\varsigma\alpha_\sigma$ from Lemma~1, we have
\begin{subequations}\label{eqs:sys-disp}
\begin{align}
	\breve\dF(\varsigma q)&= \sum_{\varsigma'=\pm}\mA_{\varsigma'}^{\varsigma} \breve \dF(\varsigma' q)+\sum_{\varsigma', \sigma'=\pm}\mA_{\varsigma'\sigma'}^{\varsigma}\breve\dF(\beta+\varsigma'\alpha_{\sigma'})~,\label{eq:sys-disp-I}\\
	\breve\dF(\beta+\varsigma\alpha_{\sigma})&= \sum_{\varsigma'=\pm}\mA_{\varsigma'}^{\varsigma\sigma} \breve \dF(\varsigma' q)+\sum_{\varsigma', \sigma'=\pm}\mA_{\varsigma'\sigma'}^{\varsigma\sigma}\breve\dF(\beta+\varsigma'\alpha_{\sigma'})~;\quad \varsigma, \sigma=\pm~.\label{eq:sys-disp-II}	
\end{align}
\end{subequations}
By~\eqref{eq:res} and~\eqref{eqs:residues0}, after some algebra we express the coefficients $\mA^{\varsigma}_{\varsigma'}$, $\mA^{\varsigma}_{\varsigma'\sigma'}$, $\mA^{\varsigma\sigma}_{\varsigma'}$ and $\mA^{\varsigma\sigma}_{\varsigma'\sigma'}$ by the following formulas (for $\varsigma, \sigma, \varsigma', \sigma'=\pm$): 
\begin{subequations}\label{eqs:matrix-forms}
\begin{align}
\mA^{\varsigma}_{+}&=\frac{2\beta^2 e^2\eta_0}{\varepsilon_0}\frac{1}{4\beta^2 q^2-\omega^2}\sum_{n=0}^\infty\left\{\sum_{\sigma'=\pm} \Cs_0^{\sigma'}\frac{\Lambda_n^{\sigma'}}{\varsigma q-s_n^{\sigma'}}+\frac{\slashed{\Lambda}_n}{\varsigma q-\slashed{s}_n}  \right\}~,\\ 
\mA^{\varsigma}_-&=\frac{2\beta^2 e^2\eta_0}{\varepsilon_0}\frac{1}{4\beta^2q^2-\omega^2} \sum_{n=0}^\infty \frac{\slashed{\Lambda}_n}{\varsigma q-\slashed{s}_n}~,\\
\mA^{\varsigma\sigma}_{+}&=\frac{2\beta^2 e^2 \eta_0}{\varepsilon_0}\frac{1}{4\beta^2q^2-\omega^2}
\sum_{n=0}^\infty\left\{ \sum_{\sigma'=\pm} \Cs_0^{\sigma'}\frac{\Lambda_n^{\sigma'}}{\beta+\varsigma\alpha_\sigma-s_n^{\sigma'}}+\frac{\slashed{\Lambda}_n}{\beta+\varsigma\alpha_\sigma-\slashed{s}_n}\right\}~,\\
\mA^{\varsigma\sigma}_{-}&=\frac{2\beta^2 e^2 \eta_0}{\varepsilon_0}\frac{1}{4\beta^2q^2-\omega^2} \sum_{n=0}^\infty \frac{\slashed{\Lambda}_n}{\beta+\varsigma\alpha_\sigma-\slashed{s}_n}~,
\end{align}
\begin{align}
\mA^{\varsigma}_{\varsigma'\sigma'}&=\frac{2\beta^2 e^2 \eta_0}{\varepsilon_0}\frac{1}{4\beta^2q^2-\omega^2}\Cs_{\varsigma'}^{\sigma'} \sum_{n=0}^\infty \frac{\Lambda_n^{\sigma'}}{\varsigma q-s_n^{\sigma'}}~,\\ 
\mA^{\varsigma\sigma}_{\varsigma'\sigma' }&=\frac{2\beta^2 e^2\eta_0}{\varepsilon_0}\frac{1}{4\beta^2q^2-\omega^2}\Cs_{\varsigma'}^{\sigma'}\sum_{n=0}^\infty \frac{\Lambda_n^{\sigma'}}{\beta+\varsigma\alpha_\sigma-s_n^{\sigma'}}~. 
\end{align}
\end{subequations}
The quantities $\Lambda_n^\sigma$ and $\slashed{\Lambda}_n$ are given by~\eqref{eqs:residues-n}; and the factors $\Cs_0^\sigma$ and $\Cs^\sigma_\varsigma$ are
\begin{align*}
	\Cs_0^\sigma&=-q\,\frac{\alpha_\sigma+\beta}{q^2-\sigma\omega}~,\\
	\Cs_+^\sigma&=\frac{1}{\alpha_\sigma}\frac{4q^2\beta^2-\omega^2}{4\beta}\frac{1}{q^2-\sigma\omega}\,\frac{(\beta+\alpha_\sigma)^2}{(\beta+\alpha_\sigma)^2-q^2}=-\frac{(\beta+\alpha_\sigma)^2}{q^2-\sigma\omega}\frac{(\beta-\alpha_\sigma)^2-q^2}{4\beta\alpha_\sigma}~,\\
	\Cs_-^\sigma&=-\frac{(\beta+\alpha_\sigma)^2-q^2}{4\beta\alpha_\sigma}~.
\end{align*}

The linear system described by~\eqref{eqs:sys-disp} must have nontrivial 
vector solutions,
\begin{displaymath}
 (\breve\dF(q),\, \breve\dF(-q),\, \breve\dF(\beta+\alpha_+),\, \breve\dF(\beta-\alpha_+),\, \breve\dF(\beta+\alpha_-),\, \breve\dF(\beta-\alpha_-))\neq 0~.
\end{displaymath} 
Thus, the determinant of the corresponding $6\times 6$ matrix $\boldsymbol{\mathfrak{A}}-\boldsymbol{I}$ must vanish. This condition is written as
\begin{align}\label{eq:det-zero}
\Lambda(\omega, q):=\mathrm{det}(\boldsymbol{\mathfrak{A}}-\boldsymbol{I})=
\begin{vmatrix}
	\mA^+_+-1\, & \mA^+_-\, & \mA^+_{++}\, & \mA^+_{-+}\, & \mA^+_{+-}\, & \mA^{+}_{--} \\
	\mA^-_+\, & \mA^-_{-}-1\, & \mA^-_{++}\, & \mA^-_{-+}\, & \mA^-_{+-}\, & \mA^{-}_{--} \\
	\mA^{++}_+\, & \mA^{++}_{-}\, & \mA^{++}_{++}-1\, & \mA^{++}_{-+}\, & \mA^{++}_{+-}\, & \mA^{++}_{--} \\
	\mA^{-+}_+\, & \mA^{-+}_{-}\, & \mA^{-+}_{++}\, & \mA^{-+}_{-+}-1\, & \mA^{-+}_{+-}\, & \mA^{-+}_{--} \\
	\mA^{+-}_+\, & \mA^{+-}_{-}\, & \mA^{+-}_{++}\, & \mA^{+-}_{-+}\, & \mA^{+-}_{+-}-1\, & \mA^{+-}_{--} \\
	\mA^{--}_+\, & \mA^{--}_{-}\, & \mA^{--}_{++}\, & \mA^{--}_{-+}\, & \mA^{--}_{+-}\, & \mA^{--}_{--}-1 
	\end{vmatrix}	
=0~.
\end{align}
Each term in the series for the matrix elements $\mA^\mu_\nu$ of the above determinant is well defined, because of the restriction $0<q^2<|\omega| < \beta q$. Let us now introduce the dimensionless parameters $\tilde{\vbq}=\tfrac{\vbq}{\beta}$, $\tilde\omega=\frac{\omega}{\beta^2}$ and $\tilde\alpha_\sigma=\tfrac{\alpha_\sigma}{\beta}=\alpha(-1+\sigma\tilde\omega,\tilde{\vbq})$ in $\mA^\mu_\nu$. Notice that the coefficients 
$\Cs_0^\sigma$ and $\Cs^\sigma_\varsigma$ are nondimensional, and are solely functions of $\tilde q$ and $\tilde\omega$. The ensuing matrix elements $\mA^\mu_\nu$ are thus given by~\eqref{eqs:prop3-matrix} with~\eqref{eqs:prop3-pmts} and~\eqref{eqs:residues-n-scaled}.  \hfill $\Box$

\medskip 

In principle, \eqref{eq:det-zero} implies that the solution $\breve\dF(s)$ of functional equation~\eqref{eq:funct-eq} contains at least one of the values $\breve\dF(\pm q)$, $\breve\dF(\beta\pm \alpha_{\pm})$ as a free parameter; see  Sec.~\ref{sssec:explct-F}. In the semiclassical limit, we show that $\Lambda(\omega, q)\sim -(\mA^+_+-1)$ where $\mA^+_+$ is approximately equal to $\tfrac{e^2\eta_0}{\beta^3\varepsilon_0}\,\tfrac{\tilde q}{\tilde \omega^2}$, while the origin of the higher-order terms is more complicated (Sec.~\ref{sec:asymptotics}).

\subsubsection{The solution $F(z)$}
\label{sssec:explct-F}
Next, we describe the even solution $F(z)=e^{\beta |z|}\dF(|z|)$ ($z\in\mathbb{R}$) if $q^2<|\omega|<q\beta$. By Laplace-inverting formula~\eqref{eq:ML-exp} for $z>0$ in view of~\eqref{eq:LT-inv}, we obtain the expansion
\begin{align}\label{eq:F-soln-ser}
F(z)=\sum_{n=0}^\infty\sum_{\sigma=\pm} R_n^\sigma e^{(s_n^\sigma+\beta)|z|}+\sum_{n=0}^\infty \slashed{R}_n e^{(\slashed{s}_n+\beta)|z|}\qquad (\mbox{all}\ z\in\mathbb{R})~,	
\end{align}
where $R_n^\sigma= \Lambda_n^\sigma R_0^\sigma$, $\slashed{R}_n=\slashed{\Lambda}_n \slashed{R}_0$, and $R_0^\sigma$, $\slashed{R}_0$, $\Lambda_n^\sigma$ and $\slashed{\Lambda}_n$ are defined by~\eqref{eqs:residues0} and~\eqref{eqs:residues-n}. In particular, $R_0^\sigma$ and $\slashed{R}_0$ linearly depend on the transformed amplitudes $\breve \dF(\pm q)$, $\breve \dF(\beta\pm \alpha_+)$ and $\breve \dF(\beta\pm\alpha_-)$.
By~\eqref{eq:poles} and~\eqref{eq:F-soln-ser}, $F(z)$ is characterized by the {\em exponential-decay} constants
\begin{align*}
	-(s_n^\sigma+\beta)&=\alpha_\sigma+2n\beta=\sqrt{\beta^2+q^2-\sigma\omega}+2n\beta~,\\
	 -(\slashed{s}_n+\beta)&=q+(2n+1)\beta\quad (\mbox{all}\ n\in\mathbb{N},\ \sigma\in \{+,-\})~. 
\end{align*}
Equation~\eqref{eq:F-soln-ser} is a uniformly convergent series. This can be proved via the previous estimates for $\Lambda_n^\sigma$ and $\slashed{\Lambda}_n$; see proof of Lemma~1 (Sec.~\ref{sssec:exp-ML}). We also verify that $F\in L^2(\mathbb{R})\cap L^1(\mathbb{R})$.

The variable $(\omega, \vbq)$ satisfies  dispersion relation~\eqref{eq:det-zero}; see proof of Proposition~\ref{prop:disp_reln} (Sec.~\ref{sssec:proof-Prop3}). 
Hence, the transformed-amplitude values $\breve \dF(\pm q)$, $\breve \dF(\beta\pm \alpha_+)$ and $\breve \dF(\beta\pm\alpha_-)$ are linearly dependent. The rank of the $6\times 6$ matrix $\boldsymbol{\mathfrak{A}}-\boldsymbol{I}$ is less than or equal to $5$. We can show that it is impossible for each of the first minors of $\boldsymbol{\mathfrak{A}}-\boldsymbol{I}$ to vanish under $\Lambda(\omega, q)=0$, if $0<q^2<|\omega|< q\beta$; hence, the rank of $\boldsymbol{\mathfrak{A}}-\boldsymbol{I}$ is equal to $5$. The nullity of this matrix is equal to $1$. To illustrate consequences of this statement, we consider the first minor $M_{6,6}:=\mathrm{det}(\bar{\boldsymbol{\mathfrak{A}}}_{6,6})$ of $\boldsymbol{\mathfrak{A}}-\boldsymbol{I}$, assuming that $M_{6,6}\neq 0$. Here, the $5 \times 5$ matrix
\begin{align*}
\bar{\boldsymbol{\mathfrak{A}}}_{6,6}:=
\begin{pmatrix}
	\mA^+_+-1\, & \mA^+_-\, & \mA^+_{++}\, & \mA^+_{-+}\, & \mA^+_{+-}\\
	\mA^-_+\, & \mA^-_{-}-1\, & \mA^-_{++}\, & \mA^-_{-+}\, & \mA^-_{+-}\\
	\mA^{++}_+\, & \mA^{++}_{-}\, & \mA^{++}_{++}-1\, & \mA^{++}_{-+}\, & \mA^{++}_{+-}\\
	\mA^{-+}_+\, & \mA^{-+}_{-}\, & \mA^{-+}_{++}\, & \mA^{-+}_{-+}-1\, & \mA^{-+}_{+-}\\
	\mA^{+-}_+\, & \mA^{+-}_{-}\, & \mA^{+-}_{++}\, & \mA^{+-}_{-+}\, & \mA^{+-}_{+-}-1
	\end{pmatrix}
\end{align*}
results from the elimination of the sixth column and sixth row of $\boldsymbol{\mathfrak{A}}-\boldsymbol{I}$.
Therefore, each of the five amplitudes $\breve \dF(\pm q)$, $\breve \dF(\beta\pm \alpha_+)$ and $\breve \dF(\beta+\alpha_-)$ can be expressed in terms of $\breve\dF(\beta-\alpha_-)$. For example, we have
\begin{align*}
	\breve \dF(q)=-\breve\dF(\beta-\alpha_-)\,
	\frac{\mathrm{det}(\bar{\boldsymbol{\mathfrak{A}}}^{(1)}_{6,6})}{M_{6,6}}~;\quad 
	\bar{\boldsymbol{\mathfrak{A}}}^{(1)}_{6,6}:=
	\begin{pmatrix}
	\mA^{+}_{--}\, & \mA^+_-\, & \mA^+_{++}\, & \mA^+_{-+}\, & \mA^+_{+-}\\
	\mA^{-}_{--}\, & \mA^-_{-}-1\, & \mA^-_{++}\, & \mA^-_{-+}\, & \mA^-_{+-}\\
	\mA^{++}_{--}\, & \mA^{++}_{-}\, & \mA^{++}_{++}-1\, & \mA^{++}_{-+}\, & \mA^{++}_{+-}\\
	\mA^{-+}_{--}\, & \mA^{-+}_{-}\, & \mA^{-+}_{++}\, & \mA^{-+}_{-+}-1\, & \mA^{-+}_{+-}\\
	\mA^{+-}_{--}\, & \mA^{+-}_{-}\, & \mA^{+-}_{++}\, & \mA^{+-}_{-+}\, & \mA^{+-}_{+-}-1
	\end{pmatrix}~.
\end{align*}
The $5\times 5$ matrix $\bar{\boldsymbol{\mathfrak{A}}}^{(1)}_{6,6}$ comes from the replacement of the first column of $\bar{\boldsymbol{\mathfrak{A}}}_{6,6}$ by 
$\boldsymbol V=(\mA^{+}_{--}, \mA^{-}_{--}, \mA^{++}_{--}, \mA^{-+}_{--}, \mA^{+-}_{--})^T$.
Similarly, each of the remaining amplitudes $\breve \dF(-q)$, $\breve \dF(\beta\pm \alpha_+)$ and $\breve \dF(\beta+\alpha_-)$ can be written in the form $-\breve\dF(\beta-\alpha_-)\mathrm{det}(\bar{\boldsymbol{\mathfrak{A}}}^{(j)}_{6,6})/M_{6,6}$ where the matrix $\bar{\boldsymbol{\mathfrak{A}}}^{(j)}_{6,6}$ comes from the replacement of the $j$-th column of $\bar{\boldsymbol{\mathfrak{A}}}_{6,6}$ by $\boldsymbol V$ ($j=2\,,\ldots,\,5$). The substitution of these formulas for $\breve \dF(\pm q)$, $\breve \dF(\beta\pm \alpha_+)$ and $\breve \dF(\beta+\alpha_-)$ into the residues $R_0^\sigma$ and $\slashed{R}_0$, which enter~\eqref{eq:F-soln-ser}, should furnish $F(z)$ as proportional to $\breve\dF(\beta-\alpha_-)$. We omit the details here.

It is of interest to point out that $F(z)$ is approximately  simplified if $\beta|z|\gg 1$ with $q|z|=\mathcal O(1)$. If $\omega\in\mathbb{R}$ with $0<q<|\omega|/q<\beta$, \eqref{eq:F-soln-ser} becomes 
\begin{align}\label{eq:F-soln-ser-simple}
	F(z)\sim R_0^+ e^{-\alpha_+ |z|}+R_0^- e^{-\alpha_- |z|}+\slashed{R}_0 e^{-(q+\beta)|z|}~,\quad \beta |z|\gg 1~.
\end{align}
In particular,  for $q^2<|\omega|\ll q\beta$ the related decay constants satisfy $0<(q+\beta)-\alpha_\pm= \mathcal O(q)$. The remaining terms of expansion~\eqref{eq:F-soln-ser} are asymptotically subdominant to~\eqref{eq:F-soln-ser-simple}. 
Evidently, the SP-type excitation wave tends to be localized near $z=0$ inside a layer of width $\mathcal O(\beta^{-1})$, as expected on physical grounds.

\subsection{On an analytic property of $\Lambda(\omega, q)$}
\label{subsec:soln-rmks}
The condition $0<q<\tfrac{|\omega|}{q}< \beta$ permeates our derivations. We expect that for fixed $q$ with $0<q<\beta$ the function $\Lambda(\omega, q)=\mathrm{det}(\boldsymbol{\mathfrak{A}}(\omega, q)-\boldsymbol{I})$ of~\eqref{eq:Lambda-disp-reln} is holomorphic in $\{\omega\in \mathbb{C}\,:\,\Rel\,\alpha_\pm>0,\, q^2<|\omega|< 2q(\beta+\beta_0)\}$ for some $\beta_0>0$. Let us discuss  why the points $\omega=\pm 2q\beta$, which appear as simple poles of $\mA^\mu_\nu$ in~\eqref{eqs:prop3-matrix}, are removable singularities of $\Lambda(\omega, q)$. We formally treat the series for $\mA^\mu_\nu$ as convergent in the punctured vicinities of $\omega=\pm 2q \beta$. 

Consider expansion~\eqref{eq:ML-exp} and the representations of~\eqref{eqs:sys-disp} for $0<q<\beta$. As 
$\omega\to \pm 2\beta q$, we have $\beta+\varsigma \alpha_\sigma\to (1+\varsigma)\beta \mp \varsigma \sigma q$ ($\varsigma, \sigma=\pm$). Note that the poles $s_n^{\mp}$ tend to coincide with the poles $\slashed{s}_n$ in the limit $\omega\to \pm 2\beta q$. In particular, as $\omega\to 2\beta q$ by~\eqref{eqs:prop3-matrix} we can show that
\begin{subequations}\label{eqs:asympt-rels-matrix}
\begin{align}\label{eq:asympt-rels-I}
	\mA^{\pm}_{+}+\mA^{\pm}_{-+}\sim \mA_+^{-\pm}+\mA^{-\pm}_{-+}=\mathcal O(1)~,\quad \mA^{\pm}_{-}+\mA^{\pm}_{--}\sim \mA^{-\pm}_{-} +\mA^{-\pm}_{--}=\mathcal O(1)~.
\end{align}
In contrast, each of the matrix elements $\mA^{\pm}_{+}$, $\mA^{\pm}_{-+}$, $\mA_+^{-\pm}$, $\mA^{-\pm}_{-+}$, $\mA^{\pm}_{-}$, $\mA^{\pm}_{--}$, 
$\mA^{-\pm}_{-}$ and $\mA^{-\pm}_{--}$ participating in~\eqref{eq:asympt-rels-I} is found to have a simple pole at $\omega=2\beta q$. In addition, we obtain
\begin{align}\label{eq:asympt-rels-II}
	\mA^{++}_{\pm}+\mA^{++}_{-\pm}=\mathcal O(1)~,\quad \mA^{+-}_{\pm}+\mA^{+-}_{-\pm}=\mathcal O(1)~;
\end{align}
\end{subequations}
while $\mA^{\varsigma}_{+\pm}=\mathcal O(1)$ and $\mA^{\varsigma\sigma}_{+\pm}=\mathcal O(1)$ with $\mA^{\pm}_{++}\sim \mA^{- \pm}_{++}$ 
and $\mA^{\pm}_{+-}\sim \mA^{- \pm}_{+-}$. Notably,
each of the matrix elements $\mA^{++}_{\pm}$, $\mA^{++}_{-\pm}$, $\mA^{+-}_{\pm}$ and $\mA^{+-}_{-\pm}$ has a simple pole at $\omega=2\beta q$.
Equation~\eqref{eqs:asympt-rels-matrix} is consistent with a dominant-balance argument applied to~\eqref{eqs:sys-disp}  for a function $\breve{\dF}(s)$ that is continuous near the points $\pm q$ and $2\beta\pm q$, in view of $\breve{\dF}(\beta+\varsigma \alpha_\sigma)\to \breve{\dF}\left((1+\varsigma)\beta - \varsigma \sigma q\right)$.

Hence, the determinant $\Lambda(\omega, q)$ is $\mathcal O(1)$ as $\omega\to 2\beta q$, as confirmed by a direct calculation in Appendix~\ref{app:sec:Lambda-lim-q}. By $\Lambda(\omega, q)=\Lambda(-\omega, q)$, we also  have $\Lambda(\omega, q)=\mathcal O(1)$ as $\omega\to -2\beta q$.

\section{Semiclassical regime}
\label{sec:asymptotics}
In this section, we show how~\eqref{eq:disp-SPP} emerges from~\eqref{eq:Lambda-disp-reln} by appropriately expanding each matrix element $\mA^\mu_\nu$ if $0<\tilde{q}^2\ll |\tilde \omega|\ll \tilde{q}$ and $\Cs_0=\tfrac{e^2\eta_0}{2\varepsilon_0\beta^3}\ll 1$ ($\tilde q=\frac{q}{\beta}$, $\tilde{\omega}=\frac{\omega}{\beta^2}$).  We also derive higher-order terms, thus proving Proposition~\ref{prop:strong-b} from Proposition~\ref{prop:disp_reln}.

To organize our calculations, we define the small parameter 
\begin{equation*}
\delta_\sigma := \tilde\alpha_\sigma -1= \sqrt{1+\tq^2-\sigma\tomega}-1=\frac{\tq^2-\sigma \tomega}{\tilde \alpha_\sigma+1}\sim \frac{\tq^2-\sigma\tomega}{2}\left(1-\frac{\tq^2-\sigma\tomega}{4} \right)~,	
\end{equation*}
where $\Rel\,\tilde\alpha_\sigma>0$ and $\tilde\alpha_\sigma=\alpha_\sigma/\beta$ ($\sigma=\pm$).
By virtue of the condition $\Cs_0\ll 1$, in~\eqref{eqs:prop3-matrix} we can replace the power series for each matrix element $\mA^\mu_\nu$  by the respective first terms (for $n=0$).  Recall that $\Lambda(\omega, q)$ is an even function of $\omega$. For later algebraic convenience, let us define the dimensionless quantities
\begin{align}\label{eq:defs-calc}
	\mH_0&=\frac{2\Cs_0}{4\tq^2-\tomega^2}~,\quad 
	\mH^\pm = \frac{\Cs_0/4}{\tq^2\pm \tomega}~,\quad  \wp=\frac{4\tq^3}{\tomega^2-\tq^4}~.
\end{align}
By~\eqref{eq:disp-SPP} or~\eqref{eqs:classical-SP-disp}, the  SP-type dispersion law in the semiclassical limit states that
\begin{equation}\label{eq:critical}
	\mH_0 \wp= \mathcal O(1)~.
\end{equation}
This result is implied by the dominance of $-(\mA^+_+(\omega, q)-1)$ in the determinant of~\eqref{eq:Lambda-disp-reln}.

In our computations, we will leave the factors $(4\tq^2-\tomega^2)^{-1}$ and $(\tq^2\pm \tomega)^{-1}$ intact in the approximate formulas for $\mA^\mu_\nu$. The retainment of these factors in our asymptotic formulas is accurate up to orders $\tfrac{\tomega^2}{\tq^2}$ and $\tfrac{\tq^4}{\tomega^2}$, respectively. We will also approximate linear-in-$\tq$ terms by respective fractions for orders up to $\tq$, setting, e.g., $1-\tq\sim (1+\tq)^{-1}$.

\subsection{Sketch of proof of Proposition~\ref{prop:strong-b}}
\label{subsec:proof-Prop4}
We compute each $\mA^\mu_\nu$ for $0<\tilde{q}^2\ll |\tilde \omega|\ll \tilde{q}$ and $\Cs_0=\tfrac{e^2\eta_0}{2\varepsilon_0\beta^3}\ll 1$. By~\eqref{eq:A-varsigma-varsigmap}, we approximate
\begin{align*}
	\mA^+_+&= \mH_0 
	\left\{-\frac{2\tq(2+\delta_+)}{\tq^2-\tomega}\left[\frac{1}{2+\tq+\delta_+}-\frac{\Cs_0}{4+\tq+\delta_+}\frac{1}{2+\delta_+}
	\left(1+\frac{\tomega/8}{1-\tfrac{\tomega}{4}+\tfrac{\delta_+}{2}}\right) \frac{1}{4-\tomega+2\delta_+}\right]  \right.\\
	&\quad \hphantom{\frac{2\Cs_0}{4\tq^2-\tomega^2}} -\frac{2\tq(2+\delta_-)}{\tq^2+\tomega}\left[\frac{1}{2+\tq+\delta_-}-\frac{\Cs_0}{4+\tq+\delta_-}\frac{1}{2+\delta_-}
	\left(1-\frac{\tomega/8}{1+\tfrac{\tomega}{4}+\tfrac{\delta_-}{2}}\right) \frac{1}{4+\tomega+2\delta_-}\right] \\
	&\quad \hphantom{\frac{2\Cs_0}{4\tq^2-\tomega^2}}  \left. +\frac{1}{1+\tq}\left[1-\frac{1}{8}\frac{\Cs_0}{2+\tq} 
	\frac{1+\tfrac{3\tq}{4}}{\left(1+\tfrac{3\tq}{4}\right)^2-\tfrac{\tomega^2}{64}}\ \frac{1}{1+\tq}\right] +\mathcal O\left((1+\wp)\Cs_0^2\right)\right\} \\
	&= \mH_0\left\{\wp\left(1-\frac{\tq}{2}\right)+\frac{1}{1+\tq} +\mathcal O\left((1+\wp)\Cs_0\right)  \right\}\sim \mH_0 \left(\frac{\wp}{1+\tq/2}+\frac{1}{1+\tq} \right)~.
\end{align*}
We assess that this approximation holds for $\mA^+_+$ if $\tq^{\,7/4}\ll |\tomega|\ll \tq^{\,4/3}$, given the conditional retainment of the factors $(4\tq^2-\tomega^2)^{-1}$ and $(\tomega^2-\tq^4)^{-1}$ which are present in $\mH_0$ and $\wp$. These conditions on $\tomega$ can eventually be relaxed to become $0<\tilde{q}^2\ll |\tilde \omega|\ll \tilde{q}$ for the leading-order formula for $\Lambda(\omega, q)$.  For the remaining elements $\mA^+_{\varsigma(\sigma)}$, we obtain
\begin{align*}
	\mA^+_-&= \mH_0 \left(\frac{1}{1+\tq}+\mathcal O\left(\frac{\tomega^2}{\tq}\right) \right)\sim \frac{\mH_0}{1+\tq}~,\quad \mA^+_{++}=2\mH^-\left(\frac{1}{1+\tq/2}+\mathcal O(\tomega) \right)\sim \frac{2\mH^-}{1+\tq/2}~,\\
	 \mA^+_{-+}&=-\mH_0\left(\frac{1}{1+\tq/2}+\mathcal O(\tomega)\right)\sim -\frac{\mH_0}{1+\tq/2}~, \quad 
	\mA^+_{+-}= 2\mH^+\left(\frac{1}{1+\tq/2}+\mathcal O(\tomega) \right)\sim \frac{2\mH^+}{1+\tq/2}~,\\
	 \mA^+_{--}&=-\mH_0\left(\frac{1}{1+\tq/2}+\mathcal O(\tomega)\right)\sim -\frac{\mH_0}{1+\tq/2}~.
\end{align*}
By symmetry, the residual $\mathcal O(\tomega)$ terms indicated above are expected not to contribute to $\Lambda(\omega, q)$. 

By similar considerations, we proceed to approximately compute the rest of the matrix elements $\mA^\mu_\nu$. Omitting the remainders for simplicity, we write
\begin{align*}
	\mA^-_+&\sim \mH_0 \left(\frac{\wp}{1-\tq/2}+1\right)~,\quad \mA^-_-\sim \mH_0~,\quad \mA^-_{++}\sim \frac{2\mH^-}{1-\tq/2}~,\quad \mA^-_{-+}\sim -\frac{\mH_0}{1-\tq/2}~,\\
	\mA^-_{+-}&\sim \frac{2\mH^+}{1-\tq/2}~,\quad \mA^-_{--}\sim -\frac{\mH_0}{1-\tq/2}~;
\end{align*}
\begin{align*}
	\mA^{++}_+&\sim \frac{\mH_0}{2}\left(\wp+\frac{1}{1+3\tq/4}\right)~,\quad \mA^{++}_{-}\sim \frac{\mH_0}{2}\frac{1}{1+\tq/4}~,\quad \mA^{++}_{++}\sim \mH^-~,\quad \mA^{++}_{-+}\sim -\frac{\mH_0}{2}~,\\
	\mA^{++}_{+-}&\sim \mH^+~,\quad \mA^{++}_{--}\sim -\frac{\mH_0}{2}~;
\end{align*}
\begin{align*}
	\mA^{-+}_+&\sim \mH_0 \left(\wp+\frac{1}{1+\tq/2}\right)~,\quad \mA^{-+}_-\sim \frac{\mH_0}{1+\tq/2}~,\quad \mA^{-+}_{++}\sim 2\mH^-~,\quad \mA^{-+}_{-+}\sim -\mH_0~,\\
	\mA^{-+}_{+-}&\sim 2\mH^+~,\quad \mA^{-+}_{--}\sim -\mH_0~;
\end{align*}
\begin{align*}
	\mA^{+-}_+&\sim \frac{\mH_0}{2}\left(\wp+\frac{1}{1+3\tq/4}\right)~,\quad \mA^{+-}_{-}\sim \frac{\mH_0}{2}\frac{1}{1+\tq/4}~,\quad \mA^{+-}_{++}\sim \mH^-~,\quad \mA^{+-}_{-+}\sim -\frac{\mH_0}{2}~,\\
	\mA^{+-}_{+-}&\sim \mH^+~,\quad \mA^{+-}_{--}\sim  -\frac{\mH_0}{2}~;
\end{align*}
\begin{align*}
	\mA^{--}_+&\sim \mH_0 \left(\wp+\frac{1}{1+\tq/2}\right)~,\quad \mA^{--}_{-}\sim \frac{\mH_0}{1+\tq/2}~,\quad \mA^{--}_{++}\sim 2\mH^-~,\quad \mA^{--}_{-+}\sim -\mH_0~,\\
	\mA^{--}_{+-}&\sim 2\mH^+~, \quad \mA^{--}_{--}\sim -\mH_0~.
\end{align*}

By substituting these formulas into~\eqref{eq:Lambda-disp-reln} of Proposition~\ref{prop:disp_reln}, after some algebra we have
\begin{align*}
	\begin{vmatrix}
	\displaystyle {\mH_0\left(\frac{\wp}{1+\tq/2}+\frac{1}{1+\tq} \right)-1} & \displaystyle \frac{\mH_0}{1+\tq} & \displaystyle \frac{2\mH^-}{1+\tq/2} & \displaystyle -\frac{\mH_0}{1+\tq/2} & \displaystyle \frac{2(\mH^++\mH^-)}{1+\tq/2} & \displaystyle -\frac{2\mH_0}{1+\tq/2} \\
		\displaystyle \mH_0\left(\frac{\wp}{1-\tq/2}+1\right) & \mH_0-1 & \displaystyle \frac{2\mH^-}{1-\tq/2} & \displaystyle -\frac{\mH_0}{1-\tq/2} & \displaystyle \frac{2(\mH^++\mH^-)}{1-\tq/2} & \displaystyle -\frac{2\mH_0}{1-\tq/2} \\
		\displaystyle \frac{\mH_0}{2}\left(\wp+\frac{1}{1+3\tq/4} \right) & \displaystyle \frac{\mH_0}{2}\frac{1}{1+\tq/4} & \mH^--1 & \displaystyle -\frac{\mH_0}{2} & \mH^++\mH^--1 & 
		-\mH_0 \\
		\displaystyle \mH_0\left(\wp+\frac{1}{1+\tq/2}\right) & \displaystyle \frac{\mH_0}{1+\tq/2} & 2\mH^- & -\mH_0-1 & 2(\mH^++\mH^-) & -2\mH_0-1\\
		0 & 0 & 1 & 0 & 0 & 0 \\
		0 & 0 & 0 &  1 & 0 & 0
	\end{vmatrix}\simeq 0~,
\end{align*}
which approximately reduces to
\begin{subequations}\label{eqs:semiclass-asympt}
\begin{align}\label{eq:semiclass-asympt1}
\mH_0(\wp-\tq)-\tfrac{1}{2}\mH_0\wp\tq+(\mH^++\mH^-) \sim 1	
\end{align}
where
\begin{align*}
	\mH_0(\wp-\tq)&=\frac{e^2\eta_0}{\varepsilon_0 \beta^3}\frac{\tq}{4\tq^2-\tomega^2}\left(\frac{4\tq^2}{\tomega^2-\tq^4}-1\right)=\frac{e^2\eta_0}{\varepsilon_0 \beta^3}\frac{\tq}{\tomega^2-\tq^4}[1+\mathcal O(\tq^2)]~,\\
	\mH_0\wp \tq&= \frac{e^2\eta_0}{\varepsilon_0 \beta^3}\frac{\tq}{\tomega^2-\tq^4}\,\frac{4\tq^3}{4\tq^2-\tomega^2}= \frac{e^2\eta_0}{\varepsilon_0 \beta^3}\frac{\tq^2}{\tomega^2-\tq^4} \left[1+\mathcal O\left(\frac{\tomega^2}{\tq^2}\right)\right]~,\\
	\mH^++\mH^-&= -\frac{e^2\eta_0}{4\varepsilon_0 \beta^3}\frac{\tq^2}{\tomega^2-\tq^4}~.
\end{align*}
Note that $\mH_0(\wp-\tq)$ accounts for the leading-order term of expansion~\eqref{eq:semiclass-asympt1}. We thus obtain
\begin{align}\label{eq:semiclass-asympt2}
\frac{\tomega^2}{\tq}\sim \frac{e^2\eta_0}{\varepsilon_0\beta^3}\left(1-\frac{3\tq}{4}\right)\left(1-\frac{\tq^4}{\tomega^2}\right)^{-1}\sim \frac{e^2\eta_0}{\varepsilon_0\beta^3}\left(1-\frac{3\tq}{4}+\frac{\tq^4}{\tomega^2}\right)~,
\end{align}
\end{subequations}
in agreement with~\eqref{eq:leading-order}. This step concludes the sketch of our proof. \hfill $\Box$

\medskip

In~\eqref{eq:semiclass-asympt2}, the higher-order term $\mathcal O(\tq^4/\tomega^2)$ is purely kinematic, because it does not depend on the binding length $\ell_b=\beta^{-1}$ to the leading order, in contrast to the term $\mathcal O(\tq)$. The relative importance of the two higher-order terms appearing in~\eqref{eq:semiclass-asympt2} can be inferred from the comparison of unity to the parameter $\frac{\tilde \omega^2}{\tilde q^3}=\tfrac{\omega^2}{q^3\beta}=\tfrac{\ell_p^3\ell_b}{\ell_{\mathrm{dB}}^4}\sim \tfrac{2^4\ell_b}{\ell_C}\bigl(\tfrac{\ell_p}{\ell_C}\bigr)^2$, where $\tfrac{\ell_b}{\ell_C}\ll 1\ll \tfrac{\ell_p}{\ell_C}$. Notice that the parameter $\tfrac{\ell_b}{\ell_C}\bigl(\tfrac{\ell_p}{\ell_C}\bigr)^2$ is independent of $\hbar$ once the original units are restored, and thus expresses confinement classically.

\section{Conclusion and discussion}
\label{sec:conclusion}
In this paper, we formally derived the dispersion relation for SP-type waves by use of scattering theory for linearized Hartree-type dynamics. The two main ingredients of the governing equation are the repulsive Coulomb interaction and a negative external binding potential localized at one point in the vertical coordinate. The particle is forced to form a bound state near a fixed plane. Our result demonstrates the interplay of basic microscopic scales that include the binding length in the dispersion of plasmon-type wave modes. In the semiclassical limit, we show how a familiar scaling law for the nonretarded-SP energy emerges, and describe a few higher-order terms of the related asymptotic expansion.  

It would be of interest to extend our study to more realistic settings.  A next step is the addition of a family of periodic potentials, e.g., a 2D analog of the Kronig-Penney model, to the unperturbed Hamiltonian. The ensuing scattering problem calls for the use of appropriate Bloch wave functions, and is considered as analytically tractable. To allow for Landau damping in the quantum regime, it is necessary to incorporate the Fermi-Dirac statistics into the model, even at zero temperature.\cite{Nguyen24} This aspect, as well as the effect of ohmic losses, can plausibly be studied by use of the von Neumann equation for the density operator,\cite{Schulz-Baldes1998,Cances2017,Watson2023} albeit in the presence of confinement. We expect that, in the semiclassical limit for low enough dissipation, a dispersion relation of form~\eqref{eq:disp-SPP} should still emerge. On the other hand, the case of the graphene 
SP requires a subtler treatment because of the inherent Dirac dynamics.\cite{Lowetal2017,Hwang-DasSarma07,Wunsch2006,Jablan2013}

A far-reaching mathematical issue not addressed by our analysis is to understand how the ubiquitous Coulomb interaction in 3D affects the envelope of electron motion on an embedded 2D lattice. In the case of the honeycomb lattice in graphene, this interaction is sufficiently screened and can be considered as weak.\cite{GiulianiMastropietro2010,Kotovetal2012,Hwang-DasSarma07} The slow electron motion is then effectively governed by a version of the Dirac equation.\cite{FeffermanWeinstein2012} In more general settings, the description of the effective electron dynamics may be more intricate. 

\section*{ACKNOWLEDGMENTS}

This paper is fondly dedicated to the memory of Tai Tsun Wu. The author is grateful to M.~G. Grillakis, P.-E. Jabin, and C.~D. Levermore for inspiring discussions; and warmly thanks M. Luskin, T. Stauber, E. Tadmor and A.~B. Watson for fruitful conversations. 

\appendix 

\section{Calculation of forward propagator in macroscopic limit}
\label{app:sec:propagator}
In this appendix, we derive~\eqref{eqs:G-TD} by inverting~\eqref{eq:FT-G} for $\widehat G_A$ in the limit $A\to \infty$ via integral formula~\eqref{eq:FT-inv}. By the definition of $G_A$, we must have $G_A(t,\vr_\parallel, z; z')\equiv 0$ if $t<0$. Thus, we only need to consider $t>0$. 

\subsection{Free-particle propagator}
\label{app:subsec:fs-prop}

We examine each of the two terms of formula~\eqref{eq:FT-G} as $A\to \infty$. The first term reads
\begin{equation}\label{app:eq:fs-G-FT}
\widehat{G}^0(w,\vbk,z;z')=\frac{1}{2\alpha} 
 e^{-\alpha |z-z'|}~;\quad \alpha=\sqrt{k^2-w}~,\ w\in\mathbb{C}~, \ \vbk\in \mathbb{R}^2\quad  (\Rel\,\alpha>0)~.
\end{equation}
We will show that this $\widehat{G}^0$ is the Fourier transform in $(t,\vrp)$ of the well-known free-particle propagator $G_{f}(t, \vrp, z-z')$.  The propagator $G_f(t,\vr)$ obeys
\begin{equation}\label{app:eq:fs-G-TD}
(\im\partial_t + \Delta_{\vr})G_{f}(t, \vr)=-\delta(\vr) \,\delta(t)~,\quad (t, \vr)\in \mathbb{R}\times \mathbb{R}^3~,	
\end{equation}
along with the conditions that $G_{f}$ vanish if $t<0$ and be bounded in $\vr$ for fixed $t>0$. Let $\widetilde G_{f}$ be its Fourier transform with respect to $(t, \vr)$. Evidently, by~\eqref{app:eq:fs-G-TD} this $\widetilde{G}_{f}$ equals
\begin{align*}
	\widetilde{G}_{f}(w, \vp)=\left(p^2-w\right)^{-1}\qquad (\vp=(\vp_\parallel, p_z)\in\mathbb{R}^3~,\ p=|\vp|)~.
\end{align*}
The inversion of this formula yields a familiar 4-dimensional integral, which we write as
\begin{align}\label{app:eq:fs-G-TD-Inv}
	G_{f}(t, \vr)=\int_{\mathbb{R}^2}\int_\Gamma\int_{\mathbb{R}} \frac{e^{-\im wt+\im \vp_\parallel\cdot \vrp+\im p_z z}}{(p_\parallel^2+p_z^2)-w}\ \frac{\dv p_z}{2\pi}\,\frac{\dv w}{2\pi}\,\frac{\dv \vp_\parallel}{(2\pi)^2}~;\quad \vp= (p_x, p_y, p_z)=(\vp_\parallel, p_z)~,
\end{align}
where $\Gamma=\{w\in\mathbb{C}\,:\, -\infty < \Rel\,w < \infty,\, \Img\,w=\gamma_1>0\}$ for fixed positive $\gamma_1$.
We carry out the integral in $p_z$, with fixed $(w, \vp_\parallel)$ and $\Rel\sqrt{p_{\parallel}^2-w}>0$. Note that the real axis of the complex $p_z$-plane is free of singularities of the integrand. In view of the simple poles of the integrand at $p_z=\pm\im \sqrt{p_\parallel^2-w}$, we apply the residue theorem and compute
\begin{align*}
	G_{f}(t, \vr)=\frac{1}{2}\int_{\mathbb{R}^2}\int_\Gamma \frac{e^{-\im wt+\im \vp_\parallel\cdot \vrp-\sqrt{p_\parallel^2-w}\, |z|}}{\sqrt{p_\parallel^2-w}}\ \frac{\dv w}{2\pi}\,\frac{\dv \vp_\parallel}{(2\pi)^2}~.
\end{align*}
Hence, the Fourier transform of $G_{f}(t, \vr)$ with respect to the reduced variable $(t, \vrp)$ is given by~\eqref{app:eq:fs-G-FT} for $\vbk=\vp_\parallel$ and $z'=0$. 

By first integrating in $w$, we now write~\eqref{app:eq:fs-G-TD-Inv} as
\begin{align}\label{app:eq:free-prop-fin}
	G_{f}(t, \vr)&=\int_{\mathbb{R}^3}\int_\Gamma \frac{e^{-\im wt+\im \vp\cdot \vr}}{(p_\parallel^2+p_z^2)-w}\ \frac{\dv w}{2\pi}\,\frac{\dv \vp}{(2\pi)^3}=\im \int_{\mathbb{R}^3}\frac{\dv \vp}{(2\pi)^3}\,e^{\im\vp\cdot\vr-\im p^2 t}\notag\\
	&=\im \prod_{s=x,y,z}\left(\int_{\mathbb{R}}\frac{\dv p_s}{2\pi}\,e^{\im p_s s-\im p_s^2 t}\right)=\frac{\im}{\pi^3}e^{\im \tfrac{r^2}{4t}}
	\left(e^{-\im\pi/4}\int_0^\infty \dv \xi \,e^{-t\xi^2} \right)^3\notag\\
	&= \frac{\im}{\pi^3} e^{-3\im\pi/4}\left(\frac{\sqrt{\pi}}{2}\frac{1}{\sqrt{t}}\right)^3 e^{\im \tfrac{r^2}{4t}}=e^{-\im\pi/4} (4\pi t)^{-3/2} e^{\im \tfrac{r^2}{4t}}~.
	\end{align}
In the above, we rotated the integration path (real axis) for $p_s$ ($s=x, y, z$) by $-\pi/4$. Hence, we recover~\eqref{eq:G-TD-fs}.

\subsection{Contribution due to particle binding}
\label{app:subsec:corrector}
In the limit $A\to\infty$, the second term of~\eqref{eq:FT-G} reads
\begin{align}\label{app:eq:FT-corrector}
\widehat{G}^c(w, \vbk, z; z')=\frac{\beta}{2\alpha (\alpha-\beta)} 
 e^{-\alpha(|z|+|z'|)}~;\quad \alpha=\alpha(w,\vbk)=\sqrt{k^2-w}~,\quad \beta=\frac{V_0a}{2}~,
 \end{align}
with $(w, \vbk)\in \mathfrak R=\{(w, \vbk)\in \mathbb{C}\times \mathbb{R}^2:\,\alpha\neq \beta~,\ \Rel\,\alpha>0\}$.
We need to compute the inverse Fourier transform according to the iterated integral 
\begin{align}\label{app:eq:Gcorr-inv}
G^c(t, \vrp, z; z')=\int_{\mathbb{R}^2}\int_\Gamma 	e^{-\im w t+\im \vbk\cdot \vrp} \widehat{G}^c(w, \vbk, z; z') \,\frac{\dv w}{2\pi} \frac{\dv \vbk}{(2\pi)^2}~.
\end{align}


We first carry out the integration in $w$. For fixed $k=|\vbk|$, the integral of interest is
\begin{align}\label{eq:I-integral-def}
	I(t,Z)=\int_\Gamma \frac{\dv w}{2\pi}\ e^{-\im w t-\alpha(w,\vbk)Z}\frac{1}{\alpha(w,\vbk)}\frac{\beta}{\alpha(w,\vbk)-\beta}~,\quad t>0~,\ Z=|z|+|z'|>0~;
\end{align}
see Fig.~\ref{fig:integration}. We  deform the original integration path $\Gamma$ to path $\Gamma'$ which is wrapped around the cut $[k^2, +\infty)$ in the real axis. In this way, the corresponding contour integral picks up the residue of the simple pole at $w=w_p=k^2-\beta^2$, for which $\alpha(w, \vbk)=\beta$. Subsequently, we shift $w$ according to $w\mapsto \varpi=w-k$. Thus, we rewrite integral~\eqref{eq:I-integral-def} as
\begin{align}\label{app:eq:Iintegral-def}
I(t,Z)&= -2\pi \im \, \mathop{\mathrm{Res}}\limits_{w=w_p}\left\{ \frac{1}{2\pi}\ e^{-\im w t-\alpha(w,\vbk)Z}\frac{1}{\alpha(w,\vbk)}\frac{\beta}{\alpha(w,\vbk)-\beta}\right\}\notag\\
&\qquad +\int_{\Gamma'} \frac{\dv w}{2\pi}\ e^{-\im w t-\alpha(w,\vbk)Z}\frac{1}{\alpha(w,\vbk)}\frac{\beta}{\alpha(w,\vbk)-\beta}
\notag\\
&=2\im \beta e^{-\im (k^2-\beta^2) t -\beta Z}
-2\im \beta e^{-\im k^2 t} \left(-\frac{\partial}{\partial Z}+\beta\right) I_1(t, Z)~,
\end{align}
where $I_1(t, Z)$ is expressed by the (convergent) integral
\begin{align}\label{app:eq:I1-int-def}
I_1(t, Z)=
\int_0^\infty \frac{\dv\varpi}{2\pi}\ \frac{e^{-\im \varpi t}}{\sqrt{\varpi}}\,\frac{\cos(Z\sqrt{\varpi})}{\varpi+\beta^2}~.
\end{align}
Note that $|I_1(t, Z)|\le M\beta^{-1}$ for all $(t, Z)\in\mathbb{R}^2$, for some constant $M>0$. We can also show that $|I_1|=\mathcal O(Z^{-2})$ as $|Z|\to +\infty$.

The key task is to calculate $I_1(t, Z)$. Evidently, this $I_1$ satisfies the differential equation
\begin{align*}
	\frac{\partial^2 I_1}{\partial Z^2}-\beta^2 I_1
	=-\frac{1}{2\sqrt{\pi t}}e^{-\im\pi/4} e^{\im \tfrac{Z^2}{4t}}~,\quad Z\in\mathbb{R}~.
\end{align*}
By variation of parameters, we find that the general solution of this equation is written as
\begin{align*}
	I_1(t, Z)=K_1(t)\,e^{-\beta Z}+K_2(t)\,e^{\beta Z}+ \frac{e^{-\im\pi/4}}{\sqrt{4\pi \beta^2 t}}\int_{Z}^{+\infty e^{\im\phi}} \dv\xi\,\sinh[\beta(Z-\xi)]\,e^{\im\tfrac{\xi^2}{4t}}~,\quad  0<\phi<\pi/2~,
\end{align*}
for arbitrary integration constants $K_{1,2}(t)$. Now take $Z>0$. Because $I_1$ is bounded 
as $Z\to +\infty$, we need to set $K_2(t)\equiv 0$. By~\eqref{app:eq:I1-int-def}, we impose $\partial_Z I_1=0$ at $Z=0$, and compute
\begin{equation*}
	K_1(t)= e^{-\im\pi/4} (4\pi\beta^2 t)^{-1/2}\int_0^{+\infty e^{\im\phi}} \dv\xi\,\cosh(\beta\xi)\,e^{\im \tfrac{\xi^2}{4t}}\qquad (0<\phi<\pi/2)~.
\end{equation*}
Thus, we obtain the formula
\begin{align*}
	I_1(t, Z)=\frac{1}{2\beta}e^{\im \beta^2 t-\beta Z}+\frac{1}{\sqrt{4\pi\beta^2 t}} e^{-\im\pi/4}\int_{Z}^{+\infty e^{\im \phi}} \dv\xi\,\sinh[\beta(Z-\xi)] e^{\im \tfrac{\xi^2}{4t}}\qquad (0<\phi<\pi/2)~,
\end{align*}
for $Z>0$. The last integral is expressed in terms of the complementary error function via the decomposition of $\sinh(\cdot)$ into exponentials. After some algebra, $I_1$ becomes
\begin{align*}
	I_1(t, Z)&=\frac{1}{2\beta} e^{\im\beta^2 t-\beta Z} \notag\\
	&+\frac{1}{4\beta}e^{\im \beta^2 t}\left\{ e^{\beta Z}\,\mathrm{erfc}\biggl(e^{-\im\pi/4}\frac{Z}{2\sqrt{t}}+e^{\im\pi/4}\beta\sqrt{t}\biggr)-e^{-\beta Z}\,\mathrm{erfc}\biggl(e^{-\im\pi/4}\frac{Z}{2\sqrt{t}}-e^{\im\pi/4}\beta\sqrt{t}\biggr)\right\}\notag\\
	&=\frac{e^{\im\beta^2 t}}{4\beta}\left\{ e^{\beta Z}\,\mathrm{erfc}\biggl(e^{\im\pi/4}\beta\sqrt{t}+e^{-\im\pi/4}\frac{Z}{2\sqrt{t}}\biggr)+e^{-\beta Z}\,\mathrm{erfc}\biggl(e^{\im\pi/4}\beta\sqrt{t}-e^{-\im\pi/4}\frac{Z}{2\sqrt{t}}\biggr)\right\}~.
\end{align*}
This formula can be extended to all $Z\in\mathbb{R}$. Note that $I_1(t, -Z)=I_1(t, Z)$; cf.~\eqref{app:eq:I1-int-def}. 

Now we compute $I(t, Z)$ by use of~\eqref{app:eq:Iintegral-def}, if $Z>0$. A direct calculation yields
\begin{align}
I(t, Z)=I(t, Z; k)=\im \beta e^{-\im (k^2-\beta^2) t-\beta Z} \,\mathrm{erfc}\biggl(e^{-\im\pi/4}\frac{Z}{2\sqrt{t}}-e^{\im\pi/4}\beta\sqrt{t}\biggr), \ t>0~,\ Z>0~.	
\end{align}

Finally, to obtain $G^c(t, \vrp, z; z')$ we need to perform the integration with respect to the lateral-momentum variable, $\vbk$. By~\eqref{app:eq:Gcorr-inv}, we compute
\begin{align}\label{app:eq:prop-correction}
G^c(t, \vrp, z; z')&=\frac{1}{2}\int_{\mathbb{R}^2} e^{\im\vbk\cdot \vrp} I(t, Z; k)\,\frac{\dv \vbk}{(2\pi)^2}\qquad (Z=|z|+|z'|)\notag\\
&=\frac{\im\beta}{2}e^{\im\beta^2 t-\beta (|z|+|z'|)}\,\mathrm{erfc}\biggl(e^{-\im\pi/4}\frac{|z|+|z'|}{2\sqrt{t}}-e^{\im\pi/4}\beta\sqrt{t}\biggr)\int_{\mathbb{R}^2} \frac{\dv \vbk}{(2\pi)^2}\,e^{\im \vbk \cdot \vrp-\im k^2 t}\notag\\
&=\frac{\beta}{8\pi t} e^{\im \tfrac{r_\parallel^2}{4t}}e^{\im \beta^2 t-\beta(|z|+|z'|)}\,\mathrm{erfc}\biggl(e^{-\im\pi/4}\frac{|z|+|z'|}{2\sqrt{t}}-e^{\im\pi/4}\beta\sqrt{t}\biggr)~,\quad t>0~,
\end{align}
for all $(z, z')\in\mathbb{R}^2$. The above procedure exemplifies the role of the branch cut. Alternatively, one may derive this result for $G^c$ from~\eqref{eq:I-integral-def} via decomposing the corresponding integrand  into partial fractions. We omit the details of this scenario. 

In conclusion, the forward propagator in the macroscopic limit equals
\begin{align*}
	G(t, \vrp, z; z')= G_{f}(t, \vrp, z-z')+ G^c(t, \vrp, z; z')~.
\end{align*}
By~\eqref{app:eq:free-prop-fin} and~\eqref{app:eq:prop-correction},
we obtain formula~\eqref{eq:G-TD-total}.

\section{Properties of $\breve\dF(s)$ by functional equation}
\label{app:sec:prop-funct-eq}
In this appendix, we discuss properties of $\breve \dF(s)$ by recourse to~\eqref{eq:funct-eq}. We assume that $F\in L^1(\mathbb{R}_+)\cap L^2(\mathbb{R}_+)$, where $F(z)=e^{\beta z}\dF(z)$, by which $\breve\dF(s)$ is analytic for $\Rel\,s> -\beta$. We will exploit the right-hand side of~\eqref{eq:funct-eq} to analytically continue $\breve\dF(s)$ into the left half $s$-plane ($\Rel\,s<0$), and describe the singularities of $\breve\dF(s)$.  For algebraic convenience, we primarily consider $0< q^2<\omega< q\beta$; and comment on the extension of our results to $\omega\in\mathbb{C}$ with $0<q^2<|\omega|< q\beta$.

By inspecting the right-hand side of~\eqref{eq:funct-eq}, in view of the $\breve\dF(s+2\beta)$ term we can claim that $\breve \dF(s)$ is analytic for $\Rel\,s>-3\beta$ with the possible exception of points at which the coefficients of this functional equation are singular. The analyticity of $\breve\dF(s)$ in the left half $s$-plane is discussed in Appendices~\ref{app:subsec:remov-sing}--\ref{app:subsec:poles-res}. Note that if $0<q^2<\omega<q\beta$, we have 
\begin{align}\label{app:eq:inequal}
	q > \alpha_- -\beta > \beta-\alpha_+>0\qquad (\alpha_->\alpha_+ >0)~.
\end{align}

\subsection{Removable singularities of $\breve\dF(s)$ at $s=-\beta+\alpha_\pm$ and $s=-2\beta+q$}
\label{app:subsec:remov-sing}
We will show that $s=-\beta+\alpha_\sigma$ ($\sigma=\pm$) 
and $s=-2\beta+q$ are regular points of $\breve\dF(s)$. This task is an interesting exercise because, by a mere inspection, the coefficients of terms linear with $\breve\dF$ in functional equation~\eqref{eq:funct-eq} have simple poles at the above points. Let us now rewrite~\eqref{eq:funct-eq} in a way that reveals the true character of these points in regard to the function $\breve\dF(s)$. After some algebra, we have 
\begin{align}\label{app:eq:funct-eq-alt}
	\breve\dF(s)&=\frac{\beta}{2q}\frac{e^2 \eta_0}{\varepsilon_0}\sum_{\sigma=\pm}\left\{\frac{C_{1\sigma}\breve\dF(q)+C_{2\sigma}\breve\dF(\beta+\alpha_\sigma)+C_{3\sigma}^-\breve\dF(s+2\beta)}{s+\beta+\alpha_\sigma}-C_{4\sigma}^+\frac{\breve\dF(s+2\beta)+\breve\dF(q)}{s+2\beta+q} \right. \notag\\
	&\left. \hphantom{\frac{\beta}{2q}\frac{e^2 \eta_0}{\varepsilon_0}} +C_{4\sigma}^- \frac{\breve\dF(s+2\beta)-\breve\dF(q)}{s+2\beta-q}-C_{3\sigma}^+\frac{\breve\dF(s+2\beta)-\breve\dF(\beta+\alpha_\sigma)}{s+2\beta-(\beta+\alpha_\sigma)} \right\}~,
\end{align}
where
\begin{align*}
	C_{1\sigma}&=\frac{4q^2\beta}{(\alpha_\sigma-\beta)[(\alpha_\sigma+\beta)^2-q^2][(\alpha_\sigma-\beta)^2-q^2]}~,\quad C_{2\sigma}=\frac{q}{\alpha_\sigma}\,\frac{1}{(\alpha_\sigma+\beta)^2-q^2}\frac{\alpha_\sigma+\beta}{\alpha_\sigma-\beta}~,\notag\\
	C_{3\sigma}^{\varsigma}&=\frac{q}{\alpha_\sigma} \frac{1}{(\beta+\varsigma\alpha_\sigma)^2-q^2}~,\quad C_{4\sigma}^\varsigma= -\frac{1}{(\beta+\varsigma q)^2-\alpha_\sigma^2}\qquad (\varsigma,\sigma=\pm)~.
\end{align*}	
Notice the terms in the second line of~\eqref{app:eq:funct-eq-alt}.
The analyticity of $\breve\dF(s+2\beta)$ for $\Rel\,s> -3\beta$ entails the removal of the singular terms $(s+2\beta-q)^{-1}$ and $(s+\beta-\alpha_\pm)^{-1}$.

\subsection{Analyticity of $\breve\dF(s)$ for $\Rel\,s>-\beta-\min\{\Rel\,\alpha_+, \Rel\,\alpha_-\}$}
\label{app:subsec:reg-analt}
Next, we argue that $\breve\dF(s)$ is holomorphic for every $s\in\mathbb{C}$ with $\Rel\,s>-\beta-\min\{\Rel\,\alpha_+, \Rel\,\alpha_-\}$. For this purpose, we invoke functional equation~\eqref{eq:funct-eq}.

Let us first discuss the behavior of $\breve \dF(s)$ in the vicinity of the points $s=-q$ and $s=\beta-\alpha_\pm$. These points may be considered as special because at least two of them lie in the left half $s$-plane for real $\omega$, and the values $\breve \dF(\pm q)$ and $\breve \dF(\beta\pm \alpha_\sigma)$ (with $\sigma=\pm$) are employed as auxiliary variables in our derivation of the SP-type dispersion relation (Sec.~\ref{subsec:disp-soln}). Note that, by virtue of the assumed analyticity of $\breve\dF(s)$ for $\Rel\,s>-\beta$, the function $\breve \dF(s)$ is analytic in the vicinity of $s=q$ and $s=\beta+\alpha_\pm$ ($\Rel\,\alpha_\pm >0$). Evidently, the right-hand side of~\eqref{eq:funct-eq} is analytic in the vicinity of $s=-q$, since $\breve\dF(s+2\beta)$ is analytic for $\Rel\,s> -3\beta$. Thus, $\breve \dF(s)$ itself is analytic, and therefore bounded, in a neighborhood of $s=-q$. Ditto for $\breve\dF(s)$ at each of the points $s=\beta-\alpha_\pm$.

If $0<q^2<\omega<q\beta$, by~\eqref{eq:funct-eq} these conclusions hold for any point $s=s_{\#} \in \mathbb{C}$ with $\Rel\,s_{\#}>-\alpha_+-\beta$. A key observation is that $\breve\dF(s+2\beta)$ is analytic at every such point $s_{\#}$. Furthermore, there is no point in the region $\{\Rel\,s>-\alpha_+-\beta\}$ that happens to be a singularity of any coefficient in~\eqref{eq:funct-eq} with the exception of removable singularities (Appendix~\ref{app:subsec:remov-sing}). If $\omega\in\mathbb{C}$ with $0<q^2<|\omega|< q\beta$, we can repeat these considerations for $\Rel\,s>-\beta-\min\{\Rel\,\alpha_+, \Rel\,\alpha_-\}$.

\subsection{Poles and residues of $\breve\dF(s)$}
\label{app:subsec:poles-res}
Next, we locate and characterize as simple poles all singularities of $\breve\dF(s)$ in the left half $s$-plane (for $\Rel\,s<0$); and compute the associated residues in closed forms. Suppose that $\omega$ is real and $0<q^2<\omega< q\beta$. Note that $\breve\dF(s)$ is holomorphic if 
$\Rel\,s> -q-\epsilon $ for sufficiently small $\epsilon>0$; in fact, $0<\epsilon\le \alpha_++\beta-q$ (Appendix~\ref{app:subsec:reg-analt}). 
For definiteness, we first consider $\breve\dF(s)$ in the strip $S^\epsilon_{2\beta}:=\{s\in\mathbb{C}\,|\,-2\beta-q-\epsilon \le  \Rel\,s \le  -q-\epsilon\}$ where $0<\epsilon\ll\beta$, which has width $2\beta$ and is shifted by $-\epsilon$ with respect to the strip $\{-2\beta-q\le \Rel\,s\le -q\}$; and subsequently study $\breve\dF(s)$ for $s$ in $S^{\epsilon+2n\beta}_{2\beta}$ ($n=1\,,2\,,\ldots$) by induction. In our procedure, we will make use of~\eqref{eq:funct-eq} or~\eqref{app:eq:funct-eq-alt}. Our conclusions can be extended to complex $\omega$ with $0< q^2 <|\omega| < q\beta$.

\subsubsection{Poles}
\label{app:sssec:poles}

First, let us examine every point $s_{\#}\in S^\epsilon_{2\beta}$ with $s_{\#}\neq -2\beta-q,\,-\beta-\alpha_\pm$. Set $s=s_{\#}+\varpi_1$ on the right-hand side of~\eqref{app:eq:funct-eq-alt}. Notice that the resulting expression as a function of the complex variable $\varpi_1$ is analytic at $\varpi_1=0$,  admitting a Maclaurin series for  $|\varpi_1|<\epsilon'$ and small enough $\epsilon'$ ($\epsilon'>0$), under restrictions such as $\epsilon'\le \min\{|s_{\#}-(-2\beta-q)|, |s_{\#}-(-\beta-\alpha_+)|, |s_{\#}-(-\beta-\alpha_-)|\}$ and $\epsilon'\ll\beta$. Hence, $\breve\dF(s)$ is analytic at any point $s_{\#}\in S^\epsilon_{2\beta}\setminus \{-2\beta-q,\,-\beta-\alpha_-, -\beta-\alpha_+\}$.

Now consider any point $s_{\#}\in \{-2\beta-q, -\beta-\alpha_-, -\beta-\alpha_+\}$ in $S^{\epsilon}_{2\beta}$, and set $s=s_{\#}+\varpi_1$ on the right-hand side of~\eqref{app:eq:funct-eq-alt}. Because $\varpi_1=0$ is a simple pole of a coefficient but a regular point of $\breve\dF(s+2\beta)$ in this expression, we infer that $\breve\dF(s)$ admits a Laurent series of the form
\begin{align}\label{app:eq:Laurent}
	\breve\dF(s)=\frac{R_{\#}}{s-s_{\#}}+\sum_{m=0}^\infty c_m \varpi_1^m~,\quad 0<|\varpi_1|< \epsilon'~,
\end{align}
for sufficiently small $\epsilon'>0$. The residue, $R_{\#}$, at $s_\#$ is in principle nonzero and depends on $s_{\#}$. Thus, every point $s_{\#}\in \{-2\beta-q, -\beta-\alpha_-, -\beta-\alpha_+\}$ in $S^{\epsilon}_{2\beta}$ is a simple pole of $\breve\dF(s)$.

Let us now extend these considerations to the shifted strip $S^{\epsilon+2\beta}_{2\beta}$. Take any point $s_{\#}\in S^{\epsilon+2\beta}_{2\beta}\setminus\{-4\beta-q,\,-3\beta-\alpha_+, -3\beta-\alpha_-\}$. As above, we set $s=s_{\#}+\varpi_1$ on the right-hand side of~\eqref{app:eq:funct-eq-alt}. The ensuing function of $\varpi_1$ is analytic at $\varpi_1=0$,  since the term $\breve\dF(s+2\beta)$ and all coefficients are so by virtue of our preceding conclusions for $\breve\dF(s)$ in the strip $S^\epsilon_{2\beta}$. Hence, $\breve\dF(s)$ is analytic at any point $s_{\#}\in S^{\epsilon+2\beta}_{2\beta}\setminus \{-4\beta-q,\,-3\beta-\alpha_-, -3\beta-\alpha_+\}$.

Now take any point $s_{\#}\in \{-4\beta-q, -3\beta-\alpha_-, -3\beta-\alpha_+\}$ in $S^{\epsilon+2\beta}_{2\beta}$, and set $s=s_{\#}+\varpi_1$ on the right-hand side of~\eqref{app:eq:funct-eq-alt}. Notice that $\varpi_1=0$ is a regular point of the coefficients but a simple pole of $\breve\dF(s+2\beta)$ in this expression, in view of our results for strip $S^\epsilon_{2\beta}$.
Hence, $\breve\dF(s)$ admits a Laurent series of form~\eqref{app:eq:Laurent} at $s=s_\#$. We conclude that every point $s_{\#}\in \{-4\beta-q, -3\beta-\alpha_-, -3\beta-\alpha_+\}$ in $S^{\epsilon+2\beta}_{2\beta}$ is a simple pole of $\breve\dF(s)$.

By induction, we extend these results to $S^{\epsilon+2n\beta}_{2\beta}$ for every $n\ge 2$. In particular, we initially suppose that $\breve\dF(s)$ is analytic at any point $s_{\#}\in S^{\epsilon+2(n-1)\beta}_{2\beta}\setminus \{-2n\beta-q,\,-(2n-1)\beta-\alpha_+, -(2n-1)\beta-\alpha_-\}$ and has the simple poles $\slashed{s}_{n-1}=-2n\beta-q$ and $s_{n-1}^\sigma=-(2n-1)\beta-\alpha_\sigma$ in  $S^{\epsilon+2(n-1)\beta}_{2\beta}$, for some $n\ge 2$ ($\sigma=\pm$). By functional equation~\eqref{eq:funct-eq} or~\eqref{app:eq:funct-eq-alt}, we can show that $\breve\dF(s)$ must have the respective properties within the strip $S^{\epsilon+2n\beta}_{2\beta}$. 

The above results on the analyticity and singularities of $\breve\dF(s)$ hold if $\omega$ is complex with $0<q^2<|\omega|< q\beta$. In this case, we have $0<\epsilon\le \min\{\Rel\,\alpha_+, \Rel\,\alpha_-\}+\beta-q$ regarding $S^\epsilon_{2\beta}$.

\subsubsection{Computation of residues}
\label{app:sssec:residues}
Next, we provide explicit formulas for the residues, $R_n^\sigma$ and $\slashed{R}_n$, of $\breve\dF(s)$ at the poles $s=s_n^\sigma$ and $\slashed{s}_n$ ($n\in\mathbb{N}$, $\sigma=\pm$). The main idea is to obtain recursion relations for $R_n^\sigma$ and $\slashed{R}_n$ by inspection of functional equation~\eqref{eq:funct-eq} or~\eqref{app:eq:funct-eq-alt}. These relations will directly furnish the desired residues for any $n$. We routinely define
\begin{align*}
	R_n^\sigma:= \lim_{s\to s_n^\sigma}\{(s-s_n^\sigma)\breve\dF(s)\}~,\quad \slashed{R}_n:=\lim_{s\to \slashed{s}_n}\{(s-\slashed{s}_n)\breve\dF(s)\}\qquad (n\in\mathbb{N},\,\sigma=\pm)~.
\end{align*}

\medskip

(i) {\em Residues $R_n^\sigma$.} For $n=0$, by~\eqref{eq:funct-eq} or~\eqref{app:eq:funct-eq-alt} we collect all singular terms of $\breve\dF(s)$ that pertain to its poles at $s=s_0^{\pm}\in S^\epsilon_{2\beta}$. By inspecting the right-hand side of~\eqref{eq:funct-eq}, we see that such terms come solely from the displayed partial fractions. The result is
\begin{align}\label{app:eq:resR0}
	R_0^\sigma&= \frac{\beta e^2 \eta_0}{2\varepsilon_0}\left\{\frac{4q\beta \breve\dF(q)}{(\alpha_\sigma-\beta)[(\alpha_\sigma+\beta)^2-q^2][(\alpha_\sigma-\beta)^2-q^2]}+\frac{1}{\alpha_\sigma}\left[\frac{\alpha_\sigma+\beta}{\alpha_\sigma-\beta} \frac{\breve\dF(\beta+\alpha_\sigma)}{(\alpha_\sigma+\beta)^2-q^2} \right.\right.\notag\\
	& \hphantom{\frac{\beta e^2 \eta_0}{2\varepsilon_0}} \left.\left. \qquad 
	+ \frac{\breve\dF(\beta-\alpha_\sigma)}{(\alpha_\sigma-\beta)^2-q^2} \right]\right\}~,
\end{align}
which is~\eqref{eq:resR0}.

Now consider $R_n^\sigma$ if $n\ge 1$. All singular terms of $\breve\dF(s)$ that amount to its poles at $s=s_n^{\pm}\in S^{\epsilon+2n\beta}_{2\beta}$ come from the factor $\breve\dF(s+2\beta)$ on the right-hand side of~\eqref{eq:funct-eq}. By the Laurent expansions of $\breve\dF(\zeta)$ near $\zeta=s_{n-1}^{\pm}$, we derive the recursion relation
\begin{align*}
	R_n^\sigma&= \frac{\beta}{2q} \frac{e^2 \eta_0}{\varepsilon_0}\sum_{\sigma'=\pm}\sum_{\varsigma=\pm 1}\varsigma
	\left\{\frac{1}{(\beta+\varsigma q)^2-\alpha_{\sigma'}^2}\frac{1}{s_n^\sigma+2\beta+\varsigma q} \right. \\
	&\left. \hphantom{\frac{\beta}{2q}\frac{e^2 \eta_0}{\varepsilon_0}}+
	\frac{q}{\alpha_{\sigma'}}
	\frac{1}{(\beta-\varsigma \alpha_{\sigma'})^2-q^2} \frac{1}{s_n^\sigma+\beta+\varsigma\alpha_{\sigma'}}\right\} R_{n-1}^\sigma\qquad (n\ge 1~,\  \sigma=\pm)~.
\end{align*}
Carrying out the summation over $\varsigma$, we write the preceding relation as
\begin{align*}
	R_n^\sigma&=-\frac{\beta e^2 \eta_0}{\varepsilon_0}\frac{1}{[\alpha_\sigma+(2n-1)\beta]^2-q^2}\left\{\sum_{\sigma'=\pm}\frac{1}{(\alpha_\sigma+2n\beta)^2-\alpha_{\sigma'}^2}\right\}R_{n-1}^\sigma\qquad (n\ge 1~,\  \sigma=\pm)~.
\end{align*}
Hence, we assert that 
$R_n^\sigma=\Lambda_n^\sigma R_0^\sigma$ (all $n\in\mathbb{N}$). If $n=0$
we have $\Lambda_0^\sigma=1$; and if $n\ge 1$ then $\Lambda_n^\sigma$ is given 
by~\eqref{eq:resRn-Lambda-n}.
According to~\eqref{app:eq:resR0}, the residue $R_0^\sigma$ is a linear combination of three values of $\breve{\dF}(s)$, namely, $\breve\dF(q)$ and $\breve\dF(\beta\pm \alpha^\sigma)$, for $\sigma=\pm$. These values should be determined self consistently as parts of the solution, $\breve\dF(s)$, of functional equation~\eqref{eq:funct-eq}; see Sec.~\ref{subsec:disp-soln}.

\medskip

(ii) {\em Residues $\slashed{R}_n$.} For $n=0$, we focus on the singular terms of $\breve\dF(s)$ due to its pole at $s=\slashed{s}_0\in S^\epsilon_{2\beta}$. These terms solely come from partial fractions that are displayed on the right-hand side of~\eqref{eq:funct-eq}. We thus compute the residue
\begin{align}\label{app:eq:slashed-resR0}
\slashed{R}_0&=-\frac{\beta}{2q} \frac{e^2 \eta_0}{\varepsilon_0}\sum_{\sigma=\pm}\frac{\breve\dF(q)+\breve\dF(-q)}{\alpha_\sigma^2-(\beta+q)^2}	\notag\\
&= \frac{2e^2\eta_0}{\varepsilon_0}\frac{\beta^2}{4q^2\beta^2-\omega^2}\{\breve\dF(q)+\breve\dF(-q) \}~,
\end{align}
which is~\eqref{eq:slashed-resR0}.

Consider $\slashed{R}_n$ for $n\ge 1$. All singular terms of $\breve\dF(s)$ due to its pole at $s=\slashed{s}_n\in S^{\epsilon+2n\beta}_{2\beta}$ come from the factor $\breve\dF(s+2\beta)$, on the right-hand side of~\eqref{eq:funct-eq}. By using the Laurent expansion of $\breve\dF(\zeta)$ at $\zeta=\slashed{s}_{n-1}$, we obtain the recursion relation
\begin{align*}
	\slashed{R}_n&=\frac{\beta}{2q}\frac{e^2\eta_0}{\varepsilon_0}\sum_{\sigma=\pm}\sum_{\varsigma=\pm 1}\varsigma \left\{\frac{1}{(\beta+\varsigma q)^2-\alpha_\sigma^2}\frac{1}{\slashed{s}_n+2\beta+\varsigma q}\right.\\
	&\left. \hphantom{\frac{\beta}{2q}\frac{e^2\eta_0}{\varepsilon_0}}+
	\frac{q}{\alpha_\sigma}\frac{1}{(\beta-\varsigma\alpha_\sigma)^2-q^2}\frac{1}{\slashed{s}_n+\beta+\varsigma \alpha_\sigma }\right\}\slashed{R}_{n-1}~.
\end{align*}
This relation is simplified to
\begin{align*}
	\slashed{R}_n=-\frac{e^2 \eta_0}{4\varepsilon_0}\frac{1}{n(q+n\beta)}\sum_{\sigma=\pm}
	\frac{1}{[(2n+1)\beta+q]^2-\alpha_\sigma^2}\slashed{R}_{n-1}~,\quad n\ge 1~.
\end{align*}
Thus, we obtain the explicit formula $\slashed{R}_n=\slashed{\Lambda}_n \slashed{R}_0$ ($n\in\mathbb{N}$) where $\slashed{\Lambda}_0=1$ and $\slashed{\Lambda}_n$ is furnished by~\eqref{eq:slashed-resRn-Lambda-n} for $n\ge 1$.
By~\eqref{app:eq:slashed-resR0}, the residue $\slashed{R}_0$ is proportional to $\breve\dF(q)+\breve\dF(-q)$, which should be determined self consistently as a part of the solution $\breve{\dF}(s)$ (Sec.~\ref{subsec:disp-soln}).

\subsection{On the limit of $\breve\dF(s)$ as $s\to \infty$ in complex plane}
\label{app:subsec:lim-F}
Consider the region of the $s$-plane with $\Rel\,s \ge  -\beta$. By the integral of~\eqref{eq:LT-dir}, we infer that $\breve \dF(s)$ uniformly tends to $0$ as $|s+\beta|\to +\infty$ 
with $|\Arg(s+\beta)|< \tfrac{\pi}{2}-\delta_1$ for arbitrarily small but fixed $\delta_1$ ($\delta_1>0$), since $e^{\beta z}\dF(z)$ is integrable on $\mathbb{R}_+$.\cite{Widder-book}   
Similarly, by~\eqref{eq:LT-dir} we can conclude that $\breve\dF(s)\to 0$ uniformly as $|\Img\,s|\to +\infty$, for fixed $\Rel\,s$ ($\Rel\,s\ge -\beta$).\cite{Widder-book}

On the other hand, if $\Rel\,s< -\beta$ we use functional equation~\eqref{eq:funct-eq}. Consider real $\omega$ with $0<q^2<\omega< q\beta$. Equation~\eqref{eq:funct-eq-g} reads
\begin{align}\label{app:eq:funct-eq-PQ}
	\breve\dF(s)= P(s)+Q(s)\, \breve \dF(s+2\beta)~,
\end{align}
for $s\neq s_n^\pm~,\, \slashed{s}_n$ (all $n\in\mathbb{N}$).
In the above, $P(s)$ and $Q(s)$ are meromorphic functions expressed as finite sums of  partial fractions, each of which is of the form $C_{\#}(s-s_\#)^{-1}$ where the pole $s_{\#}$ is in the negative real axis; $P(s)\to 0$ and $Q(s)\to 0$ as $s\to \infty$. Note that $P(s)$ is linear with the values $\breve{\dF}(q)$ and $\breve{\dF}(\beta+\alpha_\pm)$. Hence, for every (preferably small) number $\epsilon>0$, there exists a (sufficiently large) number $M>0$ so that 
\begin{align*}
	|P(s)|< \epsilon\, f_P~\quad \mbox{and}\quad |Q(s)|<\epsilon\quad \mbox{if}\quad  |\Img\,s|>M~,
\end{align*}
with fixed $\kappa=\Rel\,s$ (thus, $\kappa< -\beta$). Here, we define $f_P:=|\breve\dF(q)|+|\breve\dF(\beta+\alpha_+)|+|\breve\dF(\beta+\alpha_-)|$. Thus, by~\eqref{app:eq:funct-eq-PQ} we obtain
\begin{align*}
	|\breve\dF(s)|< \epsilon f_P + \epsilon |\breve\dF(s+2\beta)|~, \quad \mbox{if}\quad  |\Img\,s|>M~.
\end{align*}
By successively applying this inequality $N_0$ times ($N_0\ge 1$), we have 
\begin{align*}
	|\breve\dF(s)|< \epsilon (1+\epsilon+\ldots +\epsilon^{N_0-1})f_P+\epsilon^{N_0}|\breve\dF(s+2N_0\beta)|=\epsilon\frac{1-\epsilon^{N_0}}{1-\epsilon}f_P+\epsilon^{N_0}|\breve\dF(s+2N_0\beta)|~.
\end{align*}
By choosing an $N_0$ such that $2N_0\beta > -\beta+|\kappa|>0$, we have $\Rel(s+2N_0\beta)>-\beta$; thus, the term $|\breve\dF(s+2N_0\beta)|$ tends to $0$ uniformly as $|\Img\,s|\to +\infty$. It follows that $\breve\dF(s)$ approaches $0$ uniformly in the $s$-plane as $|\Img\,s|\to +\infty$. Similarly, one can show that $\breve\dF(s)\to 0$ uniformly as $|s+\beta|\to +\infty$ with $\tfrac{\pi}{2}\le  |\Arg(s+\beta)|<\pi-\delta_1$ for arbitrarily small but fixed $\delta_1$ ($\delta_1>0$). In this case, $N_0$  depends on $|s+\beta|$, since we impose $2N_0\beta>-\Rel\,(s+\beta)=|s+\beta| |\cos\theta|$ with $\theta=\Arg(s+\beta)$. The above study can directly be extended to complex $\omega$ (for $0<q^2<|\omega|< q\beta$), when the poles $s_n^\sigma$. In this case, the poles  $s_n^\pm$ and $\slashed{s}_n$ that have nonzero imaginary parts  approach the negative real axis asymptotically in the limit $n\to +\infty$, for fixed parameters $\omega$, $q$, and $\beta$.

\section{On proving Lemma~1: Limit of $g(s)$ as $s\to -\infty$}
\label{app:sec:asympt-g}
In this appendix, we show that~\eqref{eq:funct-eq-g} implies $|g(s)|=\mathcal O(1/|s|)$ as $s\to -\infty$ (in the negative real axis) for the entire function $g(s)$. Equation~\eqref{eq:funct-eq-g} is written in the canonical form
\begin{align}\label{app:funct-eq-g}
	g(s)= \sum_{l=0}^{\lo} \frac{c_l}{s+\xi_l}\{g(s+2\beta)-g(-\xi_l+2\beta)\}
\end{align}
for fixed parameters $c_l$, $\xi_l$, $\beta$ and $\lo$, where $\lo$ is a positive integer ($\lo\ge 1$) and $\beta>0$. For definiteness, and without loss of generality, we assume that $\xi_l>0$ for all $l=0,\,1,\,\ldots,\, \lo$ ($\xi_i\neq \xi_j$ if $i\neq j$).

Let us first consider the sequence $\{-2\beta n\}_{n\in \mathbb{N}}$ and define $g_n:=g(-2\beta n)$. We will show that $g_n=\mathcal O(1/n)$ as $n\to +\infty$.

To prove this property, by~\eqref{app:funct-eq-g} we study the difference scheme
\begin{align}\label{app:eq:diff-sch}
g_n=\sum_{l=0}^{\lo}\frac{\bar c_l}{n-\bar\xi_l}\{g_{n-1}-g(-\xi_l+2\beta)\}~,\quad n\ge 1~;\quad \bar c_l:=-\frac{c_l}{2\beta}~,\ \bar\xi_l:=\frac{\xi_l}{2\beta}>0~.	
\end{align}
Note that
\begin{align*}
	0<\frac{1}{n-\bar \xi_l}\le \frac{2}{n}\ \ \forall\, l\in\{0,\ldots,\lo\}\quad \mbox{if}\ \ n\ge \no=[2\max_l(\{\bar\xi_l\}_{l=0}^{\lo})]+1~,
\end{align*}
where $[x]$ is defined as the largest possible integer that is less than or equal to the real number $x$. Hence, scheme~\eqref{app:eq:diff-sch} entails 
\begin{align*}
	|g_n|\le \frac{C}{n} |g_{n-1}|+\frac{L}{n}~,\quad n\ge \no~;\quad C:=2\sum_{l=0}^{\lo}|\bar c_l|~,\quad L:=2\sum_{l=0}^{\lo}|\bar c_l|\,|g(-\xi_l+2\beta)|~.
\end{align*}
By applying the above estimate for $|g_n|$ successively, we obtain the inequality
\begin{align}\label{app:eq:sch-ineq}
|g_n|\le \frac{C^{n-\no}\no!}{n!}|g_{\no}| +\frac{L}{n}\left\{1+\sum_{j=1}^{n-\no-1}\frac{C^j}{\prod_{\varkappa=1}^j(n-j-1+\varkappa)} \right\}~,\quad n\ge \no+1~.
\end{align}
In the special case with $n=n_0+1$, the partial sum $\sum_{j=1}^{n-\no-1}(\cdot)$ is defined to be equal to $0$. 

Let us consider $n\ge \no+2$. Because $1\le j\le n-n_0-1< n-1$, we assert that
\begin{align*}
	1+\sum_{j=1}^{n-\no-1}\frac{C^j}{\prod_{\varkappa=1}^j(n-j-1+\varkappa)}< 1+\sum_{j=1}^{n-\no-1}\frac{C^j}{j!}< e^C\qquad \forall\,n\ge n_0+2~. 
\end{align*}
Thus, the increasing-with-$n$ sequence of partial sums in~\eqref{app:eq:sch-ineq} has a fixed upper bound. We infer that as $n\to +\infty$ the ensuing infinite series in~\eqref{app:eq:sch-ineq} is convergent. In this limit, the first term on the right-hand side of~\eqref{app:eq:sch-ineq} is subdominant and can be neglected. More precisely, for any $\epsilon>0$ there exists an integer $N_0\ge 1$ such that 
\begin{align*}
	|g_n|\le (\epsilon+1)\frac{Le^C}{n}~,\quad \mbox{all}\ n\ge N_0\quad  (n\in\mathbb{N})~.
\end{align*}
This statement proves the desired result for the sequence $\{g(-2\beta n)\}_{n\in\mathbb{N}}$.

Similarly, we can use~\eqref{eq:funct-eq-g} to study the asymptotic behavior of $g_n=g(c-2\beta n)$ as $n\to +\infty$ for fixed $c$ with $0<c< 2\beta$. By repeating the above steps with appropriately shifted parameters $\bar\xi_l$ ($\bar\xi_l>0$), we can show that $|g_n|=\mathcal O(1/n)$ as $n\to +\infty$. 

The above procedure and results can appropriately be extended, via minor modifications, to complex parameters $\xi_\ell$ with $\Rel\,\xi_\ell>0$. For  this case, we note that
\begin{align*}
	0<\frac{1}{|n-\bar \xi_l|}\le \frac{1}{n-\Rel\,\bar \xi_l}\le \frac{2}{n}\ \ \forall\, l\in\{0,\ldots,\lo\}\quad \mbox{if}\ \ n\ge \no=[2\max_l(\{\Rel\,\bar\xi_l\}_{l=0}^{\lo})]+1~.
\end{align*}

\section{On the behavior of $\Lambda(\omega, q)$ as $\omega\to 2\beta q$ ($0<q<\beta$)}
\label{app:sec:Lambda-lim-q}
In this appendix, we verify that $\Lambda(\omega, q)=\mathcal O(1)$ as $\omega\to 2\beta q$. To this end, we invoke the asymptotic relations for $\mA^\mu_\nu$ listed in Sec.~\ref{subsec:soln-rmks}.

Define $\epsilon:=\omega-2\beta q$ with $|\epsilon|\ll 1$. By using the asymptotic relations listed in Sec.~\ref{subsec:soln-rmks}, we start with the determinant
\begin{align*}
	&\Lambda(\omega, q)\sim 
	\begin{vmatrix}
	\displaystyle{a_1\epsilon^{-1}+b_1-1}\, & \displaystyle{a_2\epsilon^{-1}+b_2}\, & b_3\, & \displaystyle{-a_1\epsilon^{-1}+b_4}\, & b_5\, & \displaystyle{-a_2\epsilon^{-1}+b_6} \\
	\displaystyle{c_1\epsilon^{-1}+d_1}\, & \displaystyle{c_2\epsilon^{-1}+d_2}-1\, & d_3\, & \displaystyle{-c_1\epsilon^{-1}+d_4}\, & d_5\, & \displaystyle{-c_2\epsilon^{-1}+d_6} \\
	\displaystyle{e_1\epsilon^{-1}+f_1}\, & \displaystyle{e_2\epsilon^{-1}+f_2}\, & f_3-1\, & \displaystyle{-e_1\epsilon^{-1}+f_4}\, & f_5\, & \displaystyle{-e_2\epsilon^{-1}+f_6} \\
	\displaystyle{a_1\epsilon^{-1}+h_1}\, & \displaystyle{a_2\epsilon^{-1}+h_2}\, & b_3\, & \displaystyle{-a_1\epsilon^{-1}+h_4}-1\, & b_5\, & \displaystyle{-a_2\epsilon^{-1}+h_6} \\
	\displaystyle{k_1\epsilon^{-1}+l_1}\, & \displaystyle{k_2\epsilon^{-1}+l_2}\, & l_3\, & \displaystyle{-k_1\epsilon^{-1}+l_4}\, & l_5-1\, & \displaystyle{-k_2\epsilon^{-1}+l_6} \\
	\displaystyle{c_1\epsilon^{-1}+\eta_1}\, & \displaystyle{c_2\epsilon^{-1}+\eta_2}\, & d_3\, & \displaystyle{-c_1\epsilon^{-1}+\eta_4}\, & d_5\, & \displaystyle{-c_2\epsilon^{-1}+\eta_6}-1
	\end{vmatrix}~. 
\end{align*}
Here, $a_i$, $b_i$, $c_i$, $d_i$, $e_i$, $f_i$, $h_i$, $k_i$, $l_i$ and $\eta_i$ ($i\in\{1, 2, 3, 4, 5, 6\}$) are coefficients of Laurent-type expansions for the matrix elements $\mA^\mu_\nu$ at $\epsilon=0$ and $O(1)$ functions of $(\beta, q)$ if $0<q<\beta$, as discussed in Sec.~\ref{subsec:soln-rmks}.  The following relations also hold:
\begin{align*}
	b_1+b_4=h_1+h_4~,\quad  b_2+b_6=h_2+h_6~,\quad d_1+d_4=\eta_1+\eta_4~,\quad d_2+d_6=\eta_2+\eta_6~.
\end{align*}   
By basic invariance properties of determinants, we compute
\begin{align*}
	&\Lambda(\omega, q)\sim 
	\begin{vmatrix}
	\displaystyle{b_1-h_1-1}\, & \displaystyle{b_2-h_2}\, & 0\, & \displaystyle{b_1-h_1+b_4-h_4}\, & 0\, & \displaystyle{b_6-h_6+b_2-h_2} \\
	\displaystyle{d_1-\eta_1}\, & \displaystyle{d_2-\eta_2}-1\, & 0\, & \displaystyle{d_4-\eta_4+d_1-\eta_1}\, & 0\, & \displaystyle{d_6-\eta_6+d_2-\eta_2} \\
	\displaystyle{e_1\epsilon^{-1}+f_1}\, & \displaystyle{e_2\epsilon^{-1}+f_2}\, & f_3-1\, & \displaystyle{f_1+f_4}\, & f_5\, & \displaystyle{f_2+f_6} \\
	\displaystyle{a_1\epsilon^{-1}+h_1}\, & \displaystyle{a_2\epsilon^{-1}+h_2}\, & b_3\, & \displaystyle{h_1+h_4}-1\, & b_5\, & \displaystyle{h_2+h_6} \\
	\displaystyle{k_1\epsilon^{-1}+l_1}\, & \displaystyle{k_2\epsilon^{-1}+l_2}\, & l_3\, & \displaystyle{l_1+l_4}\, & l_5-1\, & \displaystyle{l_2+l_6} \\
	\displaystyle{c_1\epsilon^{-1}+\eta_1}\, & \displaystyle{c_2\epsilon^{-1}+\eta_2}\, & d_3\, & \displaystyle{\eta_1+\eta_4}\, & d_5\, & \displaystyle{\eta_2+\eta_6}-1
	\end{vmatrix}\\
&=[(b_1-h_1-1)(d_2-\eta_2-1)-(b_2-h_2)(d_1-\eta_1)]
	\begin{vmatrix}
		f_3-1 & f_1+f_4 & f_5 & f_2+f_6 \\
		b_3 & h_1+h_4-1 & b_5 & h_2+h_6 \\
		l_3 & l_1+l_4 & l_5-1 & l_2+l_6 \\
		d_3 & \eta_1+\eta_4 & d_5 & \eta_2+\eta_6-1
	\end{vmatrix}~,
\end{align*}
which is $\mathcal O(1)$.

\bibliography{Biblio}

\end{document}